\documentclass[sigconf,nonacm,10pt]{acmart}

\AtBeginDocument{%
  \providecommand\BibTeX{{%
    \normalfont B\kern-0.5em{\scshape i\kern-0.25em b}\kern-0.8em\TeX}}}

\setcopyright{acmcopyright}
\copyrightyear{2018}
\acmYear{2018}
\acmDOI{10.1145/1122445.1122456}

\acmConference[Woodstock '18]{Woodstock '18: ACM Symposium on Neural
  Gaze Detection}{June 03--05, 2018}{Woodstock, NY}
\acmBooktitle{Woodstock '18: ACM Symposium on Neural Gaze Detection,
  June 03--05, 2018, Woodstock, NY}
\acmPrice{15.00}
\acmISBN{978-1-4503-9999-9/18/06}

\usepackage{enumerate}
\usepackage{amsfonts}
\usepackage{booktabs}
\usepackage{siunitx}
\usepackage{multirow}
\usepackage{graphicx}
\usepackage{listings}
\usepackage{commath}
\usepackage{epstopdf}
\usepackage{url}
\usepackage{mathtools}
\usepackage{alltt}
\usepackage{mathrsfs}
\usepackage{wrapfig}

\usepackage[normalem]{ulem}
\usepackage{etoolbox}
\usepackage{stmaryrd}
\usepackage{amsmath,amssymb}
\usepackage[noend]{algorithmic}

\usepackage{mathtools}
\usepackage{etoolbox}
\usepackage{dsfont}
\usepackage{color}
\usepackage{colortbl}
\usepackage{multirow}
\usepackage[scaled]{helvet}

\usepackage{graphicx,fancyvrb}
\usepackage[colorinlistoftodos]{todonotes}
\usepackage[noend]{algorithmic}
\usepackage{amsmath,amssymb}
\usepackage{booktabs}
\usepackage{capt-of}
\usepackage{xcolor}
\usepackage{multirow}

\usepackage{caption}
\usepackage[FIGTOPCAP]{subfigure}
\usepackage{float}
\usepackage{subfloat}

\usepackage{array}
\usepackage{arydshln}
\usepackage{todonotes}

\definecolor{mygreen}{rgb}{0,0.6,0}
\definecolor{mygray}{rgb}{0.5,0.5,0.5}
\definecolor{mymauve}{rgb}{0.58,0,0.82}

\usepackage{listings}
\lstnewenvironment{VerbatimText}[1][]{
    
    \lstset{fancyvrb=true,basicstyle=\footnotesize,captionpos=b,xleftmargin=2em,#1}
}{}

\newcommand{\E}{{\tt \mathbb{E}}}

\newcommand{\rred}[1]{\textcolor{red}{#1}}

\newcommand{\ignore}[1]{}

\newcommand{\batya}[1]{{\texttt{\color{blue} Batya: [{#1}]}}}
\newcommand{\dan}[1]{{\texttt{\color{red} Dan: [{#1}]}}}

\newcommand{\revisionrevfive}[1]{{\color{blue}{#1}}}

\newcommand\mvd{\twoheadrightarrow}

\newcommand{\fd}{\rightarrow}

\newcommand{\real}{{\mathbb{R}}}

\newcommand{\sql}[1]{\texttt{#1}}
\newcommand{\select}{\texttt{Select}}
\newcommand{\where}{\texttt{Where}}

\newcommand{\groupby}{\texttt{Group By}}
\newcommand{\having}{\texttt{Having}}
\newcommand{\from}{\texttt{From}}

\newcommand{\as}{\texttt{as}}

\newcommand*{\rom}[1]{\expandafter\@slowromancap\romannumeral #1@}
\newcommand{\RNum}[1]{\uppercase\expandafter{\romannumeral #1\relax}}
\newtheorem{defn}{Definition}[section]

\newtheorem{problem}[defn]{Problem}
\newtheorem{prop}[defn]{Proposition}

\newtheorem{prunRule}[defn]{Prunning Rule}
\newcommand{\proj}[1]{{\Pi}_{#1}}
\newcommand{\sel}[1]{{\sigma}}

\newcommand{\cut}[1]{}
\newcommand{\eat}[1]{}

\newcommand{\defeq}{\stackrel{\text{def}}{=}}
\newcommand{\setof}[2]{\{{#1}\mid{#2}\}}        % Set (as in \setof{x}{x>0}).

\def\set#1{\mathord{\{#1\}}}

\def\MVD{\sigma}

\def\eqdef{\mathrel{\stackrel{\textsf{\tiny def}}{=}}}

\def\J{\mathcal{J}}

\def\D{\mathcal{D}}
\def\e#1{\emph{#1}}
\newenvironment{citedtheorem}[1]
{\begin{theorem}{\it\e{(#1)}}\,\,}
	{\end{theorem}}

\newenvironment{repeatresult}[2]
{\vskip0.5em\par\textsc{#1} #2.\em}
{\vskip1em}

\newenvironment{repproposition}[1]{\begin{repeatresult}{Proposition}{#1}}{\end{repeatresult}}
\newenvironment{reptheorem}[1]{\begin{repeatresult}{Theorem}{#1}}{\end{repeatresult}}
\newenvironment{replemma}[1]{\begin{repeatresult}{Lemma}{#1}}{\end{repeatresult}}

\def\appendix{\par
	\section*{APPENDIX}
	\setcounter{section}{0}
	\setcounter{subsection}{0}
	\def\thesection{\Alph{section}} }

\def\join{\bowtie}

\def\eqdef{\mathrel{\stackrel{\textsf{\tiny def}}{=}}}

\def\J{\mathcal{J}}

\def\e#1{\emph{#1}}

\newcommand{\algname}[1]{{\sf #1}}
\def\myrulewidth{3.25in}
\def\therule{\rule{\myrulewidth}{0.2pt}}

\def\myrulewidthwide{4in}
\def\therulewide{\rule{\myrulewidthwide}{0.2pt}}

\newenvironment{algseries}[2]
{\centering\begin{figure}[#1]\begin{center}\def\thecaption{\caption{#2}}
			\begin{tabular}{p{\myrulewidth}}\therule\end{tabular}%\vskip0.1em
		}
		{\thecaption\end{center}\end{figure}}

\newenvironment{algserieswide}[2]
{\centering\begin{figure}[#1]\begin{center}\def\thecaption{\caption{#2}}
			\begin{tabular}{p{\myrulewidthwide}}\therulewide\end{tabular}\vskip0.2em}
		{\thecaption\end{center}\end{figure}}

\newenvironment{insidealg}[2]
{\normalsize
	\begin{insidecode}{#1}{#2}{Algorithm}}
	{\end{insidecode}
}

\newenvironment{insidealgwide}[2]
{\normalsize\begin{insidecodewide}{#1}{#2}{Algorithm}}
	{\end{insidecodewide}}

\newenvironment{insidesub}[2]
{\begin{insidecode}{#1}{#2}{Subroutine}}
	{\end{insidecode}}

\newenvironment{insidecode}[3]
{
	\begin{tabular}{p{\myrulewidth}}
		\multicolumn{1}{c}{\rule{0mm}{3mm}{\bf #3} $\algname{#1}(\mbox{#2})$\vspace{-0.6em}}\\
		\therule\vskip-0.8em\therule
		\vspace{-1em}
		\begin{algorithmic}[1]}
		{\end{algorithmic}
		\vskip-0.4em\therule
\end{tabular}}

\newenvironment{insidecodewide}[3]
{
	\begin{tabular}{p{\myrulewidthwide}}
		\multicolumn{1}{c}{\rule{0mm}{3mm}{\bf #3} $\algname{#1}(\mbox{#2})$\vspace{-0.6em}}\\
		\therulewide\vskip-0.8em\therulewide
		\vspace{-1em}
		\begin{algorithmic}[1]}
		{\end{algorithmic}
		\vskip-0.3em\therulewide
\end{tabular}}

\newcommand{\Q}{{\mathcal{Q}}}
\newcommand{\T}{{\mathcal{T}}}
\newcommand{\Pm}{{\mathcal{P}}}

\newcommand{\comp}[1]{\overline{#1}}
\newcommand{\BD}{{\tt \mathrm{BD}}}
\def\minsep{\mathit{MinSep}}

\newcommand{\JD}{\textsc{JD}}
\newcommand{\AJD}{\textsc{AJD}}

\newcommand{\fullMVDs}{\textsc{FullMVD}}

\newcommand{\schema}{\mathbf{S}}

\newcommand{\relation}{R}

\newcommand{\jointreeMapFunction}{\chi}

\newcommand{\sm}{{\setminus}}
\def\P{\mathsf{P}}
\def\M{\mathcal{M}}
\def\minsep{\mathrm{\textsc{MinSep}}}
\newcommand{\entropyAlg}{\algname{getEntropy_R}}
\newcommand{\mergeFunc}{\algname{merge}}
\newcommand{\push}{\algname{push}}
\newcommand{\pop}{\algname{pop}}
\def\Nbr{\mathrm{Nbr}}
\newcommand{\bS}{\mathbf{S}}
\newcommand{\key}{\mathrm{key}}
\newcommand{\components}{\mathrm{dep}}
\newcommand{\edges}{\texttt{edges}}
\newcommand{\nodes}{\texttt{nodes}}
\newcommand{\parent}{\texttt{parent}}
\newcommand{\schemas}{\texttt{schemes}}
\newcommand{\numRelations}{\texttt{\#relations}}
\newcommand{\width}{\texttt{width}}
\newcommand{\intWidth}{\texttt{intWidth}}
\newcommand{\CntTbl}{\mathrm{CNT}}
\newcommand{\TidTbl}{\mathrm{TID}}
\newcommand{\valuation}{\texttt{val}}
\newcommand{\cnt}{\texttt{cnt}}
\newcommand{\tid}{\texttt{tid}}
\newcommand{\MVDAlg}{\algname{MVDMiner}}
\newcommand{\ASAlg}{\algname{ASMiner}}
\newcommand{\minSepAlg}{\algname{MineMinSeps}}

\newcommand{\system}{\algname{Maimon}}

\def\MVD{\mathrm{MVD}}
\def\bC{\textbf{C}}
\setcopyright{none}
\renewcommand\footnotetextcopyrightpermission[1]{}
\newcommand\Mark[1]{\textsuperscript#1}

\begin{document}

\title{Mining Approximate Acyclic Schemes from Relations}

\author{Batya Kenig$^1$ \ \  \ \ \  Pranay Mundra$^2$ \ \  \ \ \     Guna Prasad$^1$  \ \  \ \ \  Babak Salimi$^1$ \ \  \ \ \  Dan Suciu$^1$}
\affiliation{ \begin{tabular}{*{2}{>{\centering}p{.42\textwidth}}}
		\Mark{1} \large Computer Science and Engineering & \Mark{2} \large Department of Mathematics  \tabularnewline
		\large University of Washington  &  \large University of Washington
		\tabularnewline
		\large batyak, guna, bsalimi, suciu@cs.washington.edu&  \large pranay99@uw.edu
	\end{tabular}
}
\eat{
\author{Batya Kenig}
\affiliation{%
	\institution{University of Washington}
}
\email{batyak@cs.washington.edu}

\author{Pranay Mundra}
\affiliation{%
	\institution{University of Washington}
}
\email{pranay99@uw.edu}

\author{Guna Prasad}
\affiliation{%
	\institution{University of Washington}
}
\email{guna@cs.washington.edu}

\author{Babak Salimi}
\affiliation{%
	\institution{University of Washington}
}
\email{bsalimi@cs.washington.edu}

\author{Dan Suciu}
\affiliation{%
	\institution{University of Washington}
}
\email{suciu@cs.washington.edu}
}
\begin{abstract}
  Acyclic schemes have numerous applications in databases and in machine
learning, such as improved design, more efficient storage, and
increased performance for queries and machine learning algorithms.
Multivalued dependencies (MVDs) are the building blocks of acyclic
schemes.  The discovery from data of both MVDs and acyclic schemes is
more challenging than other forms of data dependencies, such as
Functional Dependencies, because these dependencies do not hold on
subsets of data, and because they are very sensitive to noise in the
data; for example a single wrong or missing tuple may invalidate the
schema.  In this paper we present \system, a system for discovering
\e{approximate} acyclic schemes and MVDs from data.  We give a
principled definition of approximation, by using notions from
information theory, then describe the two components of
\system: mining for approximate MVDs, then reconstructing acyclic
schemes from approximate MVDs.  We conduct an experimental evaluation
of \system\ on 20 real-world datasets, and show that it can scale up
to 1M rows, and up to 30 columns.

% Data dependencies make up the building blocks of database schema
% design and are used to enforce correctness as well as reduce
% redundancy.  Multivalued dependencies (MVDs) are the most general data
% dependencies and a relation decomposes losslessly if and only if the
% required MVD holds. Most research to date focus on the task of mining
% functional dependencies (FDs) and unique column combinations
% (UCCs). However, FDs and UCCs are a sufficient but not a necessary
% condition for lossless decomposition, making MVDs more useful for
% schema discovery and database normalization.  Data dependency theory
% makes the assumption that the dependencies hold exactly in the
% data. Real world dependencies, in contrast, are usually only
% approximate due to data exceptions, ambiguity, or data errors.  In
% this work we propose an algorithm that generates a set of acyclic
% schemas that fit the data to the largest degree.  The algorithm
% proceeds by mining a set of MVDs that approximately hold in the data,
% and then enumerates all possible combinations of the approximate MVDs
% to acyclic schemas.  We conduct extensive empirical evaluation of the
% proposed algorithm showing that the resulting schemas incur a small
% loss, depending on the accuracy threshold defined, and scale to large
% datasets.

\end{abstract}

\maketitle
\section{Introduction}

Acyclic schemes have numerous applications in databases and in machine
learning.  Originally introduced by
Beeri~\cite{DBLP:conf/stoc/BeeriFMMUY81}, they have lead to Yannakakis
celebrated linear time query evaluation
algorithms~\cite{Yannakakis:1981:AAD:1286831.1286840}, and are used
widely today in database
design~\cite{Fagin:1977:MDN:320557.320571,DBLP:journals/is/LeveneL03},
to speed up query evaluation with multiple
aggregates~\cite{DBLP:conf/pods/KhamisNR16}, and to speed up machine
learning applications such as ridge linear regression, classification
trees, and regression
trees~\cite{DBLP:conf/sigmod/SchleichOC16,DBLP:conf/sigmod/Khamis0NOS18,DBLP:conf/sigmod/SchleichOK0N19}.
When considering which types of schemes to fit the data, acyclic
schemes are the natural choice due to their many desirable
properties~\cite{Beeri:1983:DAD:2402.322389}.  In this paper we study
the following discovery problem: given a database consisting of a
single relation, generate a set of acyclic schemes that fit the data
to a large extent.  For a simple illustration, consider the database
shown on the left of Figure~\ref{fig:acyclic:schema}.  It can be
decomposed into an acyclic schema with four relations, shown on the
right.

\small{
\begin{figure*}[t]
\hspace{-1cm}
\mbox{
\begin{minipage}{0.9\linewidth}
  \begin{tabular}[c]{|c|c|c|c|c|c|l} 
    \multicolumn{6}{l}{$\relation$:} \\ \cline{1-6}
    $A$ & $B$ & $C$ & $D$ & $E$ & $F$ & \\ \cline{1-6}
    $a_1$ & $b_1$ & $c_1$ & $d_1$ & $e_1$ & $f_1$ & \tiny 1/4\\
    $a_2$ & $b_2$ & $c_1$ & $d_1$ & $e_2$ & $f_2$ & \tiny 1/4\\
    $a_2$ & $b_2$ & $c_2$ & $d_2$ & $e_3$ & $f_2$ & \tiny 1/4\\
    $a_1$ & $b_2$ & $c_1$ & $d_2$ & $e_3$ & $f_1$ & \tiny 1/4\\ 
    \rred{$a_1$} & \rred{$b_2$} & \rred{$c_1$} & \rred{$d_2$} & \rred{$e_2$} & \rred{$f_1$}&\\
    \cline{1-6}
  \end{tabular}
$=$
  \begin{tabular}[c]{|c|c|c|l} \cline{1-3}
    $A$ & $B$ & $D$ & \\ \cline{1-3}
    $a_1$ & $b_1$ & $d_1$ & \tiny 1/4\\
    $a_2$ & $b_2$ & $d_1$ & \tiny 1/4\\
    $a_2$ & $b_2$ & $d_2$ & \tiny 1/4\\
    $a_1$ & $b_2$ & $d_2$ & \tiny 1/4\\ \cline{1-3}
  \end{tabular}
$\Join$
  \begin{tabular}[c]{|c|c|c|l} \cline{1-3}
    $A$ & $C$ & $D$ & \\ \cline{1-3}
    $a_1$ & $c_1$ & $d_1$ & \tiny 1/4\\
    $a_2$ & $c_1$ & $d_1$ & \tiny 1/4\\
    $a_2$ & $c_2$ & $d_2$ & \tiny 1/4\\
    $a_1$ & $c_1$ & $d_2$ & \tiny 1/4\\ \cline{1-3}
  \end{tabular}
$\Join$
  \begin{tabular}[c]{|c|c|c|l} \cline{1-3}
    $B$ & $D$ & $E$ & \\ \cline{1-3}
    $b_1$ & $d_1$ & $e_1$ & \tiny 1/4\\
    $b_2$ & $d_1$ & $e_2$ & \tiny 1/4\\
    $b_2$ & $d_2$ & $e_3$ & \tiny 1/2\\ 
    \rred{$b_2$} & \rred{$d_2$} & \rred{$e_2$}\\
    \cline{1-3}
  \end{tabular}
$\Join$
  \begin{tabular}[c]{|c|c|l} \cline{1-2}
    $A$   & $F$ & \\ \cline{1-2}
    $a_1$ & $f_1$ & \tiny 1/2\\
    $a_2$ & $f_2$ & \tiny 1/2\\ \cline{1-2}
  \end{tabular}
\caption{A relation $R$  and it's decomposition into an acylic schema}
 \label{fig:acyclic:schema}
\end{minipage}
}
\hspace{-1cm}
\begin{minipage}{0.18\linewidth}\centering
	\includegraphics[width=0.75\textwidth]{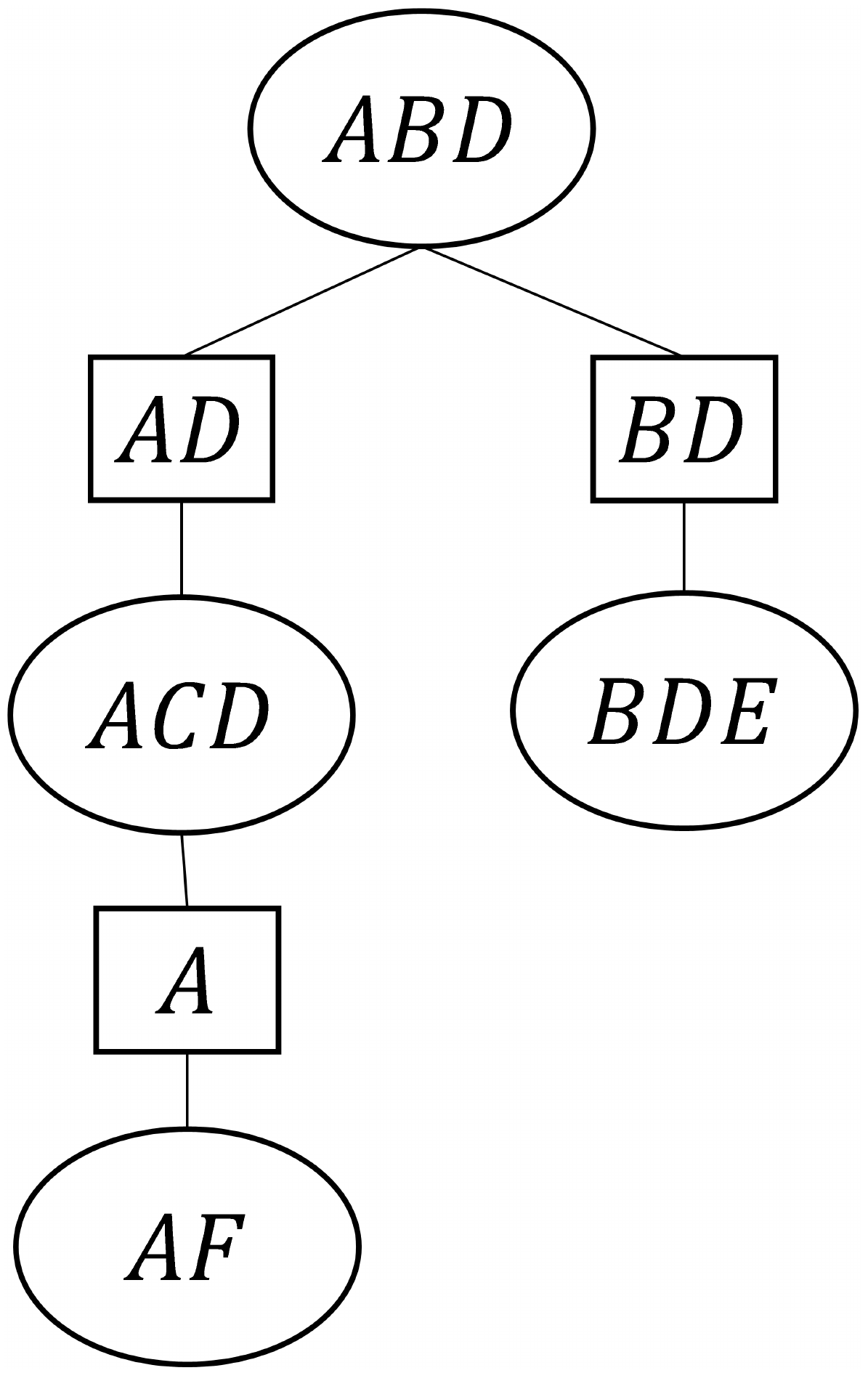}
	\caption{Join Tree}
	\label{fig:TD}
\end{minipage}
\end{figure*}
}%

\normalsize{
The building blocks of an acyclic schema are Multivalued Dependencies,
MVDs.  Every acyclic schema can be fully specified by the set of MVDs
that it implies, which we call its \e{support}.  Therefore, when
mining acyclic schemes, the first step is to mine the MVDs satisfied
by the data.  MVDs were first introduced by
Fagin~\cite{Fagin:1977:MDN:320557.320571}, which used them to
introduce the 4th normal form, a generalization of the Boyce-Codd
normal form (BCNF)~\cite{Codd1971FurtherNO}. They were studied
extensively in the database literature
~\cite{DBLP:conf/sigmod/BeeriFH77,Beeri:1980:MPF:320613.320614,DBLP:journals/jacm/Fagin82,DBLP:journals/tcs/Lakshmanan88},
have been proven to be equivalent to Saturated Conditional
Independence \eat{constraints} in graphical models~\cite{GeigerPearl1993}, and have recently been used as part of a data repairing solution  to enforce fairness of ML systems~\cite{DBLP:conf/sigmod/SalimiRHS19,DBLP:journals/corr/abs-1908-07924}.
The methods used to synthesize an acyclic schema from a
set of
MVDs are well known~\cite{Fagin:1977:MDN:320557.320571,DBLP:journals/tods/Bernstein76,Beeri:1979:CPR:320064.320066,Lien:1981:HSR:319540.319546}.
However, despite their importance, there is little research on the
{\em discovery} of MVDs from
data~\cite{doi:10.2200/S00878ED1V01Y201810DTM052}.

Work most closely related to the discovery of MVDs has been on
discovering Functional Dependencies (FDs) and Unique Column
Combinations
(UCCs)~\cite{DBLP:journals/tcs/KivinenM95,DBLP:journals/cj/HuhtalaKPT99,DBLP:conf/dawak/WyssGR01,DBLP:journals/tkde/LiuLLC12,DBLP:conf/sigmod/PapenbrockN16,DBLP:conf/cikm/BleifussBFRW0PN16,DBLP:journals/pvldb/0001N18}.
These are special cases of MVDs, but MVDs are more general.
Discovering all FDs and all UCCs is insufficient for discovering
acyclic schemes.  The only work that addressed the discovery problem
for MVDs is by Savnik and Flach~\cite{DBLP:journals/ida/SavnikF00} and
a master thesis by Draeger~\cite{draeger2016}, and none of them
address the more challenging task of discovering acyclic schemes.

There are two major challenges that make the discovery of MVDs and
acyclic schemes, much harder than that of FDs and UCCs.  First, they
don't hold on subsets of the data.  If a relation satisfies an FD, or
a UCC, then every subset also satisfies the FD, or UCC, and this is
exploited by many discovery algorithms, e.g
FastFD~\cite{DBLP:conf/dawak/WyssGR01} mines FDs in all subsets of
size 2, while HyFD~\cite{DBLP:conf/sigmod/PapenbrockN16} mines FDs in
a small subset extracted from the data.  This property fails for MVDs,
preventing us from considering subsets of the data.  Second, MVDs and
acyclic schemes are much more sensitive to data errors than FDs and
UCCs.  Even a single missing tuple may invalidate an MVD or schema.
Real-world data often has important dependencies that do not hold
exactly, but, if discovered, are very useful for a variety of
applications.  For that reason, in this paper we study the problem of
discovering {\em approximate} MVDs and consequently, {\em approximate}
acyclic schemes.

We present \system\footnote{\eat{Moses ben Maimon, commonly
  known as Maimonides was a medieval Sephardic Jewish philosopher,
  astronomer and physician, in short a \e{discoverer}.}  \system\ 
  stands for \underline{M}ultivalued \underline{A}pproximate
  \underline{I}nference \underline{M}ining and
  \underline{NO}rmalization.}, the first system for discovering
approximate MVDs and acyclic schemes in the data.  We introduce a
principled notion of approximation, based on information theory, and
develop the necessary theory for reasoning about approximate MVDs and
schemes.  We then describe algorithms for mining MVDs and
schemes, and evaluate their scalability on real-world datasets of up to 1M
rows, and 30 attributes.  By allowing approximations, \system\
finds more interesting schemes without incurring too high a loss (i.e., spurious
tuples).  We make several contributions.

Our first contribution is to introduce a principled definition of
approximation, and study its properties.  Kivinen and
Mannila~\cite{DBLP:journals/tcs/KivinenM95} give three definitions
of approximate functional dependencies, and Kruese and Naumann use one
of them in their approximate FDs and UCCs discovery
algorithm~\cite{DBLP:journals/pvldb/0001N18}.  We propose an
alternative metric of approximation, based on information theory.
Each MVD or acyclic schema is associated with an information theoretic
expression, and its value represents the
degree of approximation.  Our definition builds on early work by
Lee~\cite{DBLP:journals/tse/Lee87a}.

Second, we propose novel algorithms for mining approximate MVDs and
approximate acyclic schemes.  For mining MVDs, our theoretical results
prove that we do not need to discover \e{all} approximate MVDs, but
only the so-called \e{full} MVDs with \e{minimal} separators.  Our
algorithm builds on previous results by
Gunopulos et al.~\cite{DBLP:journals/tods/GunopulosKMSTS03} for discovering the most \e{specific} sentences in the data that meet a certain criterion (e.g., maximal sets of items whose frequency in the data is above a given threshold).
\eat{
starts by finding all minimal separators,
by leveraging a result by
Gunopulos et al.~\cite{DBLP:journals/tods/GunopulosKMSTS03} and a
hypergraph transversal enumeration algorithm by Fredman and
Khachiyan~\cite{FREDMAN1996618}.  For the next step we prove, rather
surprisingly, that, in the approximate case, one separator may have
more than one full MVD: this differs from exact MVDs where each
separator has exactly one ``best'' full MVD (a property used by
Draeger~\cite{draeger2016} in his MVD discovery algorithm), making the
discovery of approximate MVDs more challenging than that of exact
MVDs.  
}
Following the discovery of the MVDs that hold in the data, we turn to the task of enumerating the acyclic schemes that can be synthesized from the set of discovered MVDs.
Our algorithm is based on an approach for efficiently enumerating the maximal
independent sets of a graph~\cite{DBLP:journals/jcss/CohenKS08,
  DBLP:journals/ipl/JohnsonP88}, which has also been applied to the problem of
enumerating tree decompositions~\cite{DBLP:conf/pods/CarmeliKK17}.

Third, we evaluate \system\ on 20 real-world datasets that are part of the Metanome project that provides a repository of benchmarks for a variety of data profiling tasks that include the discovery of data dependencies. The datasets chosen for evaluation have been used in a large body of work on mining exact and approximate FDs~\cite{Papenbrock:2015:DPM:2824032.2824086,DBLP:conf/cikm/BleifussBFRW0PN16,draeger2016,DBLP:journals/pvldb/0001N18,DBLP:journals/tkde/LiuLLC12,DBLP:conf/sigmod/PapenbrockN16}.
\eat{, and were able to}
We show that \system~
scales up to 1M rows, and up to 30 columns.
We empirically show that the loss entailed by the generated acyclic schemes (i.e., number of spurious tuples), monotonically depends on and the information theoretic measure of approximation we develop herein. We also show that a larger degree of approximation enables the discovery of schemes that exhibit a larger degree of decomposition, that leads to significant savings in storage. These schemes generally have more relations, and the \e{width} of the schema (i.e., relation with the largest number of attributes), is smaller.
\eat{
We validate our hypothesis that, by
allowing approximations, we can discover better schemas.  Increasing the
degree of approximation increased the number of spurious tuples
only mildly, yet, at least for some datasets, allowed us to
discover better schemas, e.g. with more relations and/or smaller
separators between relations.
}
}

The most expensive operation of \system\ is the computation of the
entropy $H(X)$ of a set of attributes $X$.  Each such computation
requires a full scan over the data, and this is prohibitively
expensive due to the exponential number of subsets of attributes.  We
describe a novel, efficient approach to computing entropy, which
reduces the problem to a set of main-memory SQL queries. Our method is
inspired by the \e{PLI cache} (Position List Indices)
data structure used for mining both exact and approximate
FDs~\cite{DBLP:journals/pvldb/0001N18,DBLP:journals/cj/HuhtalaKPT99}.

To sum up, the contributions of this work are as follows:
\begin{enumerate}
\item We define a principled notion of approximate data dependencies
  based on information theory, and study its properties;
  Sec.~\ref{sec:problem:statement} and~\ref{sec:techniques}.
% 
%  show that this notion allows us to
%   bound the error of an acyclic schema based on the error of the
%   approximate MVDs it entails.
\item We describe a novel MVD enumeration algorithms and acylic schema
  enumeration algorithm; Sec.~\ref{sec:MiningMVDs}
  and~\ref{sec:EnumerateAcyclicSchemas}.
\item We conduct an extensive experimental evaluation on 20 real
  datasets; Sec.~\ref{sec:evaluation}.
% 
% We introduce an algorithm $\MVDAlg$ that discovers a set of approximate MVDs that hold in a relation. We prove that every approximate MVD can be derived from this set via a series of information theoretic identities.
% 	\item We devise an enumeration algorithm that receives a set of approximate MVDs and outputs the set of acyclic schemas that approximately hold in the relation. This algorithm follows by reduction to the problem of enumerating maximal independent sets in graphs.
% 	\item We introduce an algorithm for efficiently computing the entropy and cardinality of growing sets of attribute values using an in-memory database.
% 	\item We evaluated each one of our algorithms over state-of-the art benchmarks for mining data dependencies, and show that our algorithm discovers high-quality acyclic decompositions, and scales well to large datasets. 
% 
\end{enumerate}

\section{Running  Example}\label{sec:illustrativeExample}

We will use the following running example in this paper.  Consider the
relation $\relation$ over the signature $\Omega=\set{A,B,C,D,E,F}$ in
Figure~\ref{fig:acyclic:schema}.  Ignore the probabilities, we will
use them in Sec.~\ref{sec:notations}.  Also, ignore for now the last
row (in red).  The table with four rows can be decomposed into four
tables, shown in the figure. More precisely, the following join
dependency holds:
$\relation=\relation[ABD]\join \relation[ACD]\join \relation[BDE]\join
\relation[AF]$.  The schema of these four tables is \e{acyclic},
because it admits a join tree, shown in Fig.~\ref{fig:TD} (reviewed in
Sec.~\ref{sec:notations}).  Our goal is to discover this acyclic
schema from the data $\relation$.  For that, we note that the acyclic
schema can be entirely described by three Multivalued Dependencies:
$BD \mvd E| ACF$, $AD \mvd CF| BE$, and $A \mvd F|BCDE$.  Each
corresponds to one edge of the join tree: the left hand size of the
MVD (that we call the \e{key}) is the label of that edge, while the
two sets of attributes correspond to the subtrees connected by the
edge.  For example, the edge $ACD\stackrel{\mbox{\tiny AD}}{-}ABD$ in the
join tree defines the MVD $AD \mvd CF | BE$.  The key $AD$
``separates'' the attributes $CF$ in one subtree from $BE$ in the other
subtree, and we will also call such a set a \e{separator}.
% 
% 
% and its acyclic join decomposition are
% presented in Figure~\ref{fig:acyclic:schema}.  We first consider the
% relation and decomposition without the red lines (in $\relation$ and
% table $BDE$).  Notice that the acyclic join dependency holds because
% .
% The join tree of this decomposition (defined in
% Section~\ref{sec:notations}) is presented in Figure~\ref{ex:TD}.
% 
% Every join dependency is comprised of a set of MVDs.
% The decomposition of Figure~\ref{fig:acyclic:schema} is comprised of the following MVDs that hold in $\relation$: $BD \mvd E| ACF$, $AD \mvd CF| BE$, and $A \mvd F|BCDE$.
% 
Since MVDs are the building blocks of acyclic schemas, their discovery
is a prerequisite for discovering acyclic schemas, and our first task
is to discover MVDs from data, then use them to discover acyclic
schemas.

Consider the 5'th row in $\relation$, shown in red.  By adding it, we
need to add a 4'th row to $\relation[BDE]$, also shown in red.
However, now the join dependency no longer holds exactly, because
$\relation[ABD]\join \relation[ACD]\join \relation[BDE]\join
\relation[AF]$
contains a spurious tuple, namely $(a_2, b_2,c_2, d_2,e_2,f_2)$, which
is not in $\relation$ (it is not shown in the Figure); the first two
MVDs no longer hold, only $A \mvd F|BCDE$ still holds, and the acyclic
schema is no longer a correct decomposition of $R$.  Yet the schema
can still be useful for many applications, even it if leads to a
spurious tuple.  Insisting on exact acyclic schemas would severely
restrict their applications, and also make them very brittle since the
addition of one single tuple would invalidate the schema.  In this
paper we compute \e{approximate} acyclic schemas, and \e{approximate}
MVDs.  By allowing approximations, the schema shown in the figure is
still considered valid for the data, despite the spurious tuple.

\section{Background}
\label{sec:notations}
Table~\ref{table:notations} summarizes the notations in this paper.
We denote by $[n] = \set{1,\ldots,n}$.  Let $\Omega$ be a set of
variables, also called attributes.  If $X,Y \subseteq \Omega$, then
$XY$ denotes $X \cup Y$.

\small{
\begin{table}	
  \centering
  \begin{tabular}{|l|l|} \hline
   $\Omega$ & set of variables (attributes) \\ \hline
   $n=|\Omega|$ & number of variables (attributes) \\ \hline
   $X,Y,A,B,\ldots$ & sets of variables $\subseteq \Omega$ \\ \hline
   $\schema$ & a schema $=\set{\Omega_1, \ldots, \Omega_m}$ \\ \hline
   $X \mvd Y|Z$ & a standard MVD \\ \hline
   $X \mvd Y_1|Y_2|\cdots |Y_m$ & an MVD~\cite{DBLP:conf/sigmod/BeeriFH77} \\ \hline
   $(\T,\chi)$ & a join tree \\ \hline
   $H(X)$ & entropy of a set of variables $X$ \\ \hline
   $H(Y|X), I(Y;Z|X)$ & entropic measures \\ \hline
   $\J(\T,\chi)$ & the entropic measure in Eq.(\ref{eq:JTScore}) \\ \hline
   $\J(\schema)$ & $\J$ of any join tree for $\schema$\\ \hline
   $\J(X \mvd Y_1|\cdots|Y_m)$ & $\J$ of the schema $\set{XY_1,\ldots,XY_m}$\\ \hline
   $\J(X \mvd Y|Z)$ & $=I(Y;Z|X)$ \\ \hline
   $R$ & a relation \\ \hline
   $N=|R|$ & number of tuples \\ \hline
   $R \models \AJD(\schema)$ & $R$ satisfies an acyclic \\
    &  join dependency \\ \hline
   $R \models_\varepsilon \AJD(\schema)$ & $R$ $\varepsilon$-satisfies an acyclic \\
    &  join dependency \\ \hline    
  \end{tabular}  
  \caption{Notations \label{table:notations}}  
\end{table}

}%

\normalsize{
\subsection{Data Dependencies}

\label{subsec:ci:ic}
Fix a relation instance $\relation$ of size $N=|\relation|$, and
schema $\Omega$.  For $Y \subseteq \Omega$ we let $\relation[Y]$
denote the projection of $\relation$ onto the attributes $Y$.

Let $X, Y, Z \subseteq \Omega$.  
% We say that the instance $\relation$
% satisfies the \e{functional dependency} (FD) $X \fd Y$, and write
% $\relation \models X \fd Y$, if forall $t_1, t_2 \in \relation$,
% $t_1[X]=t_2[X]$ implies $t_1[Y]=t_2[Y]$.  
A \e{schema} is a set $\schema=\set{\Omega_1,\dots, \Omega_k}$ such
that $\bigcup_{i=1}^k\Omega_i=\Omega$ and  $\Omega_i\not\subseteq
\Omega_j$ for $i\neq j$. We say that the relation instance $\relation$
satisfies the \e{join dependency} $\JD(\schema)$, and write
$\relation \models \JD(\schema)$, if
$\relation = \Join_{i=1}^k\relation[\Omega_i]$.  We say that
$\relation$ satisfies the \e{multivalued dependency} (MVD)
$\phi = X \mvd Y_1|Y_2|\dots|Y_m$ where $m\geq 2$, the $Y_i$s are
pairwise disjoint, and $XY_1\cdots Y_m = \Omega$, if
$\relation=\relation[XY_1]\Join\cdots \Join \relation[XY_m]$.  We call
$X$ the \e{key} of the MVD and $\set{Y_1,\dots,Y_m}$ it's
\e{dependents}, denoted $\key(\phi)=X$ and
$\components(\phi)=\set{Y_1,\dots,Y_m}$.  Most of the literature
considers only MVDs with $m=2$, which we call here {\em standard
  MVDs}.  Beeri et al.~\cite{DBLP:conf/sigmod/BeeriFH77} noted that a
generalized MVD can encode concisely multiple MVDs; for example
$X \mvd A|B|C$ holds iff $X \mvd AB|C$, $X \mvd A|BC$ and
$X \mvd AC|B$ hold.
% \subsubsection{Acyclic Join Dependencies and Join Trees}\label{sec:prelimsAcyclicJoinDeps}
We review a \e{join tree} from~\cite{Beeri:1983:DAD:2402.322389}:
\begin{definition}\label{def:joinTree}
  A \e{join tree} is a pair $\left(\T,\jointreeMapFunction\right)$
  where $\T$ is an undirected tree,
  \eat{\footnote{We will assume  w.l.o.g. that $\T$ is connected.  For example, the disconnected schema $\schema = \set{AB,CDE}$ can be represented by adding artificially the edge $AB-CDE$, thus connecting the two disconnected components.}}
  and $\jointreeMapFunction$ is a
  function that maps\eat{associates to} each $u \in \nodes(\T)$ to a
  set of variables $\jointreeMapFunction(u)$, called a \e{bag}, such
  that the following \e{running intersection} property holds: for
  every variable $X$, the set
  $\setof{u \in \nodes(\T)}{X \in \chi(u)}$ is a connected component
  of $\T$.  We denote by $\chi(\T) \defeq \bigcup_u \chi(u)$, the set
  of variables of the join tree.
\end{definition}
We often denote the join tree as $\T$, dropping $\jointreeMapFunction$
when it is clear from the context. The \e{schema} defined by $\T$ is
$\schema=\set{\Omega_1,\dots,\Omega_m}$, where
$\Omega_1, \ldots, \Omega_m$ are the bags of $\T$.  We call a schema
$\schema$ \e{acyclic} if there exists a join tree whose schema is
$\schema$.  Since we required $\Omega_i\not\subseteq \Omega_j$ for
$i\neq j$, one can prove that any acyclic schema with $n$ attributes
and $m$ relations satisfies $m \leq n$.  We say that a relation
$\relation$ satisfies the \e{acyclic join dependency} $\schema$, and
denote $\relation \models \AJD(\schema)$, if $\schema$ is acyclic and
$\relation \models \JD(\schema)$.  An MVD $X \mvd Y_1 | \cdots | Y_m$
represents a simple acyclic schema, namely
$\schema = \set{XY_1, XY_2, \ldots, XY_m}$.
% Any tree with $m$
% nodes $\nodes(\T)=\set{u_1,\ldots, u_m}$ whose bags are
% $\chi(u_i) = XY_i$ is a join tree for this schema.

Let $\schema=\set{\Omega_1,\dots,\Omega_m}$ be an acyclic schema with
join tree $(\T,\jointreeMapFunction)$.  We associate to every
$(u,v) \in \edges(\T)$ an MVD $\phi_{u,v}$ as follows.  Let $\T_u$ and
$\T_v$ be the two subtrees obtained by removing the edge $(u,v)$.
Then, we denote by
$\phi_{u,v} \defeq \jointreeMapFunction(u) \cap
\jointreeMapFunction(v) \mvd \jointreeMapFunction(\T_u) |
\jointreeMapFunction(\T_v)$.
We call the \e{support of $\T$} the set of $m-1$ MVDs associated to
its edges, in notation
$\MVD(\T) = \setof{\phi_{u,v}}{(u,v) \in \edges(\T)}$.  If $\T$
defines the acyclic schema $\schema$, then it satisfies
$\relation \models \AJD(\schema)$ iff it satisfies all MVDs in its
support: $\relation \models \phi_{u,v}$ for all
$\phi_{u,v} \in MVD(\T)$~\cite[Thm.  8.8]{Beeri:1983:DAD:2402.322389}.

\begin{example}
  \label{ex:j0} We will illustrate with the running example from
  Sec.~\ref{sec:illustrativeExample}.  The tree in Fig.~\ref{fig:TD}
  is a join tree. Its bags are the ovals labeled $AF$, $ACD$, $ABD$,
  and $BDE$, and it is custom to show the intersection of two bags on
  the connecting edge.
  $\MVD(\T) = \set{BD \mvd E| ACF, AD \mvd CF| BE, A \mvd F|BCDE}$.
\end{example}

% then  $\relation \models (\Omega_i{\cap} \Omega_j)\mvd \jointreeMapFunction(\T_i)|\jointreeMapFunction(\T_j)$~(Thm. 8.8 ~\cite{Beeri:1983:DAD:2402.322389}), and we say that this MVD is implied by the join tree $(\T,\jointreeMapFunction)$. We call the MVDs implied by a join tree $(\T,\jointreeMapFunction)$ the \e{MVD support} of $(\T,\jointreeMapFunction)$, and denote this set of MVDs by $\MVD(\T)$. Since every MVD in $\MVD(\T)$ is associated with an edge of $\T$, then  $|\MVD(\T)|\leq |\nodes(\T)|-1$. Every join tree over $n$ attributes can have at most $n$ nodes~\cite{Beeri:1983:DAD:2402.322389}, hence $|\MVD(\T)|\leq |\nodes(\T)|-1\leq n-1$.

}

\subsection{Information Theory}
\label{sec:it}
Lee~\cite{DBLP:journals/tse/Lee87,DBLP:journals/tse/Lee87a} gave an
equivalent formulation of data dependencies in terms of information
measures; we review this briefly here, after a short background on
information theory.
% showed how entropic functions can be used to quantitatively represent
% the amount of information contained in subsets of attributes, and thus
% characterize the data dependencies that hold in the relation.  We
% start with a brief background on information theory, and then present
% the characterization of data dependencies using information
% theory~\cite{DBLP:journals/tse/Lee87,DBLP:journals/tse/Lee87a}.

% \subsubsection{Background on information theory}
% \label{subsec:information:theory}

Let $X$ be a random variable with a finite domain $\D$ and probability
mass $p$ (thus, $\sum_{x \in \D} p(x)=1$). Its entropy is:
\begin{equation}\label{eq:entropy}
H(X)\eqdef\sum_{x\in \D}p(x)\log\frac{1}{p(x)}
\end{equation}
If $N = |\D|$ then $H(X) \leq \log N$, and equality holds iff $p$ is
uniform.  For a set of jointly distributed random variables
$\Omega=\set{X_1,\dots,X_n}$ we define the function
$H : 2^\Omega \rightarrow \real$ as the entropy of the joint random
variables in the set. For example,
$H(X_1X_2)=\sum_{x_1 \in \D_1, x_2\in \D_2}
p(x_1,x_2)\log\frac{1}{p(x_1,x_2)}.$ Let $A,B,C \subseteq
\Omega$. The \e{mutual information} $I(B;C|A)$ is defined as:
\begin{equation} \label{eq:h:mutual:information}
I(B;C|A) \eqdef~H(AB) + H(AC) - H(ABC) - H(A)
\end{equation}
It is known that the conditional independence
$p \models B \perp C \mid A$ (i.e., $B$ is independent of $C$ given
$A$) holds iff $I(B;C|A)=0$.
% , and similarly
% $p \models A \fd B$ iff $H(B|A)=0$, thus, entropy provides a
% characterization of conditional independence.

In this paper we use only the following two properties of the mutual
information:
\begin{align}
  I(B;C|A) \geq & 0 \label{eq:shannon} \\
  I(B;CD|A) = & I(B;C|A) + I(B;D|AC) \label{eq:ChainRuleMI}
\end{align}
The first inequality follows from monotonicity and submodularity (it
is in fact equivalent to them); the second equality is called the
\e{chain rule}.  All consequences of these two (in)equalities are
called \e{Shannon inequalities}; for example, monotonicity
$H(AB) \geq H(A)$ is a Shannon inequality because it follows from
\eqref{eq:shannon} by setting $B=C$.
% 
% 
% For every entropic function $H$ the mutual information is
% positive. That is, $I(B;C|A) \geq 0$ for any three sets $A,B$, and
% $C$.  Let $A,B,C$ and $D$ be disjoint sets of variables in
% $\Omega$. The \e{chain rule} is the identity:
% \begin{equation} \label{eq:ChainRuleMI}
% I_H(B;CD|A)=I_H(B;C|A)+I_H(B;D|AC)
% \end{equation}

Let $R$ be relation with attributes $\Omega=\set{X_1,\dots,X_n}$ and
$N$ tuples. The \e{empirical distribution} is the uniform distribution
over its tuples: $\forall t{\in} R$, $p(t) {=} 1/N$.  It's entropy
satisfies $H(\Omega)=\log N$.  For $\alpha \subseteq [n]$, we denote
by $X_\alpha$ the set of variables $X_i, i \in \alpha$, and denote by
$\relation(X_\alpha{=}x_\alpha)$ the subset of tuples $t \in R$ where
$t[X_\alpha]{=}x_\alpha$, for fixed values $x_\alpha$.  By uniformity,
the marginal probability is
$p(X_\alpha{=}x_\alpha){=}\frac{|\relation(X_\alpha{=}x_\alpha)|}{N}$,
and therefore:
\begin{equation}\label{eq:jointEntropy}
H(X_\alpha)\eqdef \log N-\frac{1}{N}\sum_{x_\alpha {\in} \D_\alpha}|\relation(X_\alpha{=}x_\alpha)|\log |\relation(X_\alpha{=}x_\alpha)|
\end{equation}
The sum above can be computed using a simple SQL query: %, a property that we will exploit later in the paper:
$\select\ X_\alpha, \sql{count(*)}\times\log(\sql{count(*)})\ \from\ R\ \groupby\ X_\alpha$.
%
% \begin{align*}
%   & \select\ X_\alpha, \sql{count(*)}\times\log(\sql{count(*)})\ \from\ R\ \groupby\ X_\alpha
% \end{align*}

% \begin{align*}
%   & \texttt{SELECT } X_\alpha, \texttt{count(*)}*\log(\texttt{count(*)})   \texttt{ FROM } R \texttt{ GROUP BY } X_\alpha
% \end{align*}

Lee~\cite{DBLP:journals/tse/Lee87,DBLP:journals/tse/Lee87a} formalized
the following connection between database constraints, and entropic
measures.  Let $(\T,\chi)$ be a join tree.  We define the following
expression:
%
%\small{
\begin{equation}\label{eq:JTScore}
  \J(\T,\chi){\eqdef}\sum_{\substack{v\in\\ \nodes(\T)}}H(\jointreeMapFunction(v))-\sum_{\substack{(v_1,v_2)\in\\ \edges(\T)}}H(\jointreeMapFunction(v_1) {\cap} \jointreeMapFunction(v_2))-H(\chi(\T))
\end{equation}
%}%
\normalsize{ We abbreviate it with $\J(\T)$, or $\J$, when $\T, \chi$
  are clear from the context; we will prove later
  (Th.~\ref{thm:approximateAcyclic}) that $\J \geq 0$ is a Shannon
  inequality.  Lee proved that $\J$ depends only on the schema
  $\schema$ defined by the join tree, and not on the tree itself.  To
  see this on a simple example, consider the MVD $X \mvd U|V|W$ and
  its associated acyclic schema $\set{XU, XV, XW}$.  If we consider
  the join tree $XU-XV-XW$, then
  $\J = H(XU)+H(XV)+H(XW)-2H(X)-H(XUVW)$.  Another join tree is
  $XU-XW-XV$, and $\J$ is the same.  Therefore, if $\schema$ is
  acyclic, then we write $\J(\schema)$ to denote $\J(\T)$ for any join
  tree of $\schema$.  We denote by
  $\J(X \mvd Y_1|\cdots |Y_m) \defeq
  H(XY_1)+\cdots+H(XY_m)-(m-1)H(X)-H(XY_1\cdots Y_m)$
  for any sets of variables $X, Y_1, \ldots, Y_m$ where
  $Y_1, \ldots, Y_m$ are pairwise disjoint, even when $XY_1\cdots Y_m$
  is not necessarily $\Omega$.  When $m=2$, then
  $J(X \mvd Y|Z) = I(Y;Z|X)$.  Lee proved the following: }

\def\AcyclicCharacterizationLee{Let $H$ be the entropy of the
  empirical distribution on $\relation$, and let $\schema$ be any
  acyclic schema.  Then $R \models \AJD(\schema)$ iff $\J(\schema)=0$.
}
\begin{citedtheorem}{\cite{DBLP:journals/tse/Lee87a}}\label{thm:AcyclicCharacterizationLee}
\AcyclicCharacterizationLee
\end{citedtheorem}

In the particular case of a standard MVD, Lee's result implies that $\relation \models X \mvd Y|Z$ if and only if $I(Y;Z|X)=0$.
\eat{
  \begin{align}
%    R \models & X \fd Y & \Leftrightarrow &&& H(Y|X)=0 \label{eq:lee:fd}	\\
    R \models & X \mvd Y|Z & \Leftrightarrow &&& I(Y;Z|X)=0 \label{eq:lee:mvd}
  \end{align}
}
  \begin{example} \label{ex:j1} Continuing Example~\ref{ex:j0}, the
    empirical distribution of the relation $R$ in
    Fig~\ref{fig:acyclic:schema} (without the red tuple) assigns
    probability $1/4$ to each tuple. Thus, $H(ABCDEF)=\log 4=2$.  The
    marginal probabilities need not be uniform, e.g. the marginals for
    $BDE$ are $1/4, 1/4, 1/2$, and thus
    $H(BDE) = 1/4\log 4 + 1/4\log 4 + 1/2\log 2 = 3/2$.  The value of
    $\J$ is:
    $\J(\T) = H(AF)+H(ACD)+H(ABD)+H(BDE)-H(A)-H(AD)-H(BD)-H(ABCDEF)$.
    For the empirical distribution in the figure, this quantity is 0.
  \end{example}

\section{Problem Statement}

\label{sec:problem:statement}

Our main goal is to discover an acyclic schema for a given relation
instance $R$.  Since exact schemas are very sensitive to data errors,
\system\ discovers approximate schemas.

\begin{definition}[Approximate Acyclic Schema]\label{def:ApproxAcyclicJoin}
  Fix a relation instance $\relation$, and $\varepsilon {\geq} 0$.  We
  say that an acyclic schema $\schema$ is an \e{$\varepsilon$-schema} for $\relation$, or simply
  \e{approximate schema}, if
  $\J(\schema) \leq \varepsilon$.  In notation,
  $\relation \models_{\varepsilon} \AJD(\schema)$.
\end{definition}

\system\ takes as input $\varepsilon \geq 0$ and discovers approximate
acyclic schemas for $\relation$.  By Lee's theorem, if we set
$\varepsilon=0$, then \system\ returns exact schemas.  In practice, a
relation $R$ may not have any exact schemas, or may have very limited
schemas; by allowing $\varepsilon \geq 0$ we may find approximate
schemas that are quite useful for many applications.  
\eat{
In general, we
are not interested in listing {\em all} approximate schemas, because
there are too many, and the user only needs one, or a few to choose
from.  Instead, we ask for an enumeration algorithm:
}
\begin{problem}[Schema Enumeration Problem] \label{problem:ajd} Given
  a relational instance $R$, enumerate the approximate  acyclic schemas of $R$.
\end{problem}
In practice, we are not interested in enumerating \e{all} approximate acyclic schemas of $\relation$. This would take a 
prohibitively long time, and some acyclic schemas are superior to others. For example, consider a relation over four attributes that satisfies the acyclic join dependency $\schema=\set{XA,XB,XC}$. The following acyclic join dependencies also hold in $\relation$: $\set{XAB,XC}$, $\set{XAC,XB}$, and $\set{XA,XBC}$. The latter schemas are less useful than $\schema=\set{XA,XB,XC}$ that leads to a larger degree of decomposition. Therefore, in this paper we address the problem of enumerating acyclic schemas that cannot be extended (i.e., with additional relational instances) while continuing to satisfy the accuracy threshold.

\eat{
A naive approach is to enumerate all candidate schemas $\schema$ and
check $\J(\schema) \leq \varepsilon$.  There are two problems with
that: the number of candidates is too large, and each computation of
$\J(\schema)$ requires multiple passes over the data $R$ (see the SQL
query in Sec.~\ref{sec:it}).  
}
We derive the
approximate schemas from the MVDs in their support.  Since an MVD is,
in particular, an acyclic schema, Def.~\ref{def:ApproxAcyclicJoin}
applies to them as well: a $\varepsilon$-MVD is one for which
$\J(X \mvd Y_1|\cdots|Y_m) \leq \varepsilon$.  Our second problem is:

\begin{problem}[MVD Enumeration Problem] \label{problem:mvd} Given a
  relational instance $R$, enumerate the approximate MVDs of $R$.
\end{problem}

\system\ works as follows.  The user provides a parameter
$\varepsilon \geq 0$.  In the first phase, \system\ enumerates
$\varepsilon$-MVDs, using the algorithm in Sec.~\ref{sec:MiningMVDs}.
When it finishes, or after a timeout, it starts the second phase,
where it enumerates approximate schemas with support from the set
returned by the first phase, using the algorithm in
Sec.~\ref{sec:EnumerateAcyclicSchemas}. 
\eat{This algorithm has polynomial delay enumeration.}
Since the support of a schema consists of $m-1$
MVDs, the algorithm reports schemas with
$\J(\schema) \leq (m-1)\varepsilon$, where $m$ is the number of
relations in $\schema$ but, since the enumeration algorithm is
exhaustive, all schemas with $\J \leq \varepsilon$ are reported
eventually.

\section{Three Main Techniques}

\label{sec:techniques}

We describe here three main techniques that allow us to design
efficient schema- and MVD-discovery algorithms.  The first reduces the
approximate schema discovery to approximate MVD discovery, the next
two reduce the number of MVD's that need to be discovered.

\subsection{From  MVDs to Acyclic Schemas}

Beeri at al.~\cite{Beeri:1983:DAD:2402.322389} showed
that, for exact constraints, an acyclic schema over $m$ relations is
equivalent to the set of $m-1$ MVDs in its support.  We give here a
non-trivial generalization to approximate schemas and MVDs.  We start
with two simple inequalities which we need throughout the paper:

\def\downwardClosureGeneral{ Let $Y_1, Z_1, \ldots, Y_m, Z_m$ be
  pairwise disjoint sets of variables, and let $X$ be any set of
  variables.  Then the following are Shannon inequalities:
{\small
\begin{align}  	
  \J(X \mvd Y_1| \cdots |Y_m) \leq & \J(X \mvd Y_1Z_1 | \cdots | Y_mZ_m) \label{eq:downward:1} \\
  \J(XZ_1\cdots Z_m \mvd Y_1| \cdots |Y_m) \leq & \J(X \mvd Y_1Z_1 | \cdots | Y_mZ_m) \label{eq:downward:2}
\end{align}
}%
}
\begin{prop}\label{prop:downwardClosureGeneral}
	\downwardClosureGeneral
\end{prop}

\begin{proof}
  The first inequality follows from this chain of inequalities:
  $\J(X \mvd Y_1 | \cdots | Y_m) \leq \J(X \mvd Y_1Z_1 | Y_2 | \cdots
  | Y_m) \leq \J(X \mvd Y_1Z_1 | Y_2Z_2 | \cdots | Y_m) \leq \cdots$;
  to prove it, we show only the first step (the others are similar),
  which follows by observing
  $\J(X \mvd Y_1 | \cdots | Y_m) + I(Z_1; Y_2\cdots Y_m |
  XY_1) = \J(X \mvd Y_1Z_1 | \cdots | Y_m)$ then using
  inequality~\eqref{eq:shannon}.  The second inequality follows from a
  similar chain, where the first step follows from
  $\J(XZ_1 \mvd Y_1 | \cdots |Y_m)+\sum_{i=2}^mI(Y_i;Z_1|X)=\J(X
  \mvd Y_1Z_1 | Y_2 | \cdots | Y_m)$
  \eat{
  $\J(XZ_1 \mvd Y_1 | \cdots |Y_m) + (m-1)h(XZ_1) - (m-1)h(X) = \J(X
  \mvd Y_1Z_1 | Y_2 | \cdots | Y_m)$}
  and the inequality follows from~\eqref{eq:shannon}.  
\end{proof}

Let $(\T,\chi)$ be a join tree, defining an acyclic schema $\schema$
over the variables $\chi(\T)=\Omega$.  Choose an arbitrary root,
orient the tree accordingly, and let $u_1, \ldots, u_m$ be a
depth-first enumeration of $\nodes(\T)$.  Thus, $u_1$ is the root, and
for every $i > 1$, $\parent(u_i)$ is some node $u_j$ with $j < i$.
For every $i$, we define $\Omega_i \defeq \chi(u_i)$,
$\Omega_{i:j} \defeq \bigcup_{\ell=i,j} \Omega_\ell$, and
$\Delta_i \defeq \chi(\parent(u_i)) \cap \chi(u_i)$ (by the running
intersection property this is equal to
$\Omega_{1:(i-1)}\cap\Omega_i$).  We prove:

\begin{theorem}\label{thm:approximateAcyclic}  The following hold:
\small{	
  \begin{align}
    \J(\T) = & \sum_{i=2}^m I(\Omega_{1:(i-1)};\Omega_i|\Delta_i)\label{eq:approximateAcyclic}\\
\hspace{-1cm}  \max_{i=2,m} I(\Omega_{1:(i-1)};\Omega_{i:m}|\Delta_i) \leq \J(\T)  \leq &\sum_{i=2}^m I(\Omega_{1:(i-1)};\Omega_{i:m}|\Delta_i)
\label{eq:approximateAcyclic:2}
  \end{align}
}%
\normalsize{
  The first is an identity, and the second is a Shannon inequality.
}
\end{theorem}

The identity \eqref{eq:approximateAcyclic} captures precisely the
intuition that the information measure associated with a join tree $\T$ is equivalent to $m-1$ mutual
information.  This identity implies that $\J(\T) \geq 0$, because
$I(\cdots) \geq 0$.  But the expressions $I(\cdots)$ in
\eqref{eq:approximateAcyclic} do not correspond to MVDs, because they
do not include all variables $\Omega$.  The Shannon inequality
\eqref{eq:approximateAcyclic:2} rectifies this, by showing that
$\J(\T)$ lies between the max and the sum of $m-1$ MVDs.  Notice that
the MVDs $\Delta_i \mvd \Omega_{1:(i-1)}|\Omega_{i:m}$, $i=2,m$ are
precisely the support, $\MVD(\T)$, thus
\eqref{eq:approximateAcyclic:2} generalizes Beeri's observation to
approximate schemas.  An immediate consequence of
\eqref{eq:approximateAcyclic:2} is the following relationship between
an acyclic schema $\schema$ and its support.  \eat{ \batya{Added a
    more formal corollary below. Actually, in the set MVD(T) holds
    approximately, it only guarantees that
    $\J(\T)\leq (m-1)\varepsilon$. So we should remove this sentence.}
  This shows that the acyclic schema holds approximately iff all $m-1$
  MVDs hold approximately.  }
\begin{corollary}\label{corr:MVDT}
  Let $\bS$ be an acyclic schema with join tree
  $(\T,\jointreeMapFunction)$. Then: (1) if
  $\relation \models_\varepsilon \AJD(\bS)$ then
  $\relation \models_\varepsilon \MVD(\T)$.  (2) If
  $\relation \models_\varepsilon \MVD(\T)$ then
  $\relation \models_{(m-1)\varepsilon} \AJD(\bS)$.  In particular,
  (1) and (2) are equivalent if $\varepsilon=0$.  Here
  $\relation \models_\varepsilon \MVD(\T)$ means
  $\relation \models_\varepsilon \phi$, forall $\phi \in \MVD(\T)$.
\end{corollary}

% \begin{proof}
% 	Both are immediate from~\eqref{eq:approximateAcyclic:2} by noting that
% 	the MVDs of~\eqref{eq:approximateAcyclic:2} are precisely the set $\MVD(\T)$.
% 	If $\relation \models_\varepsilon \AJD(\bS)$ then $\J(\T)\leq \varepsilon$,
% 	which, by~\eqref{eq:approximateAcyclic:2} means that $\J(\phi)\leq \varepsilon$ for every MVD $\phi\in \MVD(\T)$.
% 	If $\relation \models_\varepsilon \MVD(\T)$ then by~\eqref{eq:approximateAcyclic:2} $\J(\T)\leq (m-1)\varepsilon$. % or that $\relation \models_{(m-1)\varepsilon}\AJD(\bS)$.
% \end{proof}
% 

\begin{proof} (of Theorem~\ref{thm:approximateAcyclic}) Let $\T_i$
  denote the subtree consisting of the nodes $u_1, \ldots, u_i$.  We
  prove \eqref{eq:approximateAcyclic} by induction on $m$.  Assume the
  identity holds for $m-1$.  Compared to $\T_{m-1}$, the tree $\T_m$
  has one extra node $u_m$ and one extra edge $(\parent(u_m), u_m)$,
  hence by the definition of $\J$ in \eqref{eq:JTScore}:
\small{
	\setlength{\abovedisplayskip}{6pt}
	\setlength{\belowdisplayskip}{\abovedisplayskip}
	\setlength{\abovedisplayshortskip}{0pt}
	\setlength{\belowdisplayshortskip}{3pt}
  \begin{align*}
 \J(\T_m)= & \J(\T_{m-1}) + H(\chi(u_m)) - H(\chi(u_m)\cap \chi(\parent(u_m)) \\
   & + H(\chi(\T_{m-1})) - H(\chi(\T_m)) \\
 = & \J(\T_{m-1}) + H(\Omega_m) - H(\Delta_m) + H(\Omega_{1:(m-1)}) - H(\Omega_{1:m})\\
 = & \J(\T_{m-1}) + I(\Omega_{1:(m-1)}; \Omega_m | \Delta_m)
  \end{align*}
}%
\normalsize{
The claim follows from the induction hypothesis on
$\J(\T_{m-1})$.

We prove \eqref{eq:approximateAcyclic:2}.  The right inequality
 follows from the fact that
 $I(\Omega_{1:(i-1)};\Omega_i|\Delta_i) \leq
  I(\Omega_{1:(i-1)};\Omega_{i:m}|\Delta_i)$ (which holds by
  Eq.~\eqref{eq:downward:2}).  For the left inequality, we make the
  following observation.  If $\T$ is any join tree and $\T'$ is
  obtained by mergining two adjacent nodes $(u,v) \in \edges(\T)$,
  then $\J(\T) \geq \J(\T')$.  This is because
  $\J(\T) = \J(\T') +
  H(\chi(u))+H(\chi(v))-H(\chi(u)\cap\chi(v))-H(\chi(u)\cup\chi(v)) =
  \J(\T') + I(\chi(u);\chi(v)|\chi(u)\cap\chi(v))$.  To prove
  \eqref{eq:approximateAcyclic:2}, we fix one edge $(\parent(u_i), u_i)$
  and repeatedly merge all other edges, until we end with a tree $\T'$
  with two bags, $\Omega_{1:(i-1)}$ and $\Omega_{i:m}$ respectively.
  Then
  $\J(\T) \geq \J(T') = I(\Omega_{1:(i-1)};\Omega_{i:m}|\Delta_i)$.
  The claim follows from the fact that this holds for any $i=2,m$.
}
\end{proof}

\begin{example}
  We illustrate the first part of the theorem on the running example
  in Fig.~\ref{fig:TD} and Example~\ref{ex:j1}. Enumerating the nodes
  depth-first ($ABD, ACD, AF, BDE$), Eq.~\eqref{eq:approximateAcyclic}
  and \eqref{eq:approximateAcyclic:2} become:
\small{
  \begin{align*}
\eat{    \J(\T) = &  I(AF;ABCD|A) + I(ACD;ABD|AD) + I(ABCDF;BDE|BD) \\
  = & I(C;B|AD)+ I(F;BCD|A) + I(ACF;E|BD)
}
    \J(\T) = & I(C;B|AD)+ I(F;BCD|A) + I(ACF;E|BD) \\
 \max(\cdots) \leq \J(\T) \leq &I(CF; BE|AD)+ I(F;BCDE|A)+I(ACF; E|BD)
  \end{align*}
}
%
% \normalsize{
% which is identical to the expression $\J(\T)$ in
% Example~\ref{ex:j1}. 
% }

%   We illustrate the first part of the theorem on the running example in
%   Fig.~\ref{fig:TD}. Enumerating the nodes counter-clockwise (e.g., $AF$ is indexed $3$, and $BDE$ indexed $4$),
%   Eq.~\eqref{eq:approximateAcyclic} proves that $\J(\T)$ satisfies the
%   following identity:
% %
% \small{
%   \begin{align*}
%     \J(\T) = &  I(AF;ABCD|A) + I(ACD;ABD|AD) + I(ABCDF;BDE|BD) \\
%   = & I(F;BCD|A) + I(C;B|AD) + I(ACF;E|BD)
%   \eat{
%     \J(\T) = &  I(AF;ACD|A) + I(ACDF;ABD|AD) + I(ABCDF;BDE|BD) \\
%            = & I(F;CD|A) + I(CF;B|AD) + I(ACF;E|BD)
%            }
%   \end{align*}
% }%
% \normalsize{
% which is identical to the expression $\J(\T)$ in
% Example~\ref{ex:j1}. 
% }
% 
\eat{
\normalsize{
while \eqref{eq:approximateAcyclic:2} becomes:
}
\begin{align*}
  \max(e_1,e_2,e_3) \leq & \J(\T) \leq  e_1+e_2+e_3 \\
   e_1 = & I(AF;ABCDE|A) = I(F;BCDE|A) \\
   e_2 = & I(ACDF; ABDE|AD) = I(CF; BE|AD) \\
   e_3 = & I(ABCDF; BDE | BD) = I(ACF; E|BD)
\end{align*}
}
% Notice that there are precisely the MVDs that we argued in
% Sec.~\ref{sec:illustrativeExample} are logical consequences of the
% acyclic schema $\T$.
\end{example}

% Thus, our strategy is to find all $\varepsilon$-MVDs, then find a
% subset that defines an acyclic schema $\schema$.  Notice that
% $\J(\schema)$ may be as large as $(m-1)\cdot \varepsilon$, where $m$
% is the number of relations in $\schema$; this is the reason why we
% may return schemas with a slightly larger $\J$.

\subsection{Full MVDs}\label{sec:measures}

The number of candidate MVD's is very large: there are\footnote{There
  are $3^n$ ways to partition $\Omega$ into three sets $X, Y, Z$.  We
  rule out the $2^n$ partitions that have $Y=\emptyset$ and the $2^n$
  partitions that have $Z=\emptyset$, and add back the 1 partition
  that has $Y=Z=\emptyset$, for a total of $3^n-2^{n+1} + 1$.
  Finally, we divide by 2 since $X \mvd Y|Z$ and $X \mvd Z|Y$ are the
  same MVD.} $(3^n+1)/2 -2^n = O(3^n)$ standard MVD's $X \mvd Y|Z$,
which is too large to consider for practical datasets.  Here, and in
the next section, we describe two techniques that allow us to restrict
the search space.  Consider a fixed key $X$.  Beeri at
al.~\cite{DBLP:conf/sigmod/BeeriFH77} noted that, in the exact case,
if any MVD $X \mvd \ldots$ holds on the data, then there exists a
``best'' one.  For example if both $X \mvd AB|C$ and $X \mvd A|BC$
hold exactly, then so does $X \mvd A|B|C$, and it suffices to discover
only the latter.  Unfortunately, this fails for approximate MVDs, as
we explain here.

We say that $\phi = X \mvd A_1|\dots|A_m$ \e{refines}
$\psi = X \mvd B_1|\dots|B_k$, denoted by $\phi \succeq \psi$ if they
both have the same key (i.e., $\key(\phi)=\key(\psi)=X$) and for every
$A_i \in \components(\phi)$ there exists $B_j\in \components(\psi)$
such that $A_i \subseteq B_j$.  For example, $X\mvd A|B|C$ refines
$X \mvd AB|C$.

\def\naiveEnumJDs{ If $\phi \succeq \psi$ then $\J(\phi)\geq
  \J(\psi)$.\eat{ is a Shannon inequality.}  }
\begin{prop}\label{prop:naiveEnumJDs}
	\naiveEnumJDs
\end{prop}

\begin{proof}
  It suffices to consider the case when two dependents in $\phi$ are
  replaced by their union in $\psi$, e.g. $\phi = X\mvd A|B|\cdots$
  and $\psi=X \mvd AB|\cdots$, since any refinement is a sequence of
  such steps.  In that case, by inspecting Eq.\eqref{eq:JTScore} we
  observe
  $\J(\phi)=\J(\psi) + H(XA)+H(XB) - H(XAB) - H(X) = \J(\psi) +
  I(A;B|X) \geq \J(\psi)$ proving the claim.
\end{proof}

We say that an MVD $\psi$ is \e{$\varepsilon$-full}, or simply
\e{full}, if $\relation \models_\varepsilon \psi$ and, for all strict
refinements $\phi \succ \psi$,
$\relation \not\models_\varepsilon \phi$.  We denote by
$\fullMVDs_\varepsilon(R,X)$ the set of all full $\varepsilon$-MVDs
with key $X$.  Thus, we only need to discover the sets
$\fullMVDs_\varepsilon(R,X)$, for all $X \subseteq \Omega$, because
all other MVDs can be derived using Shannon inequalities.

Beeri proved that, in the exact case, $\fullMVDs_0(R,X)$ has at most
one element. We next present Lemma~\ref{lem:singleMaxMVDLemma} that shows what happens in the approximate case, and allows us to 
derive Beeri's result as a special case. Given two MVDs $\phi = X \mvd A_1 | \dots |A_m$
and $\psi = X \mvd B_1|\dots|B_k$, define their \e{join} as
$\phi \vee \psi = X \mvd C_{11} | C_{12} | \cdots | C_{mk}$, where
$C_{ij} = A_i \cap B_j$.  Clearly, $\phi \vee \psi$ refines both
$\phi$ and $\psi$, i.e.
$\J(\phi \vee \psi) \geq \max(\J(\phi),\J(\psi))$.  We prove a weak
form of converse:

\def\singleMaxMVDLemma{The following are Shannon inequalities:
  $\J(\phi \vee \psi) \leq \J(\phi) + m \J(\psi)$ and
  $\J(\phi \vee \psi) \leq k \J(\phi) + \J(\psi)$.  }
\begin{lemma}\label{lem:singleMaxMVDLemma}
	\singleMaxMVDLemma
\end{lemma}

By this result, $\J(\phi)=\J(\psi)=0$ implies
$\J(\phi \vee \psi)=0$, which proves Beeri's theorem that
$\fullMVDs_\varepsilon(R,X)$ has at most one element, because if
$\phi_1,\phi_2,\cdots$ are all MVD's with key $X$ that hold exactly on
$\relation$, then $\phi_1 \vee \phi_2 \vee \cdots$ refines all of them
and holds too.  This property was also used by
Draeger~\cite{draeger2016} in his MVD discovery algorithm.  When
$\varepsilon > 0$ however, then this fails.  For a very simple
example, consider a relation with two tuples,
\parbox{2cm}{
\scriptsize
\begin{tabular}{{|cccc|}} \hline
  $X$ & $A$ & $B$ & $C$ \\ \hline
   0 & 0 & 0 & 0 \\
   0 & 1 & 1 & 1 \\ \hline
\end{tabular}
}\qquad
and fix $\varepsilon=1$.  Then
$\relation \models_\varepsilon X \mvd AB|C, X \mvd AC|B, X \mvd BC|A$,
but $\not\models_\varepsilon X \mvd A|B|C$; indeed,
$H(\emptyset)=H(X)=0$ and $H(W)=1$ for all other sets $W$, and the
reader can check
$\J(X\mvd AB|C) = \J(X \mvd AC|B) = \J(X \mvd BC|A)=1$ but
$\J(X \mvd A|B|C)=2$.
% 
% because, if it contained both $\phi$ and
% $\psi$, then $\J(\phi)=\J(\psi)=0$ and the lemma implies
% $\J(\phi \vee \psi)=0$, hence $\fullMVDs_\varepsilon(S)$ should
% contain $\phi \vee \psi$ instead.  This property, however, no longer
% holds when $\varepsilon > 0$.

In summary, our algorithm discovers
$\fullMVDs_\varepsilon(R,X)$, for every $X$.  Unlike the exact case,
$\fullMVDs_\varepsilon(R,X)$ may contain more than one element.
\eat{
\begin{proof} (of Lemma~\ref{lem:singleMaxMVDLemma}) 
We recall that $\phi=X\mvd A_1|\dots|A_m$, $\psi=X\mvd B_1|\dots|B_k$, and $\phi \vee \psi=X\mvd C_{11}|\dots|C_{mk}$ where $C_{ij}=A_i{\cap}B_j$.
	We prove the first inequality (the second is similar),
	and for that we need to show that:\\$\J(\phi)+m\left(\sum_{j=1}^k H(XB_j) - (k-1)H(X) - H(\Omega)\right)\geq
	\sum_{ij} H(XC_{ij}) - (mk-1) H(X) - H(\Omega)$. 
	Since $\J(\phi)=\sum_{i=1}^mH(XA_i)-(m-1)H(X)-H(\Omega)$, we need to show that:
	\begin{align}
	\sum_{i=1}^m  H(XA_i) + m\sum_{j=1}^k H(XB_j) \geq \sum_{ij} H(XC_{ij})+ m H(\Omega)
	\label{eq:singleMax1}
	\end{align}
	For that we prove by induction on $\ell$ that:
	\begin{align}
	H(XA_i) + \sum_{j=1}^\ell H(XB_j) \geq
	\sum_{j=1}^\ell H(XC_{ij}) + H(XA_iB_1\ldots B_\ell) \label{eq:singleMax2}
	\end{align}
	For $l=1$ the statement follows from~\eqref{eq:shannon}. Assuming the statement for $\ell-1$ holds, then the statement
	for $\ell$ follows from:
	\begin{align*}
	H(XB_\ell) + H(XA_iB_1\ldots B_{\ell-1}) \geq H(XC_{i\ell}) + H(XA_iB_1\ldots B_\ell)
	\end{align*}
	which is the submodularity inequality, since
	$XB_\ell \cap (XA_iB_1\ldots B_{\ell-1}) = XB_\ell \cap XA_i =
	XC_{i\ell}$.  Setting $\ell=k$ in \eqref{eq:singleMax2} and summing
	over $i=1,m$ we obtain
	$\sum_{i=1}^m H(XA_i) + m\sum_{j=1}^k H(XB_j) \geq \sum_{ij}
	H(XC_{ij}) + \sum_{i=1}^m H(XA_iB_1\cdots B_k) = \sum_{ij} H(XC_{ij})
	+ m H(\Omega)$, proving \eqref{eq:singleMax1}.
\end{proof}
}
\subsection{Minimal Separators}\label{sec:minSeps}

We now show that it is not necessary to discover the sets $\fullMVDs_\varepsilon(R,X)$  for all subset of attributes $X\subset \Omega$, but only those where $X$ is a \e{minimal
separator}.

\begin{definition}\label{def:minimalSJD}  Fix a relation $R$ and
  $\varepsilon \geq 0$.
  We say that a set $X$ \e{separates} two variables
  $A, B \not\in X$ if there exists an $\varepsilon$-MVD
  $X \mvd Y_1| \cdots |Y_m$ that separates $A,B$, i.e.  $A,B$ occur in different sets $Y_i, Y_j$.
  We say $X$ is a \e{minimal $A,B$-separator} if there is no
    $X_0 \subsetneq X$ that separates $A,B$.
\end{definition}
For a pair $A,B \in \Omega$, we denote by
$\minsep_{\varepsilon}(\relation,A,B)$ the set of minimal $A,B$
separators in $\relation$, and for a minimal $AB$ separator $X$
we denote by
$\fullMVDs_{\varepsilon}(R,X,A,B)$ the set of full MVDs that separate
$A,B$.  Notice that:
$$\fullMVDs_\varepsilon(R,X) = \bigcup_{A,B \in \Omega\sm X}
\fullMVDs_{\varepsilon}(R,X,A,B).$$

%%% Let $\varepsilon\geq 0$, and $\relation$ be a relation over the attribute set $\Omega$. Let $S \subset \Omega$ and $A,B \in \Omega \sm S$. We say that $S$ \e{separates} $A$ and $B$ if $\relation$ approximately satisfies an MVD $\phi$ (i.e, $\relation \models_{\varepsilon} \phi$) such that $\key(\phi)=S$ and $A$ and $B$ are in distinct sets of $\components(S)$, and call $S$ an $AB$-separator.
%%% We say that $S$ is a \e{minimal} $AB$-separator if $S$ is an $AB$-separator and no strict subset of $S$ separates $A$ and $B$.
%%% We say that $S$ is a minimal separator if it is a minimal separator for a pair of attributes in $\Omega$, and denote by $\minsep(\relation)$ the set of minimal separators of $\relation$.
%%% \begin{definition}\label{def:minimalSJD}
%%% 	Let $S \subset \Omega$ be a minimal $AB$-separator, and let $\phi = S \mvd X_1|\dots|X_m$ be an MVD. We say that $\phi$ is a minimal MVD if $A$ and $B$ are in distinct sets of $\components(\phi)$. 
%%% \end{definition}
%%% 

\begin{example}
  Let $R$ be a relation over $\Omega=\set{A,\dots,E}$. Suppose
  $R \models_\varepsilon CD \mvd A|BE$.  By \eqref{eq:downward:2} we
  also have $R \models_\varepsilon CDE \mvd A|B$, which means that
  $CDE$ cannot be a minimal separator for $A,B$.  To check that $CD$
  is a minimal $A,B$-separator, we need to check that neither $C$ nor
  $D$ separates $A,B$
\end{example}

The main result in this section is that we only need to compute the
full MVDs with minimal separators, denoted as:
\begin{align}
  \M_\varepsilon \defeq & \bigcup_{A,B \in \Omega}\bigcup_{\substack{X \in\\ \minsep_{\varepsilon}(\relation,A,B)}}\fullMVDs_{\varepsilon}(R,X,A,B) \label{eq:m}
\end{align}
because, as we show, every $\varepsilon$-MVD can be derived from the set $\M_\varepsilon$ by a Shannon inequality.
\def\FullMVDsCompleteness{
	Let $X \mvd Y|Z$ be an $\varepsilon$-MVD for $R$.  Then there exist
	$\phi_1, \ldots, \phi_m \in \M_\varepsilon$, where $m = |Y|\cdot |Z|$, such that
	the following is a Shannon inequality:
	$I(Y;Z|X) \leq \sum_i \J(\phi_i)$.
}
\begin{theorem}\label{thm:FullMVDSCompleteness}
  \FullMVDsCompleteness
\end{theorem}

In summary, our algorithm will iterate over pairs of attributes $A,B$,
will compute $\minsep_{\varepsilon}(\relation,A,B)$, then, for each
$X$ in this set will compute $\fullMVDs_{\varepsilon}(R,X,A,B)$, and
return their union, $\M_\varepsilon$; we describe it in the next
section.  We end this section with the proof of
Theorem~\ref{thm:FullMVDSCompleteness}.

\begin{proof}
  Let $Y=A_1\dots A_m$, and $Z=B_1 \dots B_k$.  By the chain rule~\eqref{eq:ChainRuleMI} it holds that:
  \begin{align*}
    I(Y;Z|X)=&\sum_{i=1}^m\sum_{j=1}^kI(A_i;B_j|XA_1\dots A_{i-1}B_{1}\dots B_{j-1})
  \end{align*}
  It suffices to prove that, for each $i,j$, there exists an MVD $\phi\in \M_\varepsilon$
  such that the following is a Shannon inequality:
  \begin{align*}
    I(A_i;B_j|XA_1\cdots A_{i-1}B_{1}\cdots B_{j-1}) \leq & \J(\phi)
  \end{align*}
  Since $X \mvd Y|Z$ is a $\varepsilon$-MVD for the relation $R$, then
  $X$ is an $A_i, B_j$ separator.  Let $S \subseteq X$ be any minimal
  $A_i, B_j$ separator, thus
  $S \in \minsep_{\varepsilon}(\relation,A_i,B_j)$, and let
  $\phi = S\mvd U_1|\cdots|U_p$ be a full MVD in $\fullMVDs_{\varepsilon}(R,S,A_i,B_j)
  \subseteq \M_\varepsilon$
  that separates $A_i, B_j$.  Assume w.l.o.g.
  $A_i \in U_1$, $B_j \in U_2$, and let $\psi \defeq S \mvd W | V$,
  where $W=U_1$, $V=U_2U_3\cdots U_p$.  Thus, $\phi \succeq \psi$, and
  therefore by Prop.~\ref{prop:naiveEnumJDs} the following Shannon
  inequality holds: $\J(\phi) \geq \J(\psi)$.  Write $\psi$ as
  $\psi = S \mvd W_0 W_1 | V_0V_1$, where
  $W_0= W \cap (XA_1\cdots A_i B_1 \cdots B_j)$, $W_1 = W - W_0$, and
  similarly $V_0= V \cap (XA_1\cdots A_i B_1 \cdots B_j)$,
  $V_1 = V - V_0$.  By Prop.~\ref{prop:downwardClosureGeneral}
  \eqref{eq:downward:1} we have the following Shannon inequality
  $\J(\psi) = \J(S \mvd W_0W_1 | V_0V_1) \geq \J(S \mvd W_0 | V_0)$.
  Finally, we notice that the set $SW_0V_0$ is the same as
  $XA_1\cdots A_iB_1 \cdots B_j$ and that $A_i \in W_0$,
  $B_j \in V_0$, therefore by Prop.~\ref{prop:downwardClosureGeneral},
  \eqref{eq:downward:2},
  $\J(S \mvd W_0 | V_0) \geq \J(XA_1\cdots A_{i-1}B_{1}\cdots B_{j-1}
  \mvd A_i | B_j)$, proving the claim.
\end{proof}

\section{Discovering  $\varepsilon$-MVDs}\label{sec:MiningMVDs}

\begin{algseries}{t}{Discover the set $\M_\varepsilon=\cup_{S\in \minsep_R}\fullMVDs_\varepsilon(S)$. \label{alg:MVDMiner}}
	\begin{insidealg}{\MVDAlg}{$\relation$, $\Omega$, $\varepsilon$}
		\STATE $\M_\varepsilon \gets \emptyset$
	%	\FORALL{$i,j\in [n] : i<j$}	
		\FORALL{pairs $A,B \in \Omega$}	
			\STATE $\minsep_{\varepsilon}(\relation,A,B) \gets \minSepAlg(\relation,\Omega,\varepsilon,(A,B))$	\label{algline:MVDAlgMinePair}
			\FORALL{$X \in \minsep_{\varepsilon}(\relation,A,B)$}
				\STATE $\M_\varepsilon \gets \M_\varepsilon \cup \algname{getFullMVDs}(X,\varepsilon,(A,B),\infty)$ \label{algline:MVDAlgFullMVDs}
			\ENDFOR	
		\ENDFOR
		\STATE return $\M_\varepsilon$
	\end{insidealg}	
\end{algseries}

In this section we present the first phase of \system: the algorithm
for the discovery of $\varepsilon$-MVDs in a relation $\relation$,
called \MVDAlg, and shown in Figure~\ref{alg:MVDMiner}.  As explained,
the algorithm returns the set $\M_\varepsilon$, defined in
Eq.\eqref{eq:m}; this set is used in the second phase of \system\ to
compute $\varepsilon$-schemes.  

\MVDAlg\ iterates over all pairs of attributes $A,B \in \Omega$.  It
first computes the set $\minsep_{\varepsilon}(\relation,A,B)$ of
minimal $A,B$-separators (line~\ref{algline:MVDAlgMinePair}): we
describe this step in Sec.~\ref{sec:minimal:separators}.  Then, for
each $X \in \minsep_{\varepsilon}(\relation,A,B)$, it computes
$\fullMVDs_{\varepsilon}(\relation,X,A,B)$
(line~\ref{algline:MVDAlgFullMVDs}): we describe this step in
Sec.~\ref{sec:getFullMVDs}.  Finally, the algorithm returns their
union, $\M_\varepsilon$.  Both steps require access to an oracle
$\entropyAlg(X)$ for computing the entropy $H(X)$, according to
Eq.~\eqref{eq:jointEntropy}, where $H$ is the entropy associated with
the empirical distribution over $\relation$. We describe the
implementation and optimization of $\entropyAlg(X)$ in
Section~\ref{sec:ComputeEntropies}.

\subsection{Discovering the  Minimal Separators}

\label{sec:minimal:separators}

We describe here how we compute all minimal $A,B$-separators,
$\minsep_{\varepsilon}(\relation,A,B)$
(line~\ref{algline:MVDAlgMinePair} of $\MVDAlg$).  One possible way to
do this could be to iterate over sets $X$ top down, because it enables pruning: if $X$ is {\em not} an $A,B$-separator, then neither
is any subset of $X$, by~\eqref{eq:downward:2} in Prop.~\ref{prop:downwardClosureGeneral}. This suggests a top-down algorithm, which
starts from the largest set $X = \Omega\sm\set{A,B}$, and checks if it
is an $A,B$-separator.  If not, then none exists. Otherwise it
exhaustively searches over subsets of $X$, from largest to smallest,
returning the minimal (with regard to inclusion) sets that separate $A,B$.  Such an
exhaustive search will explore \e{all} separators, while we only want
to find the \e{minimal} ones.  Our approach takes
advantage of the fact that we need to find only the \e{minimal}
separators, and builds on a result by Gunopulos et
al.~\cite{DBLP:journals/tods/GunopulosKMSTS03}.

Let $\bC=\set{C_1,\dots,C_m}$ be a set of distinct subsets of
$\Omega$.  A set $D \subset \Omega$ is a \e{transversal} of $\bC$ if
$D \cap C_i \neq \emptyset$ for every $C_i \in \bC$. For a set
$D \subseteq \Omega$, we denote by $\comp{D}$ the complement set
$\Omega\sm D$.

\begin{theorem}\label{thm:transversals}
  Let $\bC=\set{C_1,\dots,C_n}$ denote a set of minimal $A,B$
  separators in $\relation$.  Then there exists a minimal
  $A,B$-separator $X \not\in \bC$ iff there exists a minimal
  transversal $D$ of $\bC$ such that $\comp{D}$ is an $A,B$-separator.
%   Let $\bS=\set{S_1,\dots,S_n}$ denote a set of minimal $A,B$
%   separators in $\relation$ (w.r.t $\varepsilon$).  There exists a
%   minimal $AB$ separator $S$ in $\relation$ such that $S \notin \bS$
%   if and only if there exists a minimal transversal $D$ of $\bS$ such
%   that $\comp{D}$ is an $AB$ separator in $\relation$ (w.r.t
%   $\varepsilon$), and $S \subseteq \comp{D}$.
\end{theorem}
\begin{proof}
  \textbf{only if.}  Since $D$ is a transversal of $\bC$ then:
  \begin{equation}
    \bigwedge_{i=1}^n\left(C_i\cap D\neq \emptyset\right) \Longleftrightarrow
    \bigwedge_{i=1}^n(\comp{D}\not\supseteq C_i) \label{eq:transversal}
  \end{equation}
  Since $\comp{D}$ is an $A,B$ separator, there exists some minimal
  separator $X \subseteq \comp{D}$.
  Assume, by contradiction, that $X \supseteq C_i$ for some $C_i \in \bC$. Then
  $\comp{D} \supseteq X \supseteq C_i$, contradicting
  \eqref{eq:transversal}.

  \textbf{if.}  Since $X$ is a minimal $A,B$ separator that is not in
  $\bC$, then $\bigwedge_{i=1}^n(X \not\supseteq C_i)$,
  meaning that $\comp{X}$ is a transveral of $\bC$.  Then any minimal
  transversal $D \subseteq \comp{X}$ satisfies the claim.
\end{proof}
Algorithm $\minSepAlg$ (Fig.~\ref{alg:MineAllMinSeps}) for discovering all minimal $A,B$ separators, $\minsep_{\varepsilon}(\relation,A,B)$ is based on Theorem~\ref{thm:transversals}, 
 and proceeds as follows:
\begin{enumerate}
\item Initialize $\bC$ with a single minimal  $A,B$-separator (Line~\ref{algline:addToS0}-\ref{algline:addToS1}). \label{item:init}
\item Iterate over all minimal transversals $D$ of $\bC$ (Line~\ref{algline:enumerateMinTransversals}):
\item If $\comp{D}$ separates $A,B$ (Line~\ref{algline:DCompSeparates}), then: \label{item:isSeparator}
  \begin{enumerate}
  \item Find any minimal $A,B$ separator $X \subseteq \comp{D}$ (Line~\ref{algline:reduceMinSep}).\label{item:reduce}
  \item $\bC \gets \bC \cup \set{X}$.
  \end{enumerate}
\end{enumerate}

\eat{
Unlike the top-down algorithm, $\minSepAlg$~ starts by greedily
discovering one minimal separator.  Then, once it has a collection
$\bC$ of minimal separators, the transversal property allows it to
narrow the search to sets that are guaranteed to contain a new,
not-yet-discovered minimal separator.  As the set $\bC$ grows, so does
the size of the transversals $D$, and the search space $\comp{D}$ for
the next minimal separator is further reduced.
}
The function $\algname{ReduceMinSep}$ called in lines~\ref{algline:addToS0a}
and~\ref{algline:reduceMinSep} takes a separator ($\Omega\sm\set{A,B}$ or $\comp{D}$
respectively) and finds \e{any} subset that is a minimal separator;
this is done greedily in $\algname{ReduceMinSep}$ (Fig.~\ref{alg:ReduceToMinAPSep}).  The
function $\algname{getFullMVDs}$ called in
line~\ref{algline:getFullMVDs:call1} of $\algname{MineMinSeps}$, and in
line~\ref{algline:getFullMVDs:call2} of $\algname{ReduceMinSep}$,
takes as input an attribute set $X$, a pair of attributes $A,B$, and a threshold $\varepsilon$, and computes full
$\varepsilon$-MVDs with key $X$ that separate $A,B$; a parameter $K>0$
is used to limit the number of full MVDs returned, and here we set $K=1$
because we only check if one exists; in
line~\ref{algline:MVDAlgFullMVDs} of the main algorithm
(Fig.~\ref{alg:MVDMiner}) we set $K=\infty$.

\begin{algseries}{h}{Given a set $X \subset \Omega$, and a pair $(A,B)\in \Omega\sm X$, find a subset $S \subseteq X$ s.t. $S$ is a minimal $A,B$-separator in $R$. \label{alg:ReduceToMinAPSep}}
	\small
	\begin{insidealg}{ReduceMinSep}{$\varepsilon$, $X$, (A,B)}
		\STATE Let $p=X_1,\dots,X_m$ be a predefined ordering of $X$.
		\STATE $S \gets X$
		\FORALL{$i=1$ to $m$}		
			\STATE $M_i \gets \algname{getFullMVDs}(S\sm\set{X_i},\varepsilon, (A,B), 1)$\label{algline:getFullMVDs:call2}
			\IF{$M_i \neq \emptyset$}
				\STATE $S \gets S\sm \set{X_i}$ \label{algline:removeAtt}
			\ENDIF
		\ENDFOR
		\STATE return $S$
	\end{insidealg}
	
\end{algseries}

\begin{algseries}{t}{Given a relation $\relation$ with schema $\Omega$, two attributes $A,B\in \Omega$, and a threshold $\varepsilon$ enumerate all minimal $A,B$-separators in $\relation$.\label{alg:MineAllMinSeps}}
	\small
	\begin{insidealg}{\minSepAlg}{$\relation$, $\Omega$, $\varepsilon$, $(A,B)$}		
		\STATE $\bC\gets \emptyset$
		\STATE $X\gets nil$
		\IF{$I(A;B|\Omega\sm\set{A,B}) \leq  \varepsilon$ $\set{\text{by } \entropyAlg}$} \label{algline:addToS0}
			\STATE $X \gets\algname{ReduceMinSep}(\varepsilon,\Omega\sm\set{A,B},(A,B))$ \label{algline:addToS0a}
			\STATE $\bC \gets \bC \cup \set{X}$ \label{algline:addToS1}	
		\ELSE
			\STATE Return $\emptyset$
		\ENDIF		
		\WHILE{$\left(D\gets \algname{nextMinTransversal(\bC)}\right) \neq nil$} \label{algline:enumerateMinTransversals}
			\STATE $\comp{D} \gets \Omega\sm D $				\label{algline:startTransversalLoop}
			\STATE $\phi \gets \algname{getFullMVDs}(\comp{D},\varepsilon,(A,B),1)$\label{algline:getFullMVDs:call1} %\COMMENT{When \algname{getFullMVDs} is called with $K=1$ returns either an empty set of single full MVD}
			\IF{$\phi \neq \emptyset$} \label{algline:DCompSeparates}
			%\STATE $\phi \gets \fullMVDs(\comp{D})[1]$ 
				\STATE $X \gets \algname{ReduceMinSep}(\varepsilon, \comp{D}, (A,B))$ \label{algline:reduceMinSep}
				\STATE $\bC \gets \bC \cup \set{X}$ \label{algline:addToS2}
			\ENDIF			
		\ENDWHILE	\label{algline:endWhile}
		\STATE return $\bC$	
	\end{insidealg}	
\end{algseries}
The only sets of attributes returned in $\algname{MineMinSeps}$ are minimal $AB$-separators returned by $\algname{ReduceMinSep}$ in lines~\ref{algline:addToS0a} and~\ref{algline:reduceMinSep}.
The proof of completeness (i.e., the algorithm returns all minimal $AB$-separators) follows techniques similar to those by Gunopulos
et al.~\cite{DBLP:journals/tods/GunopulosKMSTS03}, and is given in the full
version of the paper:

\def\completenessTheorem{
	Algorithm $\algname{MineMinSeps}$ in Figure~\ref{alg:MineAllMinSeps} enumerates all minimal $A,B$-separators in $R$.	
}
\begin{theorem}\label{thm:completeness}
	\completenessTheorem
\end{theorem}

We now analyze the runtime between consecutive discoveries of minimal $A,B$-separators in $\minSepAlg$.
We let $\Omega$ be a finite set of cardinality $n$, and let $\bC \subseteq 2^\Omega$ be a finite set of sets. The problem of discovering all minimal transversals of $\bC$ is called the \e{hypergraph transversal problem}~\cite{KHACHIYAN20062350}.
The theoretically best known algorithm for solving the hypergraph transversal problem is due to Fredman and Khachiyan~\cite{FREDMAN1996618} and has a quasi incremental-polynomial delay of $poly(n)+m^{O(\log^2m)}$ where $m=|\bC|+n$. Note the dependence on the size of the discovered minimal separators $|\bC|$. We denote by $T_{minTrans}(n,\bC)$ the delay of the minimal transversal algorithm.
However, not every minimal transversal $D$ leads to the discovery of a minimal separator if $\comp{D}$ does not separate $A$ and $B$ (i.e., $\phi=\emptyset$ in line \ref{algline:DCompSeparates} of $\algname{MineMinSeps}$). In the full version of this paper we show that the number of minimal transversals processed in lines \ref{algline:startTransversalLoop}-\ref{algline:endWhile} before a new minimal separator is discovered (e.g., in line~\ref{algline:reduceMinSep}), or before the loop exists, is bounded by $n\cdot |\bC|$. This allows us to formalize the delay between the discovery of minimal $A,B$-separators. We denote by $T(\algname{getFullMVDs})$ the runtime of $\algname{getFullMVDs}$, which we analyze in the next section.

\begin{corollary}\label{corr:Delay}
	Algorithm $\algname{MineAllMinseps}$ enumerates the minimal $A,B$-separators in $R$ with a delay of $O(n\cdot |\bC|\cdot T_{minTrans}(n,\bC)\cdot T(\algname{getFullMVDs}))$, where $n=|\Omega|$.
\end{corollary}

\subsection{Discovering the Full MVDs}\label{sec:getFullMVDs}

Returning to our main algorithm, $\MVDAlg$, we have shown how to compute $\minsep_{\varepsilon}(\relation,A,B)$,
the set of minimal $A,B$ separators in $\relation$.  Next, for each minimal
$A,B$ separator $X\in \minsep_{\varepsilon}(\relation,A,B)$, we 
compute all full MVDs with key $X$ that separate $A$ and $B$, i.e. the
set $\fullMVDs_{\varepsilon}(\relation,X,A,B)$; this is
line~\ref{algline:MVDAlgFullMVDs} of $\MVDAlg$.  Recall that \e{full}
means that the MVD cannot be further refined.

% In this section we present an algorithm for discovering the set
% $\fullMVDs_\varepsilon(S)$ for a given set of attributes $S$.

The algorithm $\algname{getFullMVDs}$ starts by checking the most
refined MVD with key $X$, namely $\varphi=X\mvd Y_1|\dots|Y_n$ where
$Y_1, \ldots, Y_n$ are all attributes not in $X$ (including $A,B$). If
$\J(\varphi) \leq \varepsilon$ then we are done. Otherwise, the
algorithm considers all possible ways to \e{merge} two dependents,
while keeping $A$ and $B$ in different dependents; i.e. it tries
$X\mvd Y_1Y_2|\dots|Y_n, X\mvd Y_1Y_3|Y_2|\dots|Y_n,$ etc. We denote the MVD that results from merging
dependents $Y_i$ and $Y_j$ in $\components(\varphi)$
by $\mergeFunc_{ij}(\varphi)$.  Since $\varphi$
refines $\mergeFunc_{ij}(\varphi)$ then, by
Proposition~\ref{prop:naiveEnumJDs}, it holds that
$\J(\mergeFunc_{ij}(\varphi))\leq \J(\varphi)$.  This procedure for searching for a full $\varepsilon$-MVD can be viewed as a graph traversal algorithm where every
node $\phi$ is an $\varepsilon$-MVD candidate with key $X$, dependents $Z_1,\dots,Z_k$, and its neighbors $\Nbr(\phi)$ are
the $\varepsilon$-MVD candidates:
\begin{equation}\label{eq:tracersalAlgNbrs}
 \Nbr(\phi)\eqdef \set{\mergeFunc_{ij}(\phi) : Z_i, Z_j \in \components(\phi), A,B \notin Z_iZ_j}
\end{equation}
Clearly, if $A,B$ were separated in $\phi$, then they remain separated
in every MVD in $\Nbr(\phi)$. We present the algorithm
as a depth-first traversal, which is how we implemented it. The
pseudocode is presented in Figure~\ref{alg:NaiveMineAllJDs}.

\begin{algseries}{t}{Returns a set of at most $K$ full MVDs with key $S$ that approximately
		hold in $R$ (w.r.t $\varepsilon$) in which $A$ and $B$ are in distinct components. \label{alg:NaiveMineAllJDs}}
	\small
	\begin{insidealg}{getFullMVDs}{$S$, 	$\varepsilon$, $(A,B)$, $K$}
		\STATE $\Pm \gets \emptyset$ \COMMENT{Output set}
		\STATE $\Q \gets \emptyset$ \COMMENT{$\Q$ is a stack}
		\STATE $\phi_0=S\mvd X_1|\dots|X_n$ where $X_i$ are singletons.
		\STATE $\Q.\push(\phi_0)$
		\WHILE {$\Q \neq \emptyset$ \and $|\Pm| < K$}
			\STATE $\varphi \gets \Q.\pop()$
			\STATE Compute $\J(\varphi)$  \COMMENT{using $\entropyAlg$}
			\IF{$\J(\varphi) \leq \varepsilon $}
				\STATE $\Pm \gets \Pm \cup \set{\varphi}$
%				\IF{$|\Pm|=K$}
%				\STATE break
			%\ENDIF
			\ELSE				
				\FORALL{$\phi \in \Nbr(\varphi)$}	
				\STATE $\Q.\push(\phi)$	\COMMENT{See~\eqref{eq:tracersalAlgNbrs}}
		\ENDFOR	
		\ENDIF
		\ENDWHILE		
		\STATE \textbf{return} $\Pm$
	\end{insidealg}
	
\end{algseries}

\subsubsection{An Optimization to $\algname{getFullMVDs}$}
\label{subsubsec:an:optimization}
In the worst case, \eat{if $S$ is \e{not} an $A,B$-separator then} Algorithm~$\algname{getFullMVDs}$ will traverse the  search space of possible ways to partition $n$ attributes into  $k\in \set{2,\dots,n-1}$ sets, and there can be $O(\frac{k^n}{k!})$ such such partitions \footnote{These are \e{Stirling numbers of the second kind}: https://en.wikipedia.org/wiki/Stirling\_numbers\_of\_the\_second\_kind}. While, in general, this is unavoidable, we implemented an optimization, described in the complete version of this paper, that leads to a significant reduction in the search space.

\subsection{Computing Entropies Efficiently}
\label{sec:ComputeEntropies}
We describe the procedure $\entropyAlg$ for
calculating the joint entropy of a set of attributes. The efficiency
of this procedure is crucial to the performance of $\MVDAlg$, which
needs to repeatedly compute mutual information values $I(Y;Z|X)$, and
each such computation requires four entropic values
$H(XY)$, $H(XZ)$, $H(XYZ)$, and $H(X)$.  
Repeatedly computing values of the form $H(X_\alpha)$, for $\alpha \subseteq [n]$
requires multiple scans over the data that resides in external memory.
\eat{GROUP-BY COUNT($*$) query over the data (see Sec.~\ref{sec:it}), which
means access to external memory.} \eat{Since the algorithm computes many such
values, optimizing the entropy calculation is critical to
its performance.}  \eat{ and it is critical to
  optimize this inner loop.  }

We build on ideas introduced in the PLI cache data structure~\cite{DBLP:journals/pvldb/0001N18,DBLP:journals/cj/HuhtalaKPT99}, and reduce the problem of computing $H(X_\alpha)$ to a main memory
join-group-by query.  To describe the algorithm, we repeat here the entropy
formula~\eqref{eq:jointEntropy} for convenience:
\small{
\begin{equation}\label{eq:jointEntropyEficient}
H(X_\alpha)\eqdef \log N-\frac{1}{N}\sum_{x_\alpha {\in} \D_\alpha }|\relation(X_\alpha{=}x_\alpha)\log |\relation(X_\alpha{=}x_\alpha)|
\end{equation}
}
The algorithm uses two ideas: (1) if $x_\alpha$ is a singleton
(i.,e., its frequency $|\relation(X_\alpha{=}x_\alpha)|{=}1$) then it can be ignored
because its contribution to the total entropy
in~\eqref{eq:jointEntropyEficient} is 0 (due to the logarithm), and
(2) given two relations mapping the distinct values of attribute sets $X_\alpha$, and $X_\beta$, respectively, to the tuple ids in the relation $R$ that contain them, 
\eat{
if we store for each computed $H(X_\alpha)$ a relation mapping the
values $x_\alpha$ to the tuple IDs in the relation $R$} then we can
derive this mapping for $X_\alpha \cup X_\beta$ by simply joining the
two mappings on the tuple IDs.  Ignoring singleton valuations makes
these mappings highly compressed, enabling us to store them in main
memory and perform the join using a main memory database system.  We
used the in-memory database $H2$~\cite{h2-database}.  We describe the
details next.
We let $Hash$ denote a hash function. In our implementation we use the
hash function provided by the database system.  
\eat{
Each $x_\alpha$ in
\eqref{eq:jointEntropyEficient} represents a partial tuple,
consisting of strings, integers, etc.  Applying $Hash$ reduces
it to a single long integer.  We can do this because the entropy
$H(X_\alpha)$ does not depend on the actual values $x_\alpha$, only on
their counts.
}

Alg. $\entropyAlg$ maintains two sets of relations indexed by
$\alpha {\subseteq} [n]$: $\CntTbl_\alpha(\valuation, \cnt)$ and
$\TidTbl_\alpha(\valuation,\tid)$ defined as:
\small{
\begin{align*}
&\CntTbl_\alpha{=}\setof{(Hash(x_\alpha), \cnt)}{\cnt = |\relation(X_\alpha=x_\alpha)|, \cnt > 1}\\
&\TidTbl_\alpha{=}\setof{(Hash(x_\alpha), t[\tid])}{t{\in} R, t[X_\alpha]{=}x_\alpha, Hash(x_\alpha) {\in} \proj{\valuation}(\CntTbl_\alpha)}
\end{align*}
}%
\normalsize{
We compute $H(X_\alpha)$ by scanning table $\CntTbl_\alpha$.
%because the sum in Eq.\eqref{eq:jointEntropyEficient} is given by:
}
\eat{
\begin{align*}
& \select\ \sql{sum}(\cnt * \log(\cnt))\text{ } \from \text{ } \CntTbl_\alpha
\end{align*}
}
The algorithm starts by computing two sets of relations: (1)
$\set{\CntTbl_{\set{i}}}$ and (2)~$\set{\TidTbl_{\set{i}}}$
for every $i\in [n]$. 
% In practice, $\valuation$ is the dictionary
% encoding of the corresponding value in $\relation$. 
Assume that we have computed the relations $\CntTbl_\alpha$,
$\CntTbl_\beta$ and $\TidTbl_\alpha$, $\TidTbl_\beta$ for some subsets
$\alpha,\beta \subset [n]$ such that $\alpha \cap \beta =\emptyset$.
We first compute $\CntTbl_{\alpha \cup \beta}$ as:
\begin{align*}
&\select~ Hash(A.\valuation,B.\valuation)~\as~\valuation, \sql{count}(*)~\as~\cnt\\
&\from~\TidTbl_\alpha~A,\TidTbl_\beta~B \\
&\where~ A.\tid=B.\tid\\
&\groupby~Hash(A.\valuation,B.\valuation)~\having~\sql{count}(*)>1
\end{align*}
Next, we compute $\TidTbl_{\alpha\cup \beta}$ as:
\begin{align*}
&\select~Hash(A.\valuation,B.\valuation)~\as~\valuation, A.\tid~\as~\tid \\
&\from~\TidTbl_\alpha~A,\TidTbl_\beta~B,\CntTbl_{\alpha\cup\beta}~Z\\
&\where~ A.\tid=B.\tid~\sql{and}~Hash(A.\valuation,B.\valuation)=Z.\valuation
\end{align*}
Pruning the singleton values makes this technique very effective,
because as we move up the lattice from smaller $\alpha$'s to larger
$\alpha$'s, many more tuples $x_\alpha$ are unique in the data, and
the tables $\CntTbl_\alpha$ and $\TidTbl_\alpha$ become smaller.

\begin{example}
  For a simple illustration, Fig.~\ref{fig:entropyAlgEx} shows the
  tables generated for a $3$-attribute relation $\relation$. Both types of relations only contain values corresponding to
  non-singleton valuations in $\relation$.
\end{example}

However, even with our compression, generating and storing all $2^n-1$
tables $\CntTbl_\alpha$, and $\TidTbl_\alpha$ is intractable.
Instead, we perform the following optimization.  Fix a parameter $L$ (in our implementation we chose $L=10$), and partition the set
$\Omega$ into $\left \lceil \frac{n}{L}\right \rceil$ disjoint subsets
$\Omega_1, \Omega_2, \ldots$ each of size at most $L$.  For each $i$, compute
the tables $\TidTbl_\alpha$ and $\CntTbl_\alpha$ for all subsets
$\alpha \subseteq \Omega_i$; thus the total number of tables
precomputed is $2\left\lceil \frac{n}{L}\right\rceil\cdot2^L$.  In order to compute $H(X_\alpha)$, we express
$\alpha = (\alpha \cap \Omega_1) \cup (\alpha \cap \Omega_2) \cup
\ldots$,
where each union is treated as explained above for
$\alpha \cup \beta$.  \eat{These tables are queried and joined on-demand for computing
the entropy of any set of attributes.}

\begin{figure}[ht]
	\small
\begin{minipage}[b]{1.0\linewidth}\centering
	\begin{minipage}{0.3\linewidth}\centering
	\begin{tabular}{cccc}
    \multicolumn{4}{l}{$\relation$} \\ \hline
	$\tid$ & A & B & C \\
	\hline
	$t_1$ & $a_1$ & $b_2$ & $c_3$\\
	$t_2$ & $a_2$ & $b_1$ & $c_1$\\
	$t_3$ & $a_2$ & $b_2$ & $c_2$\\
	$t_4$ & $a_3$ & $b_3$ & $c_3$\\
	$t_5$ & $a_3$ & $b_3$ & $c_4$
	\end{tabular} \\

	\begin{tabular}{cc}
		\multicolumn{2}{l}{$\CntTbl_{AB}$} \\ \hline
		$\valuation$ & \texttt{CNT} \\
		\hline
		$Hash(a_3,b_3)$ & 2
		%\caption{$\CntTbl_A$}
	\end{tabular} \\

	\begin{tabular}{cc}
		\multicolumn{2}{l}{$\TidTbl_{AB}$} \\ \hline
		$\valuation$ & \texttt{tid} \\
		\hline
		$Hash(a_3,b_3)$ & $t_4$ \\
		$Hash(a_3,b_3)$ & $t_5$
		%\caption{$\CntTbl_A$}
	\end{tabular}

	\end{minipage}
	~
	\begin{minipage}{0.7\linewidth}\centering
	
	\begin{tabular}{cc}
		\multicolumn{2}{l}{$\CntTbl_{A}$} \\ \hline
		$\valuation$ & \texttt{CNT} \\
		\hline
		$a_2$ & 2\\
		$a_3$ & 2
		%\caption{$\CntTbl_A$}
	\end{tabular}
	~
	\begin{tabular}{cc}
		\multicolumn{2}{l}{$\CntTbl_{B}$} \\ \hline
		$\valuation$ & \texttt{CNT} \\
		\hline
		$b_2$ & 2\\
		$b_3$ & 2
		%\caption{$\CntTbl_A$}
	\end{tabular}
	~
	\begin{tabular}{cc}
		\multicolumn{2}{l}{$\CntTbl_{C}$} \\ \hline
		$\valuation$ & \texttt{CNT} \\
		\hline
		$c_3$ & 2 \\
		&
		%	\caption{$\CntTbl_C$}
	\end{tabular}
	\\
	\begin{tabular}{cc}
		\multicolumn{2}{l}{$\TidTbl_{A}$} \\ \hline
		$\valuation$ & \texttt{tid} \\
		\hline
		$a_2$ & $t_2$\\
		$a_2$ & $t_3$ \\
		$a_3$ & $t_4$ \\
		$a_3$ & $t_5$
		%\caption{$\CntTbl_A$}
	\end{tabular}
	~
	\begin{tabular}{cc}
		\multicolumn{2}{l}{$\TidTbl_{B}$} \\ \hline
		$\valuation$ & \texttt{tid} \\
		\hline
		$b_2$ & $t_1$\\
		$b_2$ & $t_3$ \\
		$b_3$ & $t_4$ \\
		$b_3$ & $t_5$
		%\caption{$\CntTbl_A$}
	\end{tabular}
		~
	\begin{tabular}{cc}
		\multicolumn{2}{l}{$\TidTbl_{C}$} \\ \hline
		$\valuation$ & \texttt{tid} \\
		\hline
		$c_3$ & $t_4$\\
		$c_3$ & $t_4$ \\
		& \\
		&			
		%\caption{$\CntTbl_A$}
	\end{tabular}
\end{minipage}
\end{minipage}
\caption{$\entropyAlg$ example.}
\label{fig:entropyAlgEx}
\end{figure}

\begin{algseries}{h}{Generate Acyclic Schemas from $\M_\varepsilon$. \label{alg:asMiner}}
	\small
	\begin{insidealg}{\ASAlg}{$\M_\varepsilon$}
                \STATE $\texttt{schemes} = \emptyset$
		\STATE Construct the graph $G=\setof{(\phi,\psi)}{\phi, \psi \in \M_\varepsilon, \phi \sharp\psi}$
%		\STATE Enumerate maximal independent sets $\Q$ 
		\FORALL{$\Q \in \texttt{MaxIndependentSet}(G)$}
			\STATE  $\texttt{schemes} \gets \texttt{schemes} \cup \set{\algname{BuildAcyclicSchema}(\Q)}$
		\ENDFOR
		\STATE return $\texttt{schemas}$
	\end{insidealg}
\end{algseries}

\begin{algseries}{h}{Gets a set $\Q$ of pairwise compatible MVDs, and
    returns an acyclic schema. \label{alg:getAcyclicSchema}}
	\small
	\begin{insidealg}{BuildAcyclicSchema}{$\Q$}
		\STATE $\schema \gets \set{\Omega}$
		\STATE Sort $\Q$ by ascending order of key cardinality \COMMENT{e.g., $X\mvd A|B$ before $XY\mvd C|D$}
		\FORALL{$\phi \in \Q$}
			\STATE Let $\phi=X\mvd C_1|\dots|C_m$
			\STATE Let $\Omega_i \in \schema$ s.t. $X \subseteq \Omega_i$		
			\STATE $\textbf{D}_\phi \gets \set{C_jX\cap\Omega_i \mid j\in [i,m]}\sm\set{X}$
			\IF{$|\textbf{D}_\phi| \geq 2$} \label{algline:testRedundancy}
				\STATE Replace $\Omega_i \in \bS$ with $\textbf{D}_\phi$ \COMMENT{$\phi$ is \e{non-redundant}}
				\eat{
				\STATE Replace $\Omega_i \in \bS$ with $\set{(C_j\cap \Omega_i)\cup X\mid j\in[1,m]}$
				}
			\ENDIF
		\ENDFOR
		\STATE return $\schema$
	\end{insidealg}
\end{algseries}

\section{Enumerating Acyclic Schemas}\label{sec:EnumerateAcyclicSchemas}

In this section we present the second phase of \system: given the set
$\M_\varepsilon$ of full $\varepsilon$-MVDs (Eq.~\eqref{eq:m}),
generate acyclic $\varepsilon$-schemes.  The algorithm $\ASAlg$ is
shown in Fig.~\ref{alg:asMiner}.  It searches for subsets of MVDs
$\Q \subseteq \M_\varepsilon$, and reconstructs a schema from that
set.  The key to the algorithm's efficiency is our new definition of
compatibility:

% We begin by defining a
% condition named \e{pairwise compatibility} on sets of
% $\varepsilon$-MVDs that make up the support of a join tree
% $(\T,\jointreeMapFunction)$. Pairwise compatibility allows us to
% reduce the problem of enumerating acyclic $\varepsilon$-schemas to
% that of enumerating maximal independent sets in graphs.  Then, we show
% how to synthesize an acyclic schema from a set of pariwise compatible
% $\varepsilon$-MVDs. This gives us the following algorithm for
% enumerating acyclic schemas.
% \begin{itemize}
%         \item[] Enumerate sets $\Q$ of maximal pairwise compatible  MVDs \\
%                 (using~\cite{DBLP:journals/ipl/JohnsonP88,DBLP:journals/jcss/CohenKS08})
% 	\item[] For each such set $\Q$:
% 	\subitem Generate an acyclic schema with support $\Q$
% 	\subitem (i.e., $\algname{BuildAcyclicSchema}(\Q)$).
% \end{itemize}

\begin{definition}\label{def:CompatibleSJDs}
	Let $\phi_1=X \mvd A_1|\dots|A_m$ and $\phi_2=Y \mvd B_1|\dots|B_k$
	be two $\varepsilon$-MVDs.
	We say that $\phi_1$ and $\phi_2$ are \e{compatible} if there exist an $i \in \set{1,\dots,m}$, and $j \in \set{1,\dots,k}$ such that:
	\begin{enumerate}
		\item $Y \subseteq XA_i$, and $X \subseteq YB_j$. In this case we say that the two MVDs are \e{split-free}~\cite{DBLP:journals/tcs/Gucht88,DBLP:journals/tods/FaginMU82,DBLP:journals/tcs/Lakshmanan88,Beeri:1983:DAD:2402.322389}.
		\item There exist two distinct indexes $j_1,j_2\in \set{1,\dots k}$ such that $XA_i \cap B_{j_1} \neq \emptyset$, and $XA_i \cap B_{j_2} \neq \emptyset$. Likewise, there exist two distinct indexes $i_1,i_2\in \set{1,\dots m}$ such that $YB_j \cap A_{i_1} \neq \emptyset$, and $YB_j \cap A_{i_2} \neq \emptyset$.
	\end{enumerate}
        We write $\phi_1 \sharp \phi_2$ to denote the fact that
        $\phi_1, \phi_2$ are \e{in}compatible.
\end{definition}

We say that a set $\Q$ of $\varepsilon$-MVDs is \e{pairwise
  compatible} if every pair of $\varepsilon$-MVDs in $\Q$ is
compatible.  Recall that every join tree $\T$ with $m$ nodes defines a
set of $m-1$ MVDs called its \e{support} and denoted by $\MVD(\T)$.
%is not limited to the pairwise interaction between the MVDs. 

\def\CompatibleCompletenessThm{
	Let $\bS$ be an acyclic schema with join tree $(\T,\jointreeMapFunction)$. 
	Then the set $\MVD(\T)$ is pairwise compatible. 
}
\begin{theorem}\label{thm:equivalenceThm}
\CompatibleCompletenessThm
\end{theorem}

Thus, it suffices to iterate over sets of pairwise compatible
$\varepsilon$-MVDs.  
% Theorem~\ref{thm:equivalenceThm} tells us that every
% $\varepsilon$-schema is supported by a set of pairwise compatible
% $\varepsilon$-MVDs.  Therefore, we enumerate all pairwise compatible
% sets of $\varepsilon$-MVDs, and synthesize an $\varepsilon$-acyclic
% schema from every set generated.
Specifically, our algorithm enumerates the \e{maximal} sets of
pairwise compatible $\varepsilon$-MVDs, and for this task we use a
graph algorithm from the literature.
%  since these lead to a higher degree of
% decomposition. Furthermore, any other $\varepsilon$-schema can be
% derived from one that is maximal by joining one or more relations.
%
% Let $\phi_1,\phi_2 \in \M_\varepsilon$, we denote by $\phi_1
% \sharp\phi_2$ the fact that $\phi_1,\phi_2$ are \e{not} pairwise compatible.
% We define
Define the graph $G(\M_\varepsilon,E)$ as follows:
\begin{equation}
E=\set{(\phi_1,\phi_2): \phi_1,\phi_2\in \M_\varepsilon \text{ and } \phi_1 \sharp \phi_2}
\end{equation}
By this definition every maximal independent set in $G$ corresponds to
a maximal set of pairwise compatible $\varepsilon$-MVDs.  We apply the
following result.
\begin{citedtheorem}{\cite{DBLP:journals/ipl/JohnsonP88,DBLP:journals/jcss/CohenKS08}}\label{thm:enumeration}
	Let $G(V,E)$ be a graph. The maximal independent sets of $G$ can be enumerated such that the delay between consecutive outputs is in $O(|V|^3)$.
\end{citedtheorem}

  \eat{Therefore, we enuemrate acyclic schemas with a delay in $O(\M^3+n^3)$.}
%In our implementation we avoid the recursion.

\eat{
\begin{algseries}{t}{Enumerates all maximal independent sets of a graph $G$. \label{alg:enumerateMaxIndSets}}
	\begin{insidealg}{generateMaxIndSets}{$G, G_{curr}$}		
		\FORALL{$u \in \nodes(G)\sm \nodes(G_{curr})$}
			\FORALL{Maximal independent sets $H$ of $G[\nodes(G_{curr})\cup\set{u}]$}
				\STATE $H^m \gets \algname{ExtendToMaxIndSet}(H,G)$
				\IF{$H^m \notin B$}
					\STATE $B \gets B\cup H^m$
					\STATE $\algname{generateMaxIndSets}(G, H^m)$
				\ENDIF
			\ENDFOR
		\ENDFOR	
	\end{insidealg}	
\end{algseries}
}

\eat{
\subsection{Generating a Single Acyclic Schema} \label{sec:singleSchema}

By Theorem~\ref{thm:approximateAcyclic} the generated acyclic schema is guaranteed to approximately hold in $\relation$ w.r.t. $(m-1)\varepsilon$.
}

In summary, algorithm $\ASAlg$ in Fig.~\ref{alg:asMiner} enumerates all
maximal independent sets $\Q$, then for each of them constructs one
acyclic schema $\schema$, by calling $\algname{BuildAcyclicSchema}$ shown in
Fig.~\ref{alg:getAcyclicSchema}, and described next.
% Next, we present an algorithm for synthesizing an acyclic schema 
% from a set of pairwise compatible $\varepsilon$-MVDs.
% Algorithm \algname{BuildAcyclicSchema} 
% The algorithm receives as input a set $\Q$ of pairwise compatible
% MVDs, and the universal relation $\relation[\Omega]$. 

Algorithm $\algname{BuildAcyclicSchema}$ starts with a
schema that contains a single relation with all attributes
(i.e., $\schema=\set{\Omega}$). It then builds the acyclic schema for
$\relation$ by repeatedly using an $\varepsilon$-MVD from $\Q$ to
decompose one of the relations in $\schema$.  The MVDs are processed
in ascending order of the cardinality of their keys. Therefore, when
an MVD $S \mvd C_1|\dots|C_m$ is processed, then we know that $S$ is
contained in exactly one of the relations in $\schema$ (e.g.,
otherwise, $S$ must be contained in a key of a previously processed
$\varepsilon$-MVD).  The algorithm then applies this $\varepsilon$-MVD
to the single relation that contains it, and continues until all
$\varepsilon$-MVDs in $\Q$ have been processed.  An MVD is said to be
\e{redundant}~\cite{DBLP:journals/ipl/GoodmanT84} if it does not split
the single relation that contains it (i.e., condition of
line~\ref{algline:testRedundancy} does not hold). Redundant MVDs are
simply ignored in $\algname{BuildAcyclicSchema}$.
\eat{
Figure~\ref{fig:buildAcyclicSchemaEx} presents an example adapted from
Fagin's university database~\cite{Fagin:1977:MDN:320557.320571} (e.g.,
$\text{CLASS} \mapsto A$, $\text{SECTION} \mapsto B$, etc.).
}
\begin{theorem}\label{thm:SynthesizeSchema}
Algorithm $\algname{BuildAcyclicSchema}$ generates an acyclic schema $\bS$ with join tree $(\T,\jointreeMapFunction)$ such that $\MVD(\T)\subseteq \Q$. If $\Q$ is a non-redundant set of $\varepsilon$-MVDs then $\MVD(\T)= \Q$. The algorithm runs in time $O(n^3)$.
\end{theorem}

\eat{
\begin{figure}[ht]
	\subfigure[Original relation signature $\Omega$]{
		\includegraphics[width=0.3\linewidth]{acyclicSchema1}
 	}
	\hfill
	\subfigure[Apply $A\mvd J|\Omega\sm\set{J}$]{
		\includegraphics[width=0.5\linewidth]{acyclicSchema2}
	}
	\\
	\subfigure[Apply $C\mvd DF|ABEGHIJKL$]{
			\includegraphics[width=0.45\linewidth]{acyclicSchema3}
	}
	\hfill
	\subfigure[$AB\mvd CDEF|HI|G|J|KL$]{
		\includegraphics[width=0.45\linewidth]{acyclicSchema4}
	}
\caption{Construction of an acyclic schema by Algorithm~$\algname{BuildAcyclicSchema}$ (Figure~\ref{alg:getAcyclicSchema})\label{fig:buildAcyclicSchemaEx}}
\end{figure}
}
%\subsection{Discussion}

The novel insight of our algorithm is the characterization of
(in)compatibility in Definition~\ref{def:CompatibleSJDs}, which
depends only on the pairwise relationship between the MVDs, and
therefore enables the reduction to enumerating maximal independent
sets in graphs. Previous
characterizations~\cite{DBLP:journals/tcs/Gucht88,DBLP:journals/tods/FaginMU82,DBLP:journals/tcs/Lakshmanan88,Beeri:1983:DAD:2402.322389}
are for entire sets of MVDs, and are not pairwise (more precisely, 
they have a different second condition called \e{intersection} which
relies on the existence of a third MVD in the set).  Goodman and Tay~\cite{DBLP:journals/ipl/GoodmanT84} present an algorithm for synthesizing an acyclic schema from a set $\Q$ of MVDs that satisfy the \e{subset property}. As in Theorem~\ref{thm:SynthesizeSchema}, they show that if the set $\Q$ is non-redundant then the synthesized acyclic schema has a join tree whose support is precisely the set $\Q$. However, we are not aware of any characterization of non-redundant MVDs. While the subset property is pairwise, it is applicable only to binary MVDs, while our MVDs may have any number of dependents. 
Algorithms for constructing a (single) acyclic schema from data dependencies have been previously developed by Bernstein~\cite{DBLP:journals/tods/Bernstein76} where the input is a set of functional dependencies, and by Beeri et al. and Lien whose algorithms work by combining \e{conflict-free} MVDs~\cite{Beeri:1983:DAD:2402.322389,Lien:1981:HSR:319540.319546}.

\section{Evaluation}
\label{sec:evaluation}
%

%Table~\ref{tab:evaluationDatasets} presents the evaluation datasets and runtimes for mining full MVDs with threshold $0.0$.

\begin{table}
	\small{
	\begin{tabular}{ccc|cc}		
	\multicolumn{3}{c}{\multirow{2}{*}{Dataset}}	& \multicolumn{2}{c}{Full MVDs} \\ 
	\multicolumn{3}{c}{\multirow{2}{*}{}} & \multicolumn{2}{c}{threshold=0.0}\\
	\hline \\
	Dataset & Cols. & Rows & \eat{$\substack{\text{Size}\\ \text{[KB]}}$ &} $\substack{\text{Runtime} \\ \text{[sec]}}$ & Full MVDs\\
	\hline
	 \textsf{Ditag Feature} & 13 &3960124 & \eat{348331 &} TL & NA\\
	\textsf{Four Square (Spots)} & 15 &973516 &\eat{132791 &} 17017 & 105\\
	\textsf{Image} & 12 & 777676 & \eat{99749 &} 3747 & 151\\
	\textsf{FD\_Reduced\_30} & 30 & 250000 & \eat{69581 &} 8024 & 21\\
	\textsf{FD\_Reduced\_15} & 15 & 250000 & \eat{35157 &} 1006 & 21\\
	\textsf{Census} & 42  & 199524 & \eat{109043 &} TL & NA \\

	\textsf{SG\_Bioentry} & 7 & 184292 & \eat{20719 &} 101 & 3\\
	\textsf{Atom Sites} & 26 & 160000 & \eat{17423 &} TL & 242\\
	\textsf{Classification} & 12 & 70859 & \eat{8 &} 1327 & 27\\
	\textsf{Adult} & 15 & 32561 & \eat{3528 &} 1083 & 58\\
	\textsf{Entity Source} & 33 &26139 &\eat{6402 &} 14155 & 153\\
	\textsf{Reflns} & 27 & 24769 & \eat{3796 &} TL & 543\\
	 \textsf{Letter} & 17 & 20000 & \eat{696 &} 605 & 44\\
     \textsf{School Results} & 27 & 14384 & \eat{ 1782 &} 7202 & 2394\\
    \textsf{Voter State} & 45 &10000&\eat{ 2895 &} TL & 262\\   
    \textsf{Abalone} & 9 & 4177 & \eat{188&} 602 & 36 \\    
  	\textsf{Breast-Cancer} & 11 & 699 & \eat{20 &} 5 & 30 \\
   	\textsf{Hepatitis} & 20 & 155 & \eat{8 &} 479 & 2953\\
	\textsf{Echocardiogram} & 13 & 132 &\eat{ 6 &} 6 & 104 \\	
	\textsf{Bridges} & 13 &108 & \eat{ 6&} 3.8 & 60\\
%	\textsc{horse} & 28 & 300 & f\\
	\end{tabular}
	\caption{{\small  Datasets used in the experiments. We show
            the runtimes (in seconds) for mining full MVDs with
            threshold $0.0$, with a  time limit (TL) of 5 hours.} }
	\label{tab:evaluationDatasets}
}%
\end{table}

In this section we conduct an experimental evaluation of \system.
We start with an end-to-end evaluation of its
  usefulness in Section~\ref{sec:endEndApp}, then evaluate the
  accuracy of the approximate schemas in terms of the relationship
  between the $J$-measure and number of spurious tuples in
  Section~\ref{sec:proofCorrectness}.  Next, we evaluate the
  efficiency and scalability of \system, measuring the time to find
  the minimal separators in Section~\ref{sec:scalability}.  Finally,
  we report the rate of enumeration, and some quality metrics of the generated acyclic schemes in
  Section~\ref{sec:enumerationExperiments}.

We used 20 real-world datsets~\cite{naumann-datasets} that are part of
the Metanome data profiling
project~\cite{Papenbrock:2015:DPM:2824032.2824086}, shown in
Table~\ref{tab:evaluationDatasets} (we discuss the runtimes in
Sec.~\ref{sec:scalability}).  \system\ was implemented in Java 1.8 and
all experiments are conducted on a 64bit Linux machine with 120 CPUs
and 1 TB of memory, running Ubuntu 5.4.0; our algorithm is
single-threaded and runs on a single core.

\eat{
\begin{table*}[!ht]
	\begin{tabular}[t]{|c|c|c|c|c|c|c|c|c|}
		\hline 
		{$A$}&{ $B$}&{ $C$}&{ $D$}&{ $E$}&{ $F$}&{ $G$}&{ $H$}&{ $I$} \\
		\hline 
		{Occupation}&{ $\substack{\text{Child} \\ \text{current}\\ \text{care}}$}& { $\substack{\text{Family}\\\text{ form}}$}&{ \#children}& { Housing} & { $\substack{\text{Financial}\\ \text{standing}}$} & { $\substack{\text{Social}\\ \text{cond.}}$} & { Health} & { $\substack{\text{Application}\\ \text{Priority}}$}\\
		\hline				
	\end{tabular}
	\captionof{table}{\small{Attributes of the \textsf{Nursery} dataset.	
	\label{tab:NurseryFields}
	}}
\end{table*}
}
\begin{figure*}[!ht]
		\begin{tabular}[t]	%{|m{0.15\textwidth}|m{0.15\textwidth}|m{0.15\textwidth}|m{0.16\textwidth}|m{0.16\textwidth}|m{0.16\textwidth}|}
		{|c|c|c|c|c|c|}
		\hline		
		\subfigure[]{\includegraphics[width=0.09\textwidth]{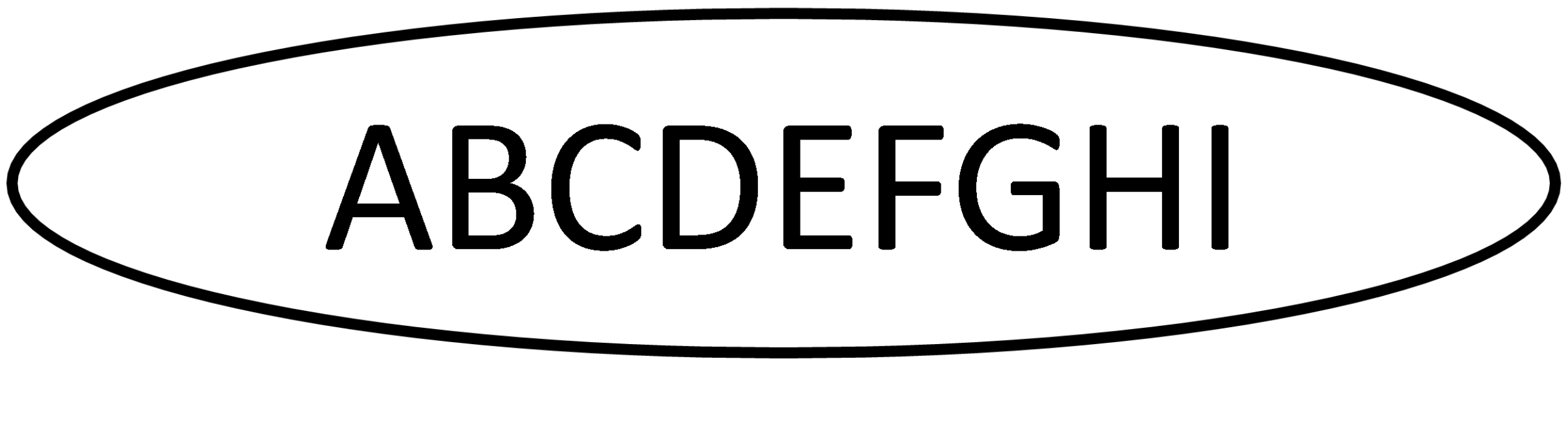}\label{chart:B2}}&
		\subfigure[]{\includegraphics[width=0.09\textwidth]{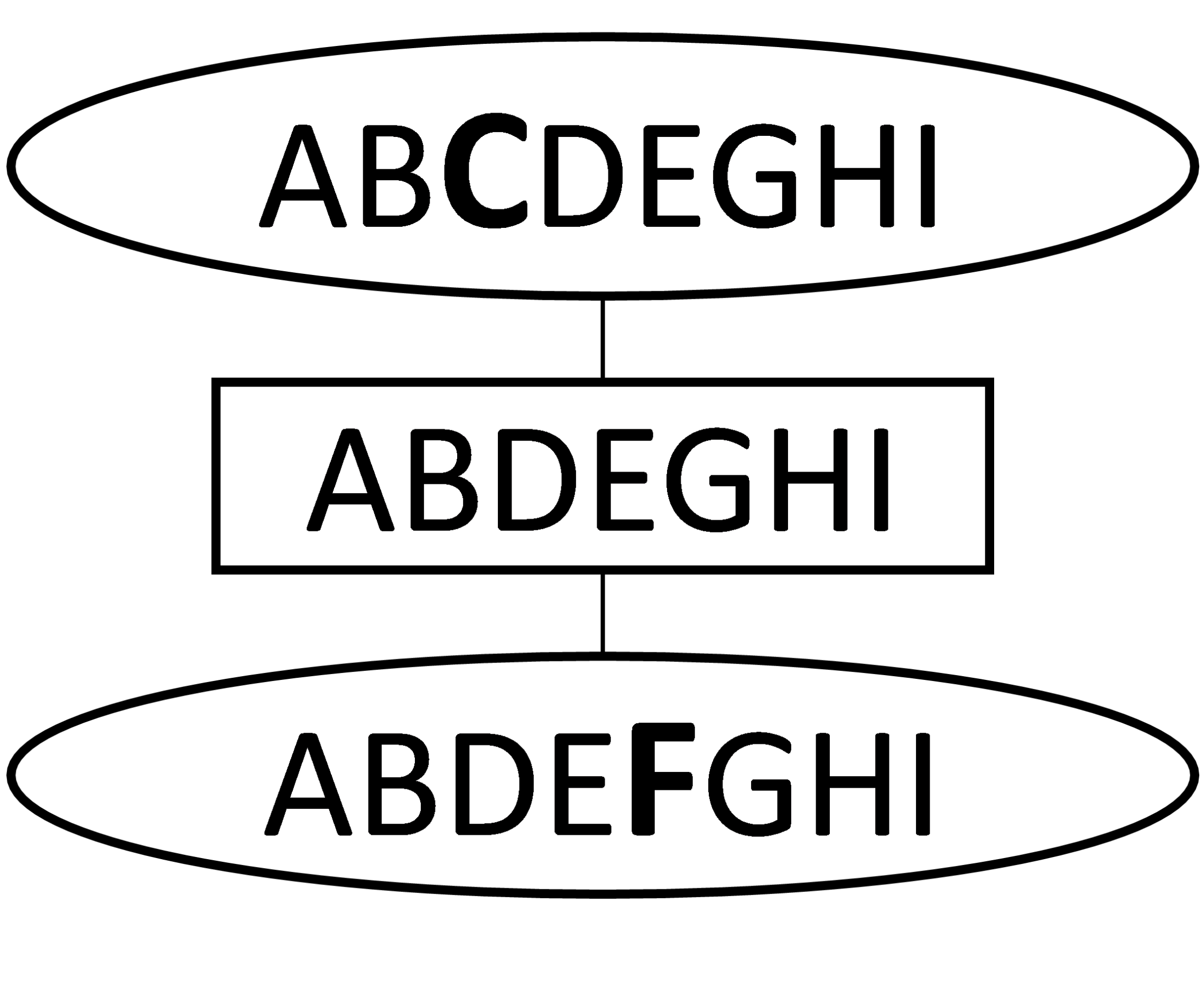}\label{chart:B2}}&
		\subfigure[]{\includegraphics[width=0.09\textwidth]{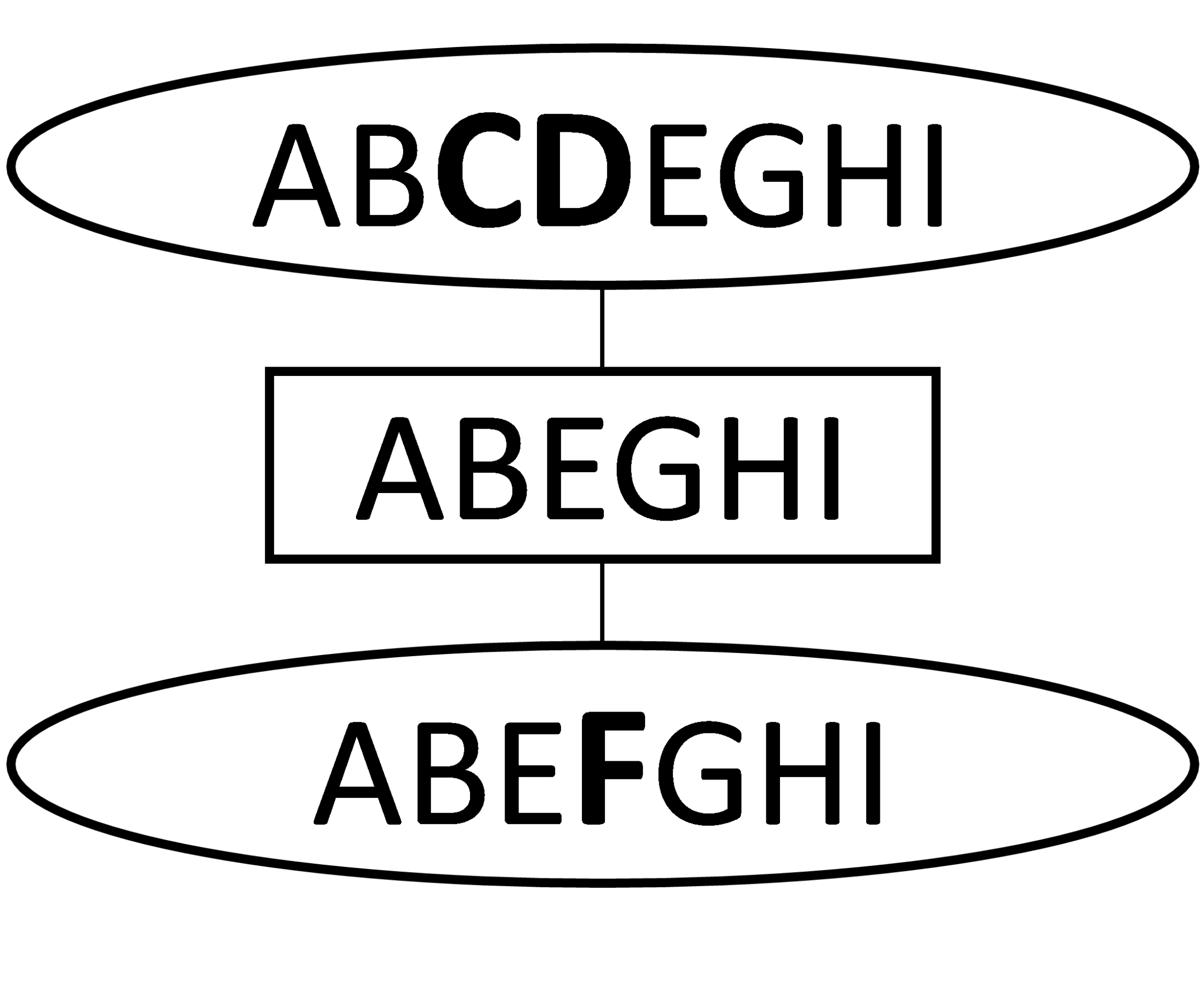}\label{chart:B2}}&
		\subfigure[]{\includegraphics[width=0.18\textwidth]{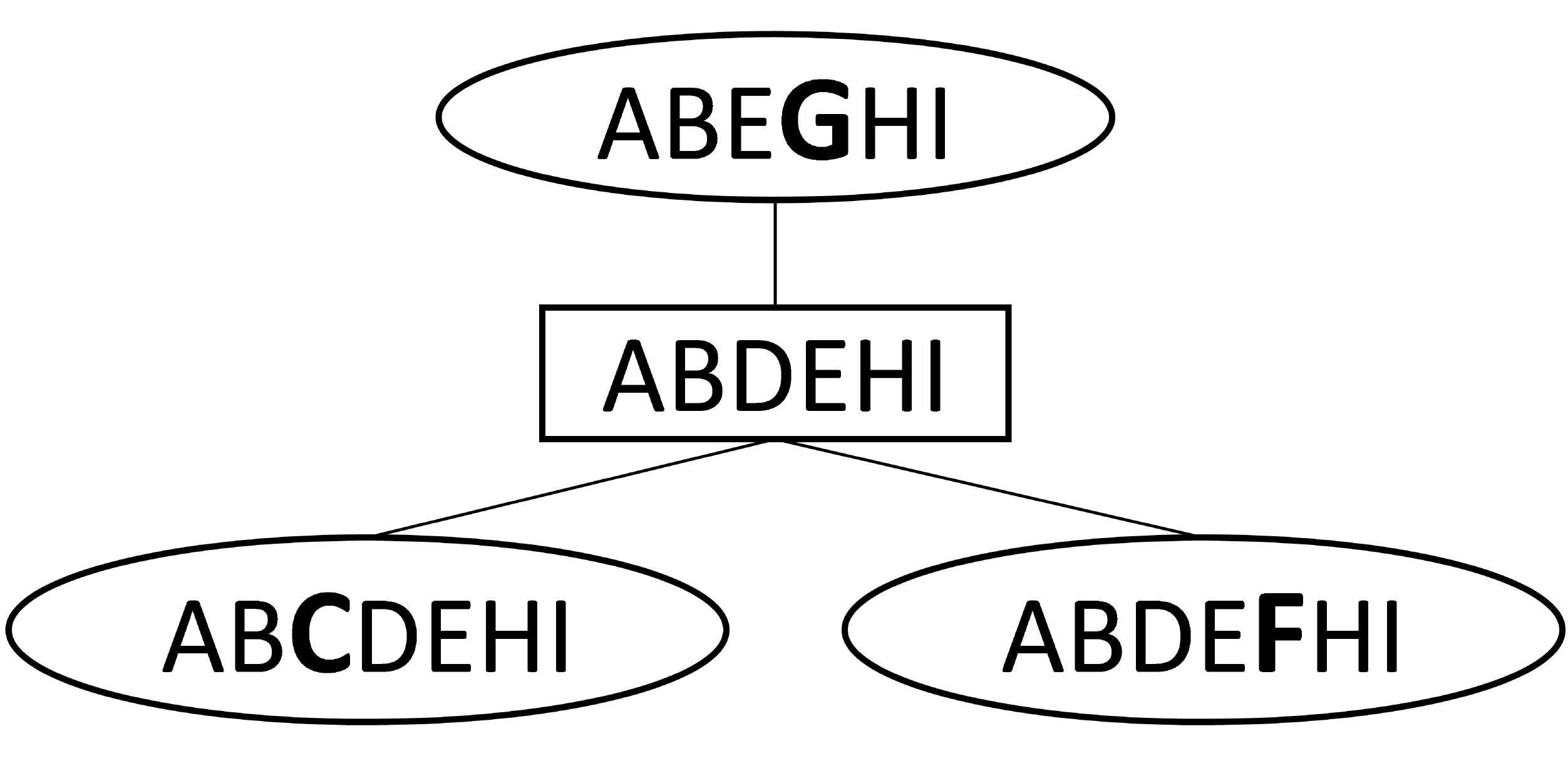}\label{chart:B2}}&
		\subfigure[]{\includegraphics[width=0.18\textwidth]{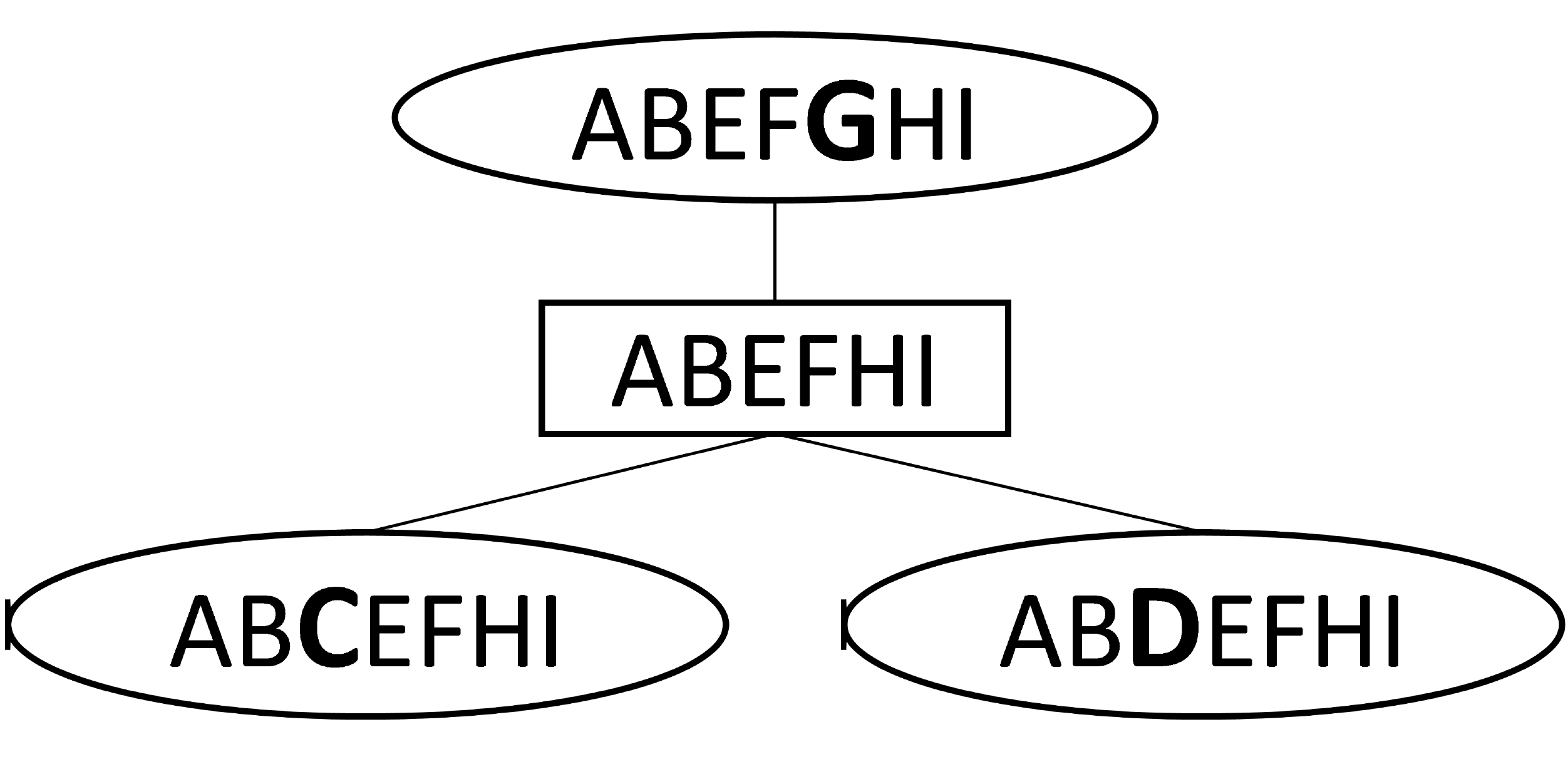}\label{chart:B2}}&
		\subfigure[]{\includegraphics[width=0.18\textwidth]{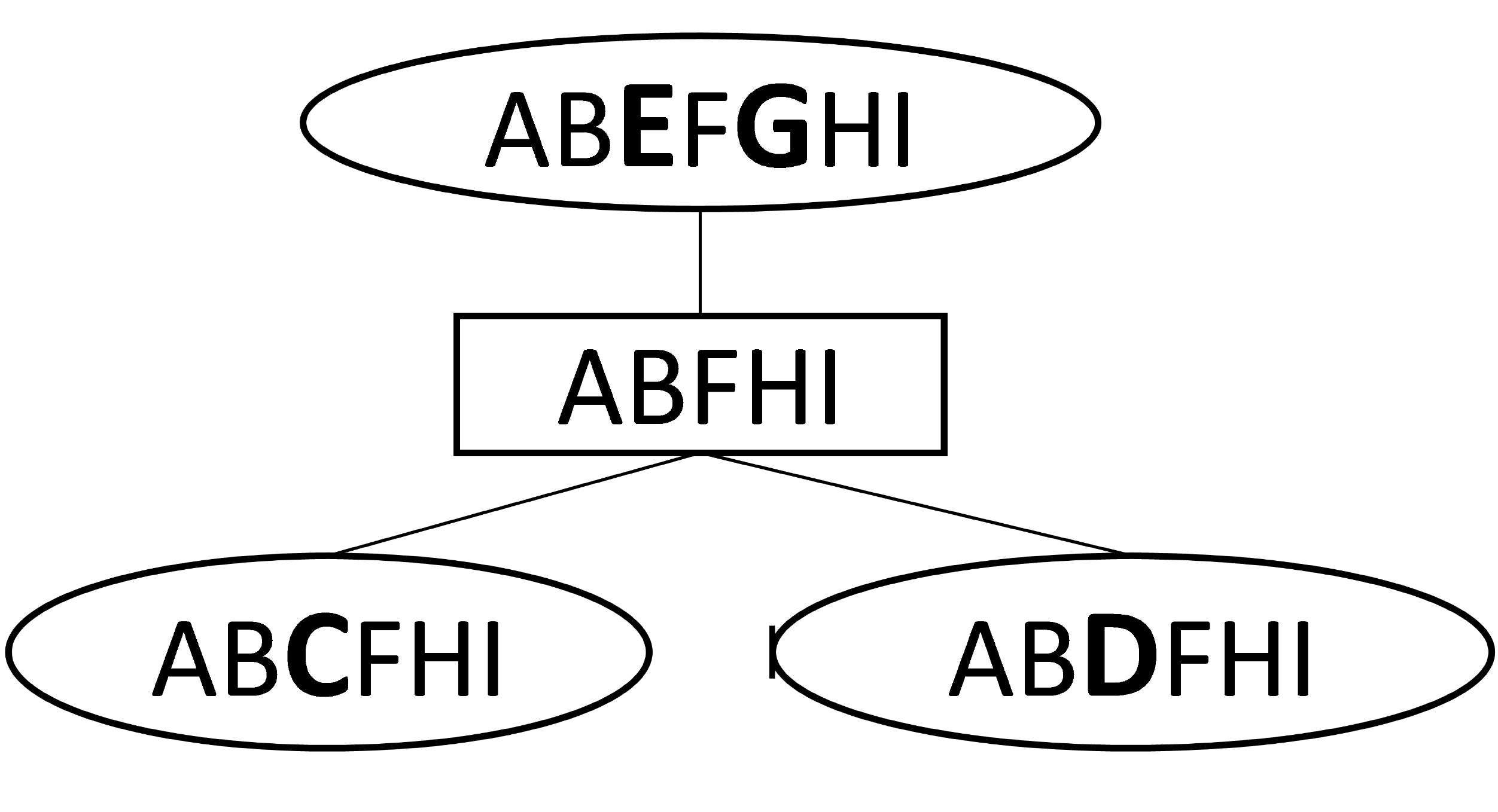}\label{chart:B2}}
		\\
		{ $\substack{J=0,S=0,\\E=0\%,m=1}$}& 	{ $\substack{J=0.009,S=28\%,\\ E=1.08\%,m=2}$}&	{ $\substack{J=0.021,S=46\%,\\E=3.42\%,m=2}$}&
			{ $\substack{J=0.044,S=65\%,\\E=7.62\%,m=3}$}&
				{ $\substack{J=0.062,S=78\%,\\ E=8.61\%,m=3}$}&
					{ $\substack{J=0.097,S=89\%,\\E=16.48\%, m=3}$}\\
		
		\eat{
		\multicolumn{2}{|c|}{
			\subfigure[]{\includegraphics[width=0.15\textwidth]{pics/NurseryExperiment/T1n.pdf}\label{chart:B2}}}&
		\multicolumn{2}{|c|}{
			\subfigure[]{\includegraphics[width=0.15\textwidth]{pics/NurseryExperiment/T2n.pdf}\label{chart:B2}}}&
		\multicolumn{2}{|c|}{
			\subfigure[]{\includegraphics[width=0.15\textwidth]{pics/NurseryExperiment/T3n.pdf}\label{chart:B2}}}\\	
		
		\multicolumn{2}{|c|}{
				{ $\substack{J=0,S=0,\\E=0\%,m=1}$}
		}
		&
		\multicolumn{2}{|c|}{
			{ $\substack{J=0.009,S=28\%,\\ E=1.08\%,m=2}$}
		}
		&
		\multicolumn{2}{|c|}{
			{ $\substack{J=0.044,S=65\%,\\E=7.62\%,m=3}$}
		}\\
	}
			\eat{
		\hline
		\multicolumn{2}{|c|}{
			\subfigure[]{\includegraphics[width=\eat{0.35}0.2\textwidth]{pics/NurseryExperiment/T4n.pdf}\label{chart:B2}}}&
		\multicolumn{2}{|c|}{
			\subfigure[]{\includegraphics[width=\eat{0.35}0.2\textwidth]{pics/NurseryExperiment/T5n.pdf}\label{chart:B2}}}&
		\multicolumn{2}{|c|}{
			\subfigure[]{\includegraphics[width=\eat{0.35}0.2\textwidth]{pics/NurseryExperiment/T6n.pdf}\label{chart:B2}}}\\
		
		\multicolumn{2}{|c|}{
				{ $\substack{J=0.17,S=94\%,\\E=26.6\%,m=3}$}
		}
		&
		\multicolumn{2}{|c|}{
			{$\substack{J=0.277,S=95.7\%,\\ E=26.8\%,m=4}$}
		}
		&
		\multicolumn{2}{|c|}{
			{ $\substack{J=0.345,S=97.4\%,\\ E=45.2\%,m=4}$}
		}			
	\\
}
	\hline
		\subfigure[]{\includegraphics[width=0.1\textwidth]{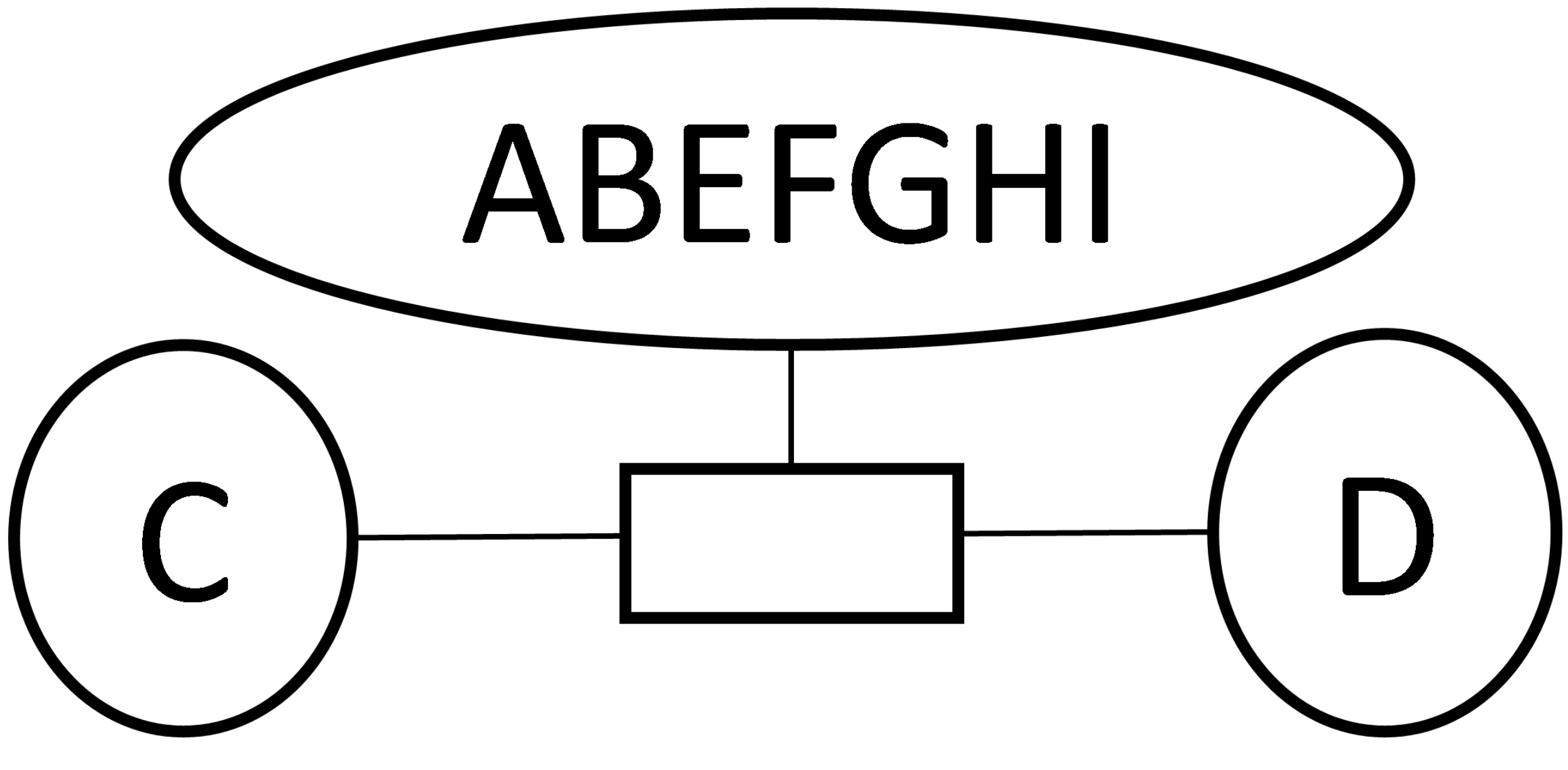}\label{chart:B1}}& 
	\multicolumn{2}{|c|}{
		\subfigure[]{\includegraphics[width=0.25\textwidth]{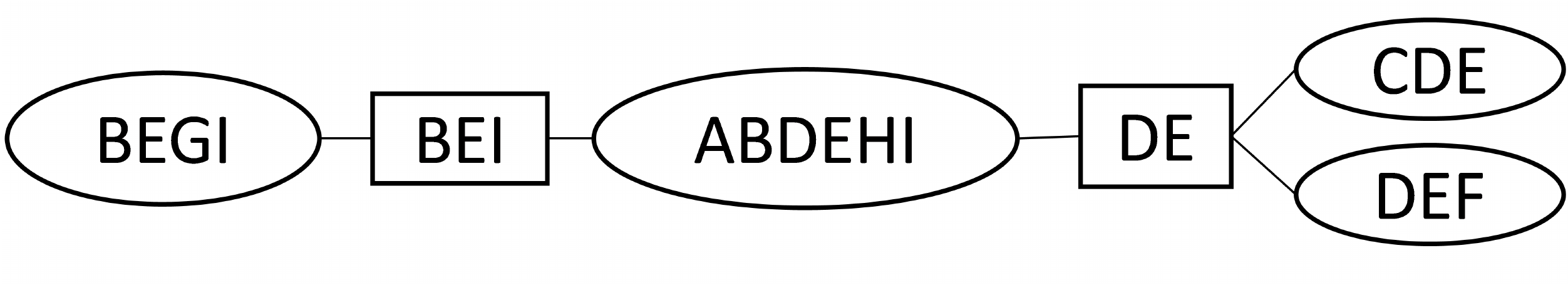}\label{chart:B2}}}& 
	\subfigure[]{\includegraphics[width=0.12\textwidth]{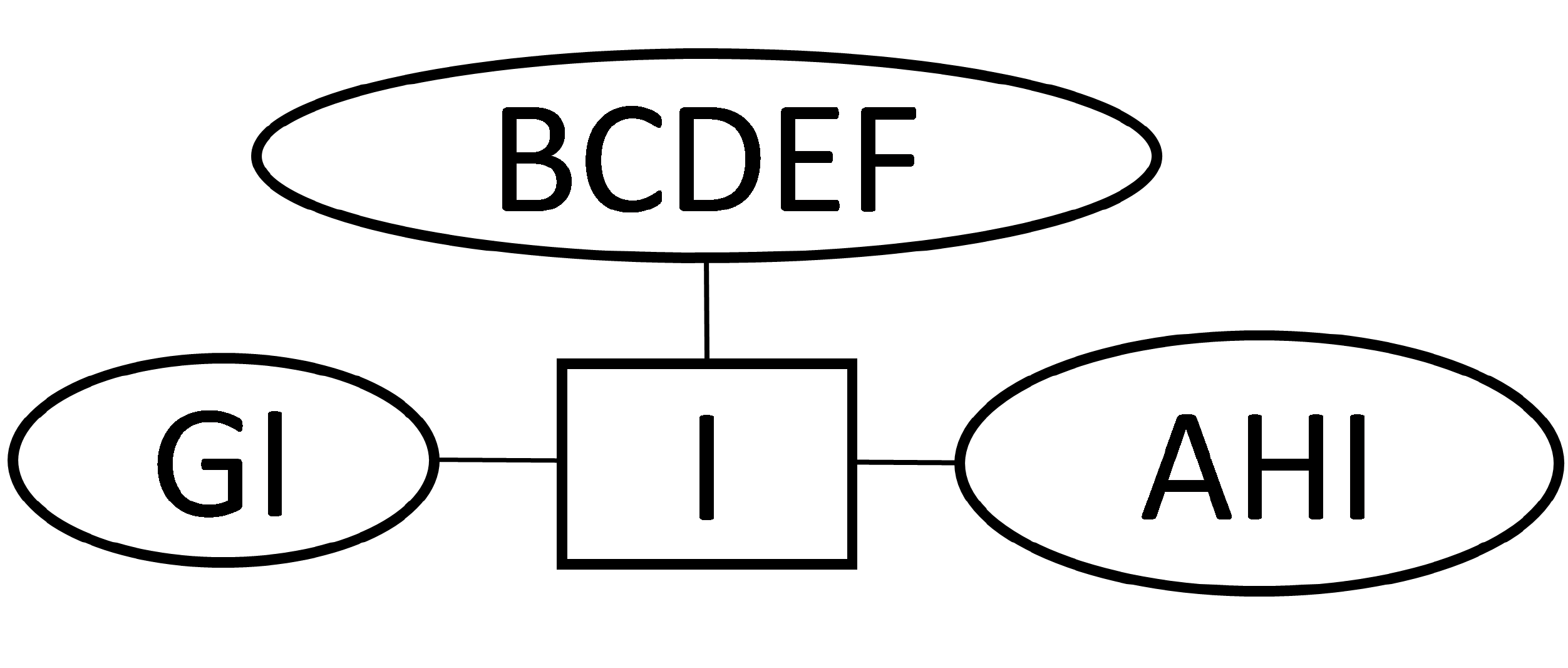}\label{chart:B3}}& 
	\multicolumn{2}{|c|}{
		\subfigure[]{\includegraphics[width=0.25\textwidth]{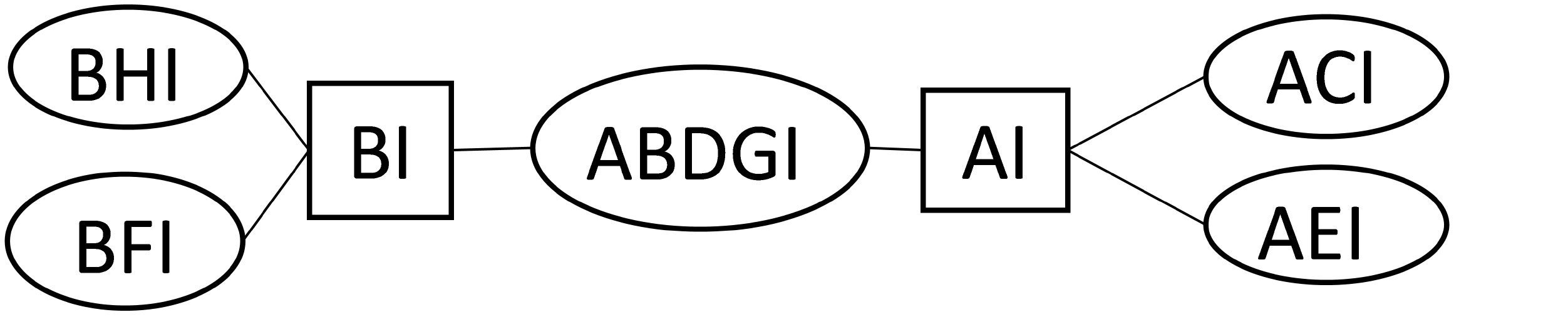}\label{chart:B4}}}\\
	{$\substack{J=0.17,S=94\%,\\E=26.6\%,m=3}$}&
	\multicolumn{2}{|c|}{{$\substack{J=0.277,S=95.7\%,\\ E=26.8\%,m=4}$}}&
	{$\substack{J=0.33,S=92.6\%,\\E=51.4\%,m=3}$}&
	\multicolumn{2}{|c|}{{$\substack{J=0.345,S=97.4\%,\\ E=45.2\%,m=4}$}}\\					
	\hline					
	\end{tabular}
	\captionof{figure}{\small{The \textsf{Nursery} use case, showing the 10 pareto optimal schemes (out of 415).  We encode the 9 attributes as $A, B, \cdots, I$ (top).  The data does not admit a exact decomposition (a), but we obtain increasingly better schemes (b)-(j) as we increase the $J$-measure, with increased space savings $S$, at the cost of increased rate of spurious tuples $E$; for example, for $J=0.277$ the data decomposes into 4 relations, $S= 95.7\%$ (see text for the explanation of why it is so high) and $E = 26.8\%$.	
	}}
	\label{fig:NurserySchemes}
\end{figure*}

\begin{figure}
	\begin{center}
	\includegraphics[width=0.35\textwidth]{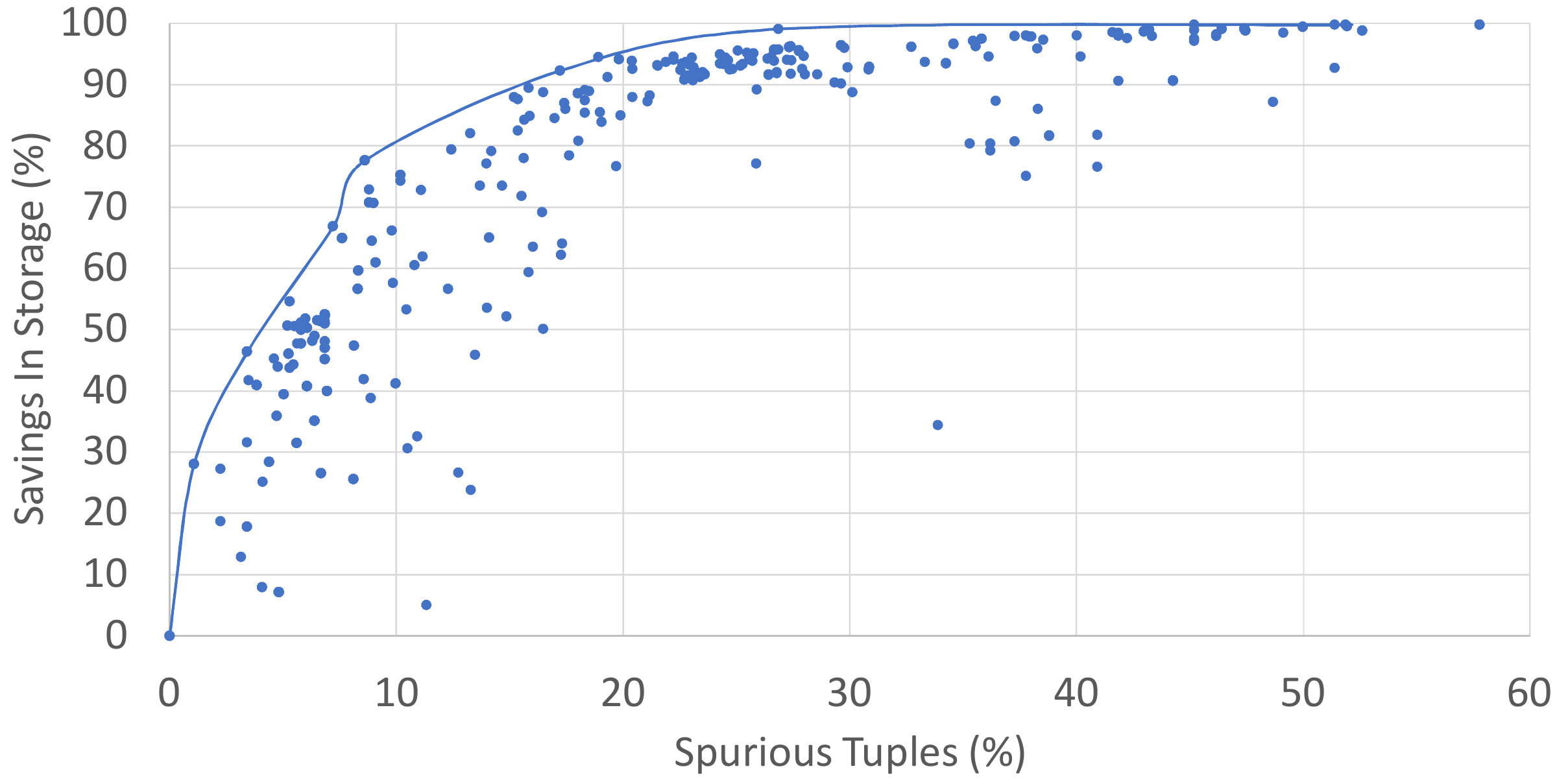}%	
	\end{center}
	\captionof{figure}{\small{All 415 schemes discovered for
            \textsf{Nursery}. The plot shows the  savings $S$ v.s. the
             spurious tuples $E$.  The line connects the ten
             pareto-optimal schemes further
             detailed in Fig.~\ref{fig:NurserySchemes}. }.	\label{fig:NurseryChart}}
\end{figure}

  \subsection{A Use Case: \textsf{Nursery}} \label{sec:endEndApp}

  To evaluate the usefulness of \system\ we applied it to the
  \textsf{Nursery}
  dataset\footnote{https://archive.ics.uci.edu/ml/datasets/nursery}, a
  training data for classifying and ranking applications for nursery
  schools.  The dataset contains eight attributes describing
  occupational, financial, social and health conditions of the family,
  and a classification attribute that indicates the priority of the
  application; we renamed the attributes $A\ldots I$ for brevity.  The
  data has 12960 tuples and
  \eat{ 9 attributes (Table~\ref{tab:NurseryFields}) with}
  a total of $12960*9=116640$ cells.  By increasing the threshold $J$
  from $0$ to $0.5$, we found $415$ acyclic schemes
  (Fig~\ref{fig:NurseryChart}), and show ten of them in detail in
  Fig.~\ref{fig:NurserySchemes}.  As one can see in
  Fig.~\ref{fig:NurserySchemes}(a), when $J=0$, no exact decomposition
  is possible; a traditional (exact) decomposition of this data is not
  possible.  As we increase $J$, however, we find better and better
  schemas in Fig.~\ref{fig:NurserySchemes} (b)-(j), in the sense that
  it decomposes into more relations, each with fewer attributes.  For
  example, the schema in (h) ($J=0.277$) has 4 relations, $BEGI$,
  $ABDEHI$, $CDE$, $DEF$.  For each scheme we report the percentage
  cell savings, $S$, and the percentage of spurious tuples, $E$.
  There is a good tradeoff between space savings and error rate:
  several schemes have under $10\%$ spurious tuples yet achieve over
  $80\%$ space saving.  The space savings are very high (e.g. over
  $90\%$), because the \textsf{Nursery} data is dense: the attribute
  domains have sizes $3,5,4,4,3,2,3,3,5$.  For example, the extreme
  schema where each attribute is a separate relation (not shown in the
  Figure) has $3+5+4+4+3+2+3+3+5=32$ cells and a savings of
  $(116640-32)/116640$ i.e.  $S = 99.9725\%$; however, its fraction of
  spurious tuples is $(3*5*4*4*3*2*3*3*5 - 12960)/12960=4$, i.e.
  $E= 400\%$.  Fig.~\ref{fig:NurseryChart} shows the values $S,E$ for
  all $415$ schemes.  Users are likely to select the pareto optimal
  schemes, i.e. whose $S, E$ values are not dominated by any other
  schemes: the ten pareto optimal schemes in this graph are connected
  by a line, and are precisely those we have selected to show in
  detail in Fig.~\ref{fig:NurserySchemes}.  In addition to savings $S$
  and spurious tuples $E$, applications are likely to define their own
  domain specific quality measure and choose the optimal schema for
  that application.

\begin{figure}[t]
  %  \includegraphics[width=0.014\textwidth]{pics/SpriousTuplesLegend}~
  %  \vline~~
  %	\begin{minipage}[b]{1.0\textwidth}\centering
  		\eat{
	\begin{minipage}{0.017\textwidth}\flushleft
    	\includegraphics[width=1.0\textwidth]{pics/SpriousTuplesLegend}       	
    \end{minipage}
    ~
		}
	\subfigure[$\textsf{BreastCancer}$]{\includegraphics[width=0.23\textwidth]{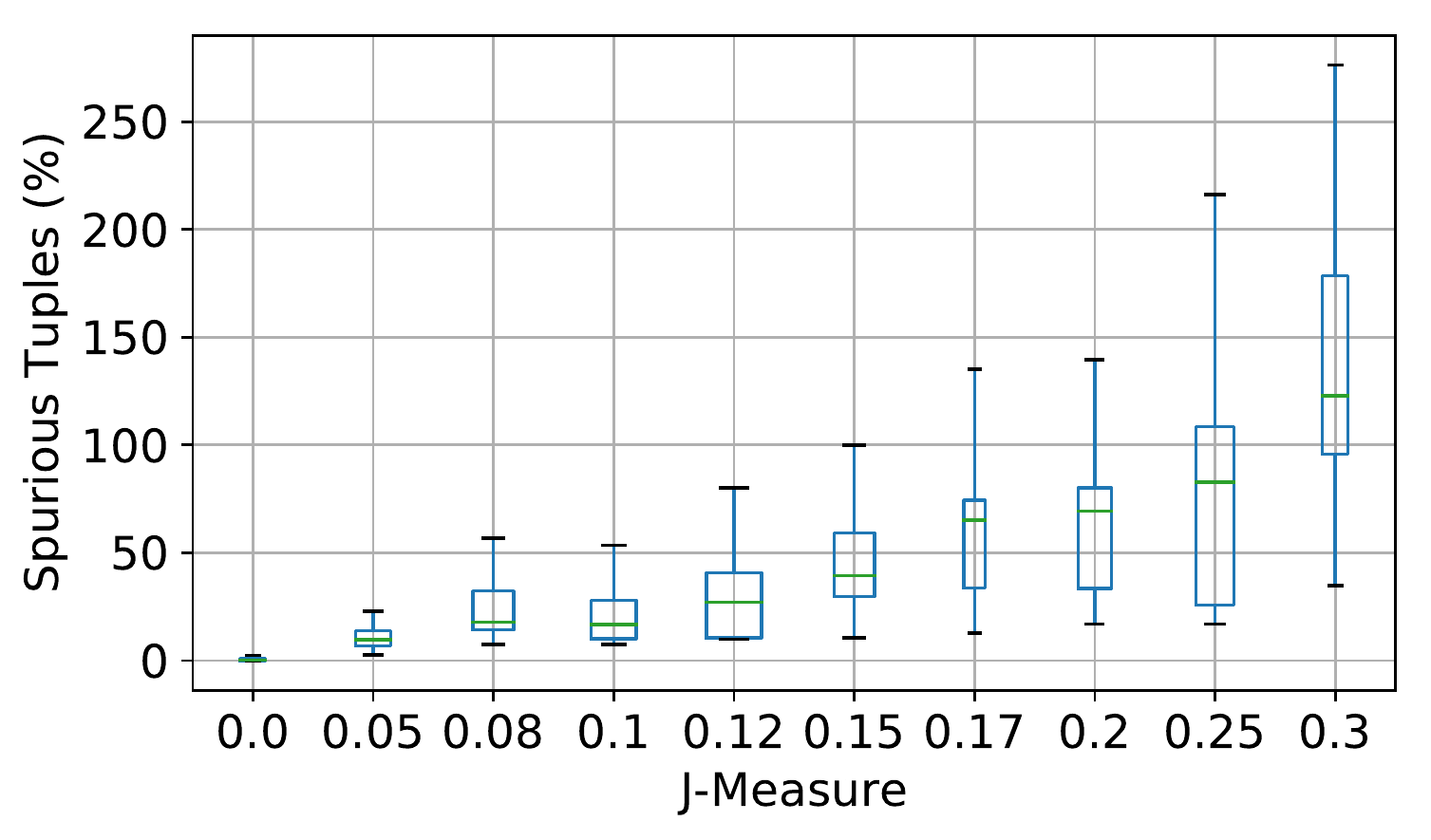}}
	\hfill
	\subfigure[$\textsf{Bridges}$]{\includegraphics[width=0.23\textwidth]{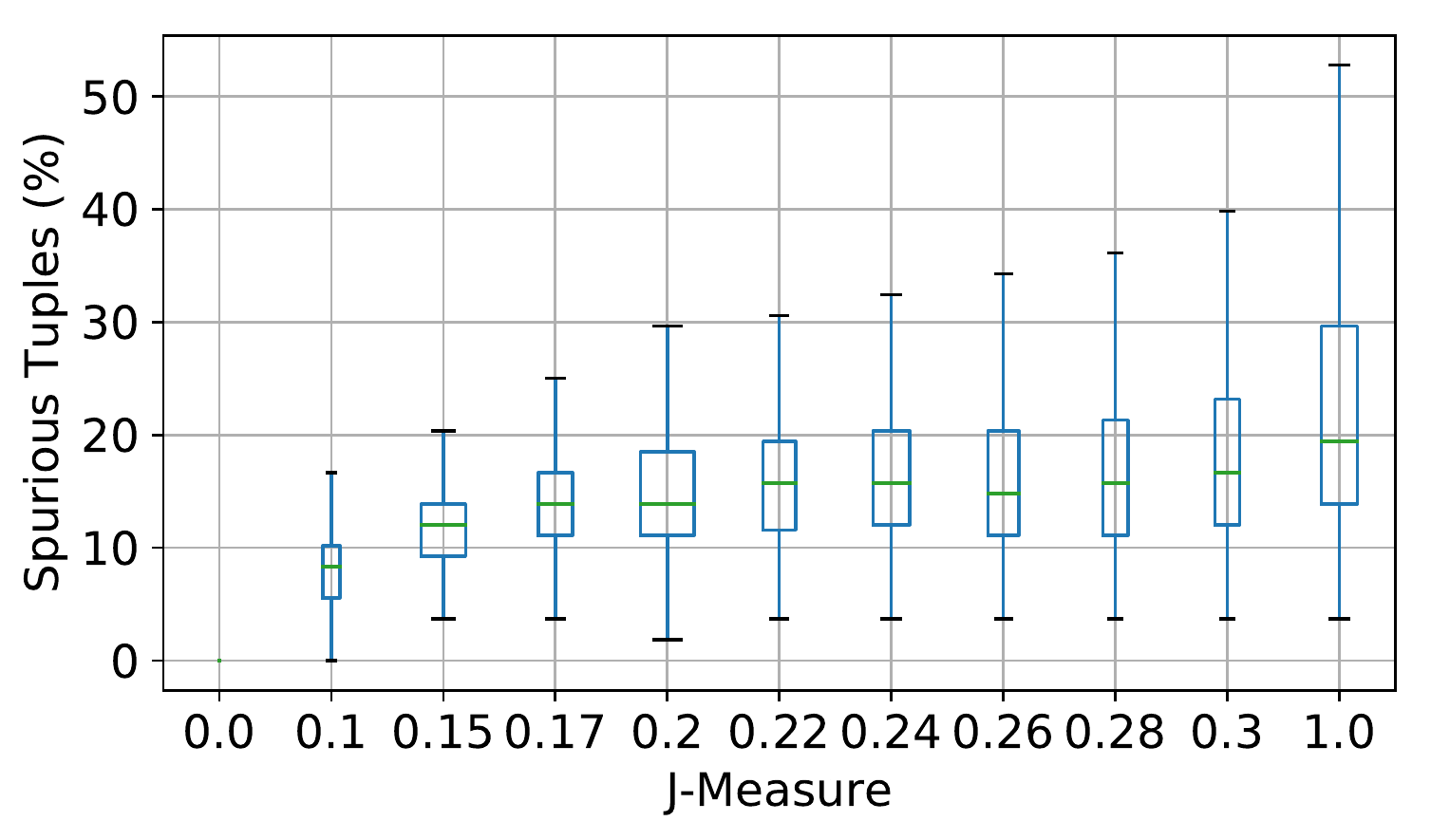}}
	\\
	\subfigure[$\textsf{Nursery}$]{\includegraphics[width=0.23\textwidth]{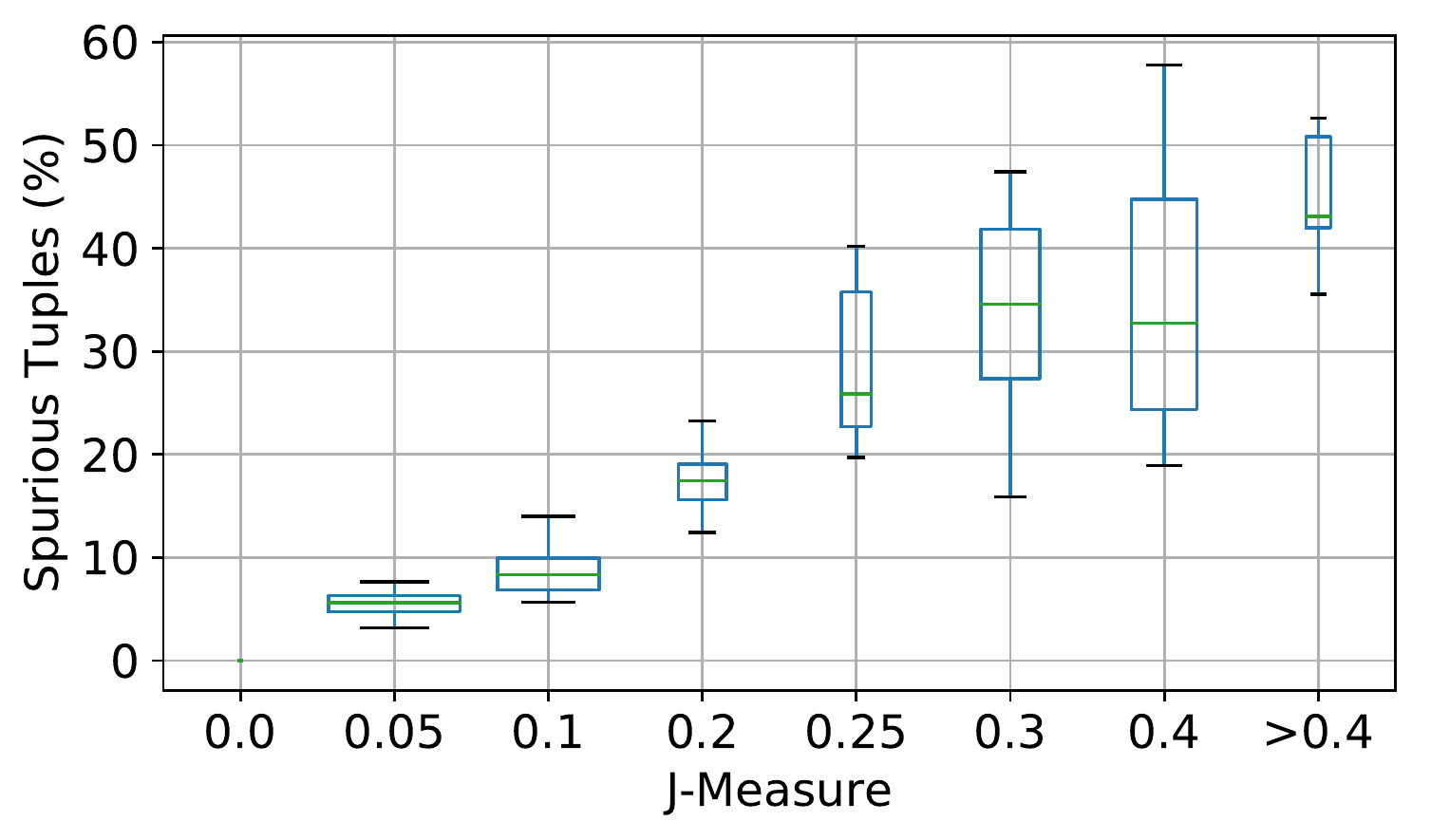}}
	\hfill
	\subfigure[$\textsf{Echocardiogram}$]{\includegraphics[width=0.23\textwidth]{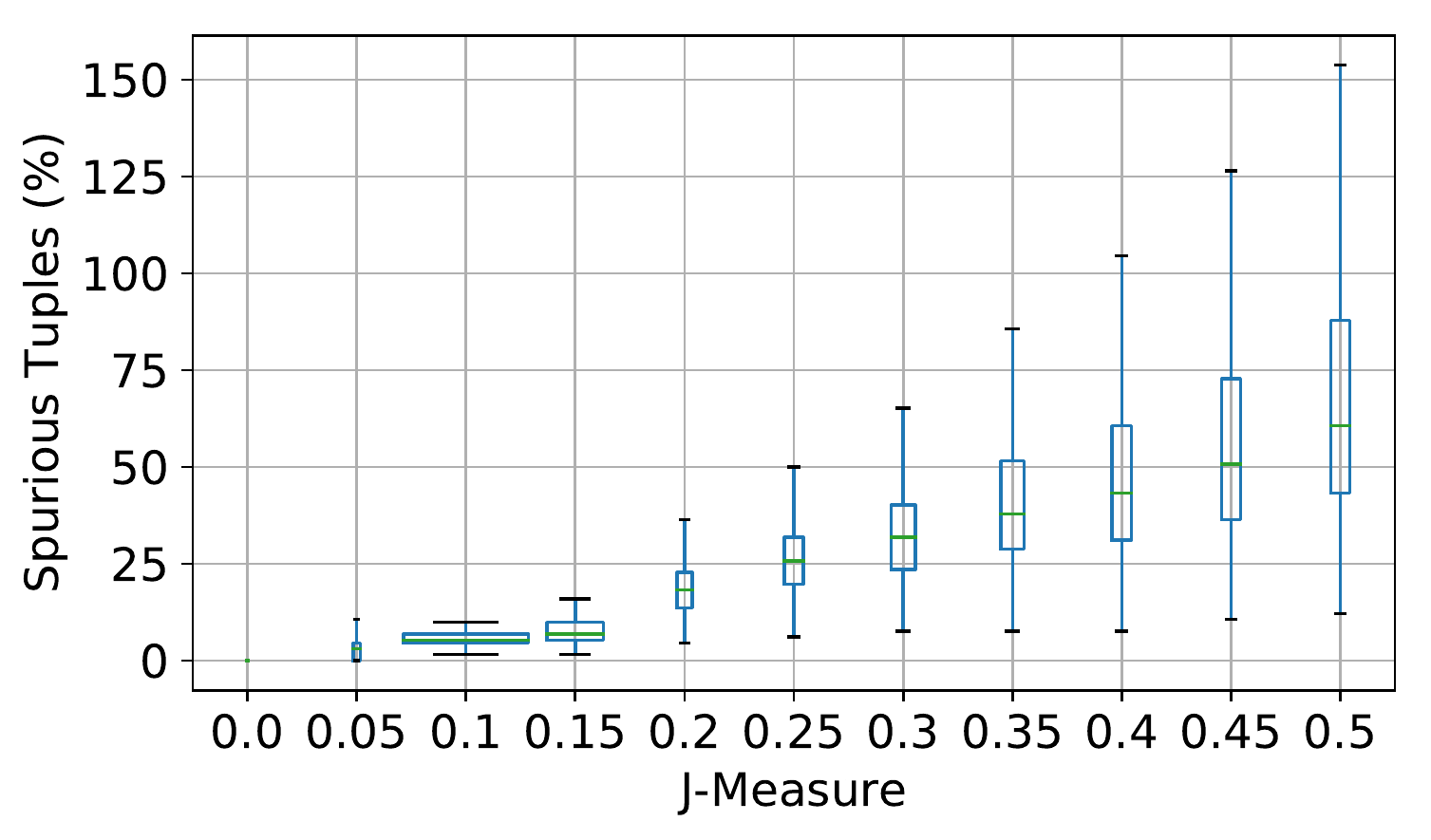}}
	
	\eat{
	\subfigure[$\textsf{Letter}$]{\includegraphics[width=0.135\textwidth]{pics/spuriousTuplesLetter}}
	\qquad
	\subfigure[$\textsf{Bridges}$]{\includegraphics[width=0.135\textwidth]{pics/spuriousTuplesBridges}} \\
	\subfigure[$\textsf{Echocardiogram}$]{\includegraphics[width=0.135\textwidth]{pics/spuriousTuplesEcho}}
	\qquad
	\subfigure[$\textsf{Abalone}$]{\includegraphics[width=0.135\textwidth]{pics/spuriousTuplesAbalone}}\hfill	
}
 %   \end{minipage}
%	\end{minipage}
	\caption{\small Spurious Tuples (\%) vs. J-measure (see Sec. \ref{sec:proofCorrectness}). \eat{	Width of the box-plots represent the number of schemes whose $J$-measure is in the appropriate range }
	}
\eat{
	\caption{{\small The $Y$-axis shows the percentage of spurious tuples. The width of the box-plots represent the number of schemes whose $J$-measure is in the appropriate range (e.g., in \textsf{Nursery} there are are more schemes in the range $(0,0.05]$ than in $(0.3,0.4]$).}
	\eat{For every threshold we show the range of these percentages over all acyclic schemas generated for the datasdet in that threshold (e.g., in \textsf{Letter} the percentage of spurious tuples for $\varepsilon=0.02$ is between $0\%$ (minimum), and $5\%$ (maximum)).}
	}
}
	\label{fig:spuriousTuplesChart}
\end{figure}

\normalsize{
\subsection{Accuracy}\label{sec:proofCorrectness}

Next, we analyzed the relationship between the
  $J$-measure of the acyclic schemes, and the percentage of spurious
  tuples.  There is no tight theoretical connection between these two
  measures, except that $J{=}0$ iff there are no spurious tuples, hence
  the need for an empirical evaluation.  The results are presented in
  Figure~\ref{fig:spuriousTuplesChart}.  We generated all acyclic
  schemes with a threshold $\varepsilon \in [0, 0.5]$, partitioned the
  schemes into buckets according to their $J$-measure, and report the
  quantiles of the number of spurious tuples in each bucket.  The
  experiments confirm a consistent relationship between the
  $J$-measure and the percentage of spurious tuples.  Assuming we want
  to have no more than $20\%$ spurious tuples, then we can increase
  $J$ up to $0.1{-}0.3$, depending on the dataset.  The width of the
  boxes represent the number of acyclic schemes in that bucket.  In
  general, as $J$ increases, the number of acyclic schemes will
  eventually decrease: this is particularly visible in
  Fig.~\ref{fig:spuriousTuplesChart} (d).  The explanation lies in the
  fact that larger $J$'s reduce the size (and, hence, the number) of
  minimum separators.  If we allowed $J$ to increase further,
  eventually we find a single schema, where each attribute is a
  separate relation, and where the sole minimal separator is the empty
  set.

}

\begin{figure*}[ht]
	\begin{tabular}{ccc}
	\subfigure[$\textsf{Image}$]{\includegraphics[width=0.33\textwidth]{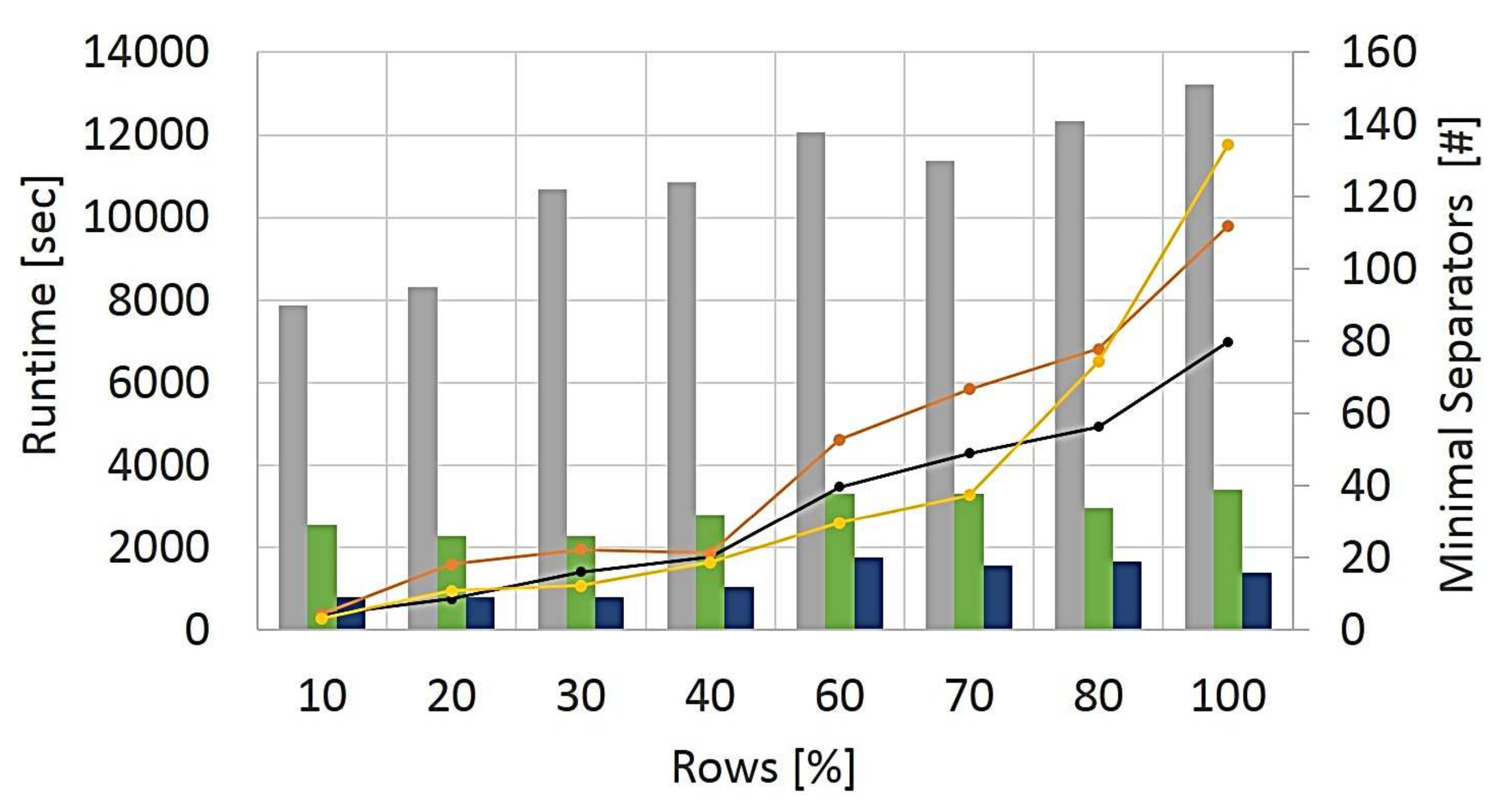}}& 
    \subfigure[$\textsf{Spots} $]{\includegraphics[width=0.32\textwidth]{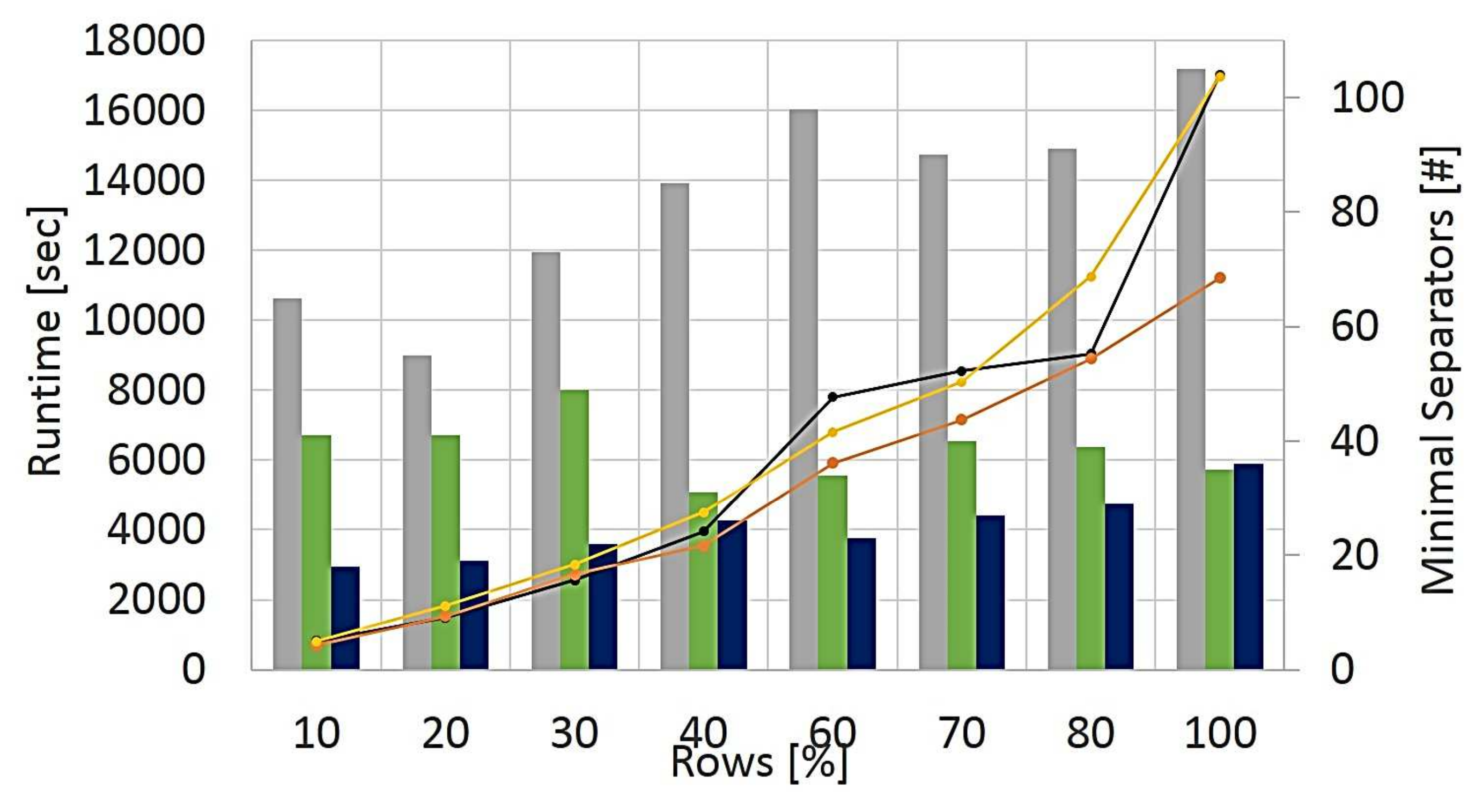}}&
    \subfigure[$\textsf{Ditag Feature}$]{\includegraphics[width=0.33\textwidth]{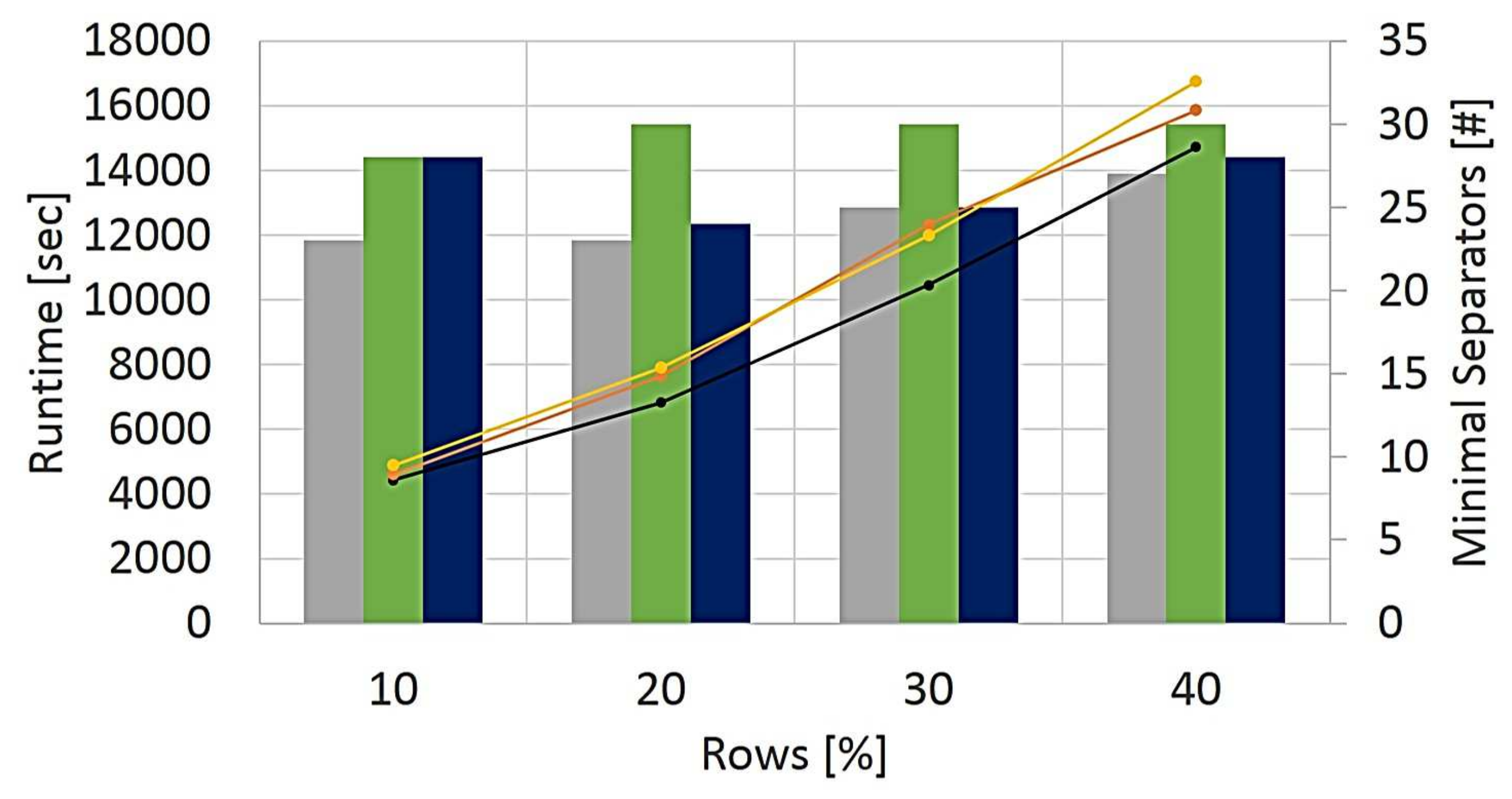}}\\	
    \includegraphics[width=0.12\textwidth]{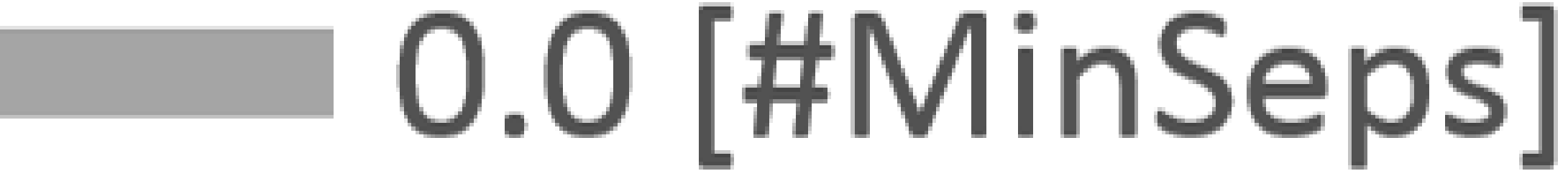}\hfill\includegraphics[width=0.12\textwidth]{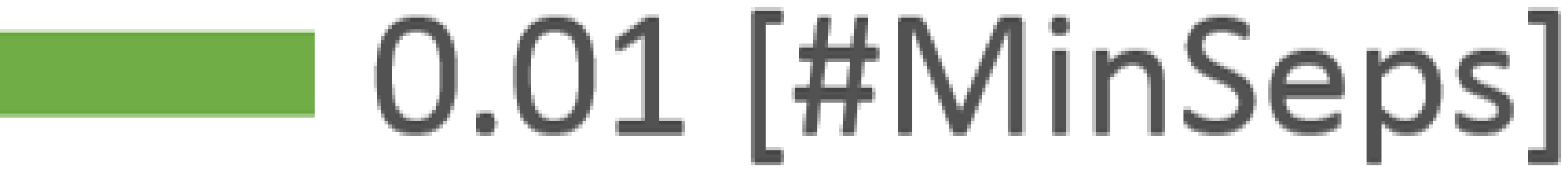}&\includegraphics[width=0.12\textwidth]{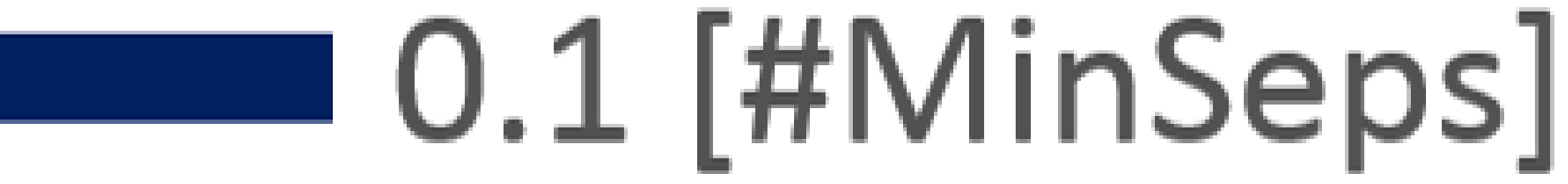}\hfill\includegraphics[width=0.1\textwidth]{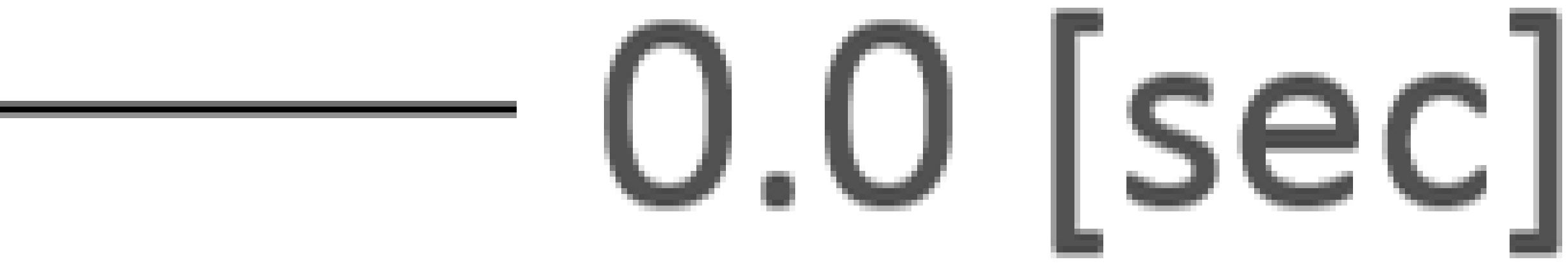}& \includegraphics[width=0.1\textwidth]{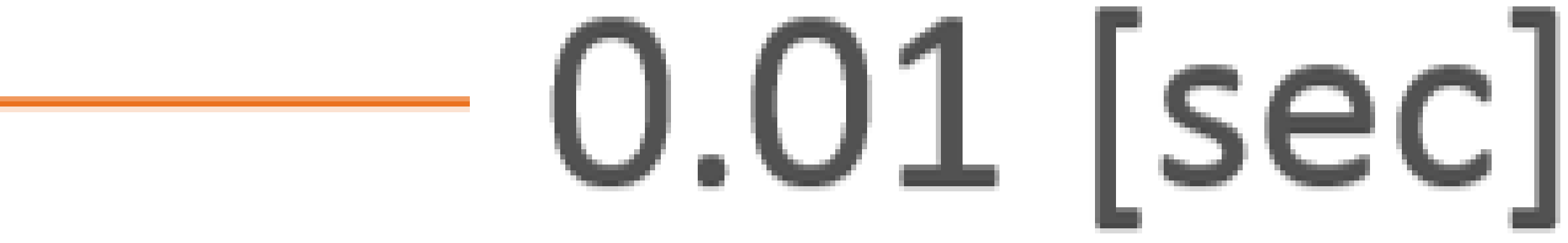}\hfill\includegraphics[width=0.1\textwidth]{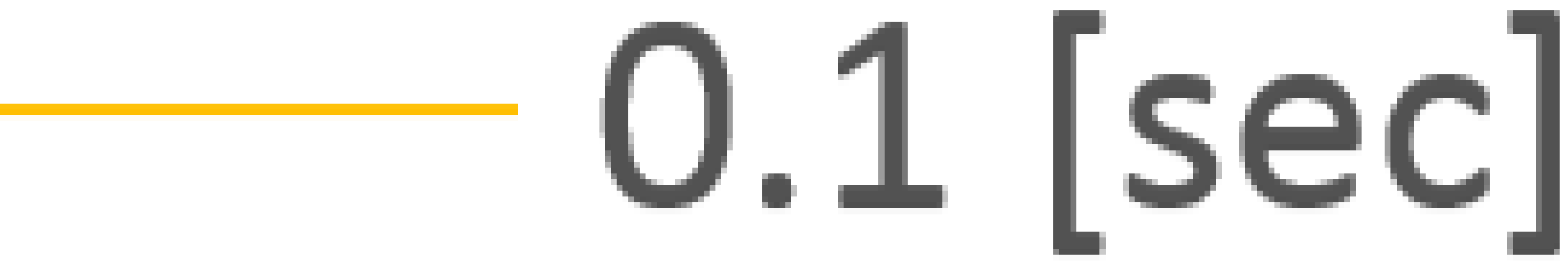}
	\end{tabular}
	\caption{Row scalability experiments, for $\varepsilon \in \set{0., 0.01, 0.1}$ (Sec~\ref{subsubsec:row}).}
	\label{tab:rowScalabilityExperiments}	
\end{figure*}

\begin{figure*}[ht]
	\begin{tabular}{ccc}
		\subfigure[$\textsf{Entity Source}$]{\includegraphics[width=0.33\textwidth]{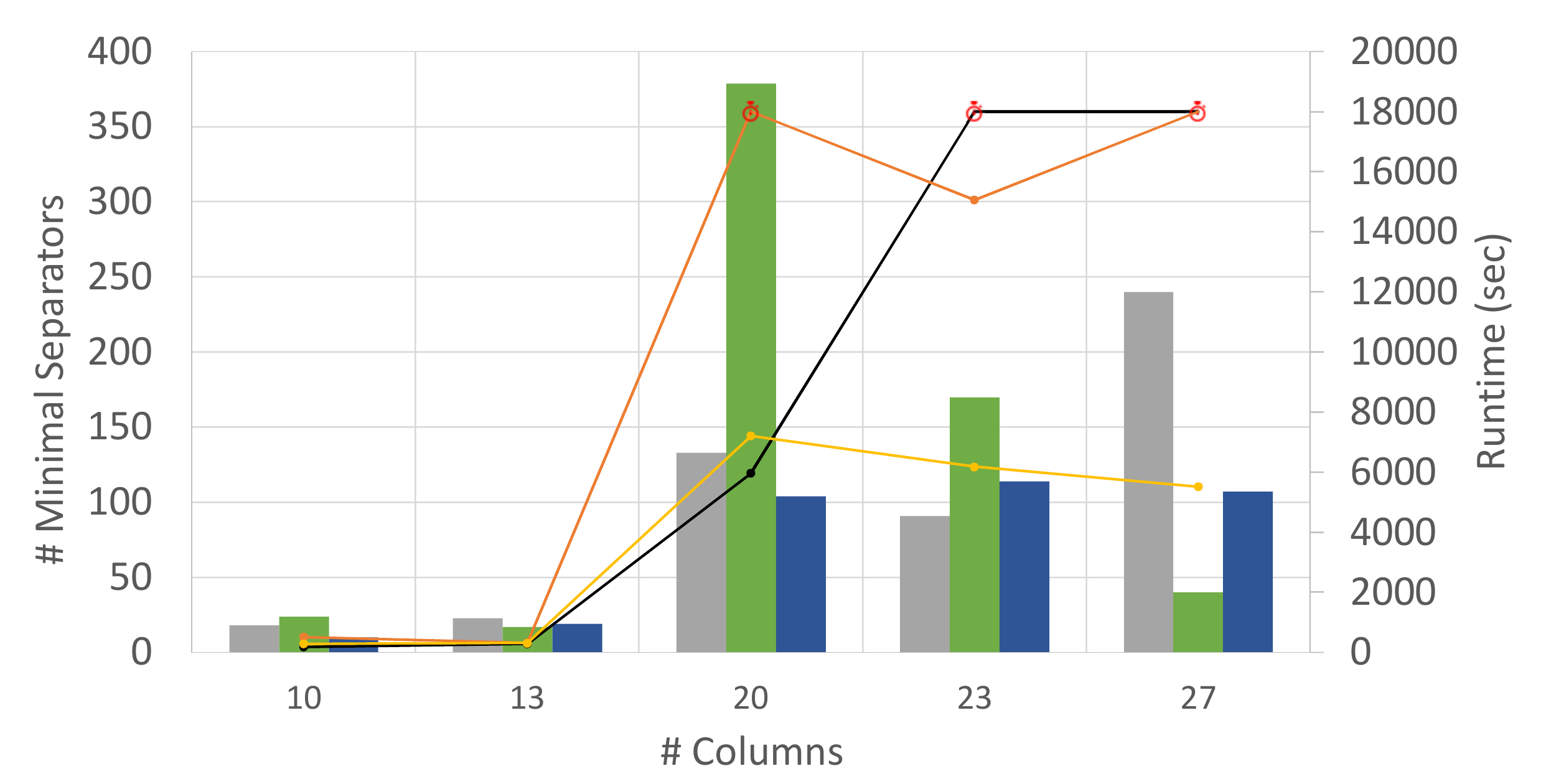}\label{fig:EntitySrc}}& 
		\subfigure[$\textsf{Voter State}$]{\includegraphics[width=0.33\textwidth]{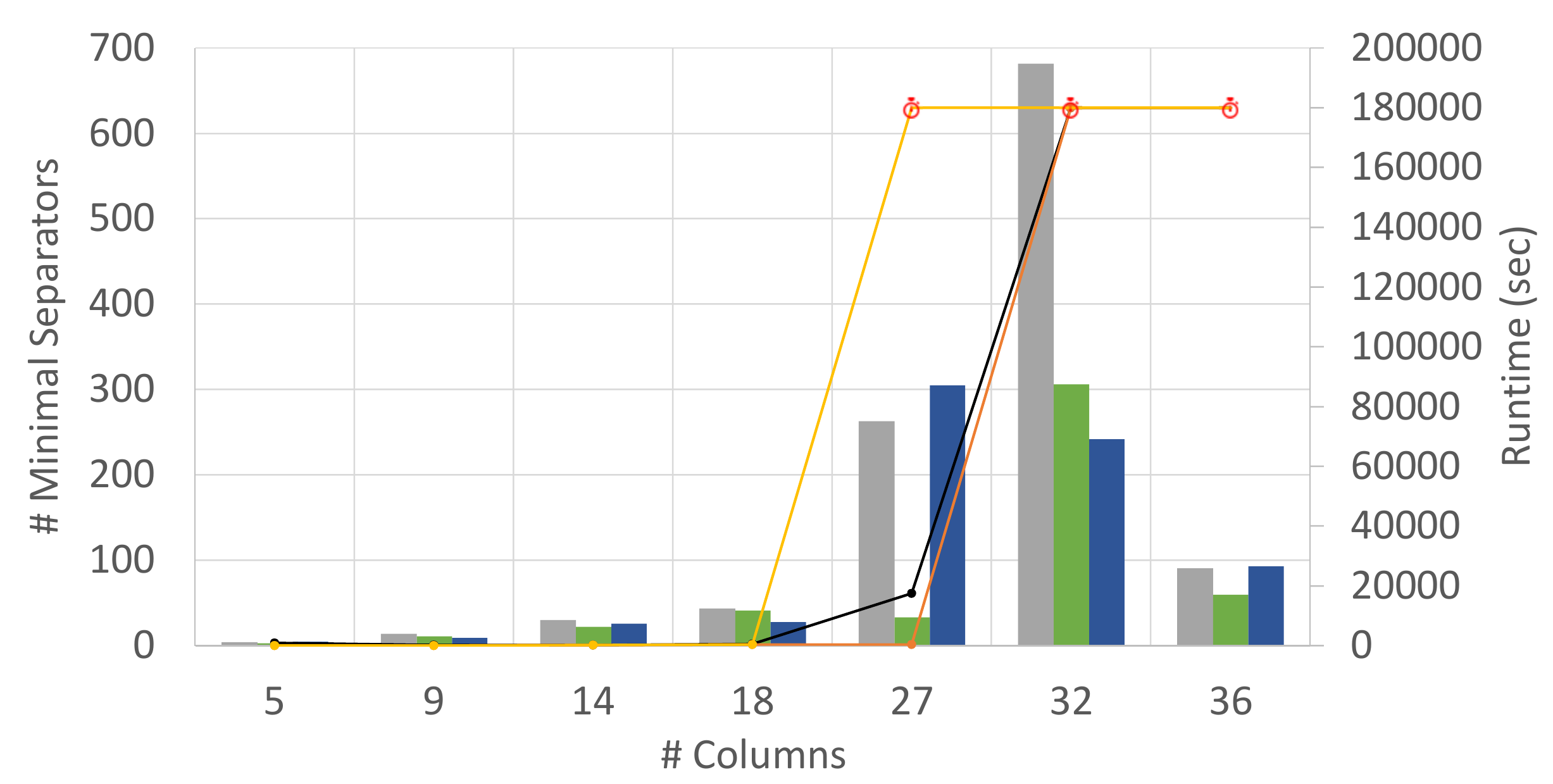}\label{fig:ncVoter}}&
		\subfigure[$\textsf{Census}$]{\includegraphics[width=0.33\textwidth]{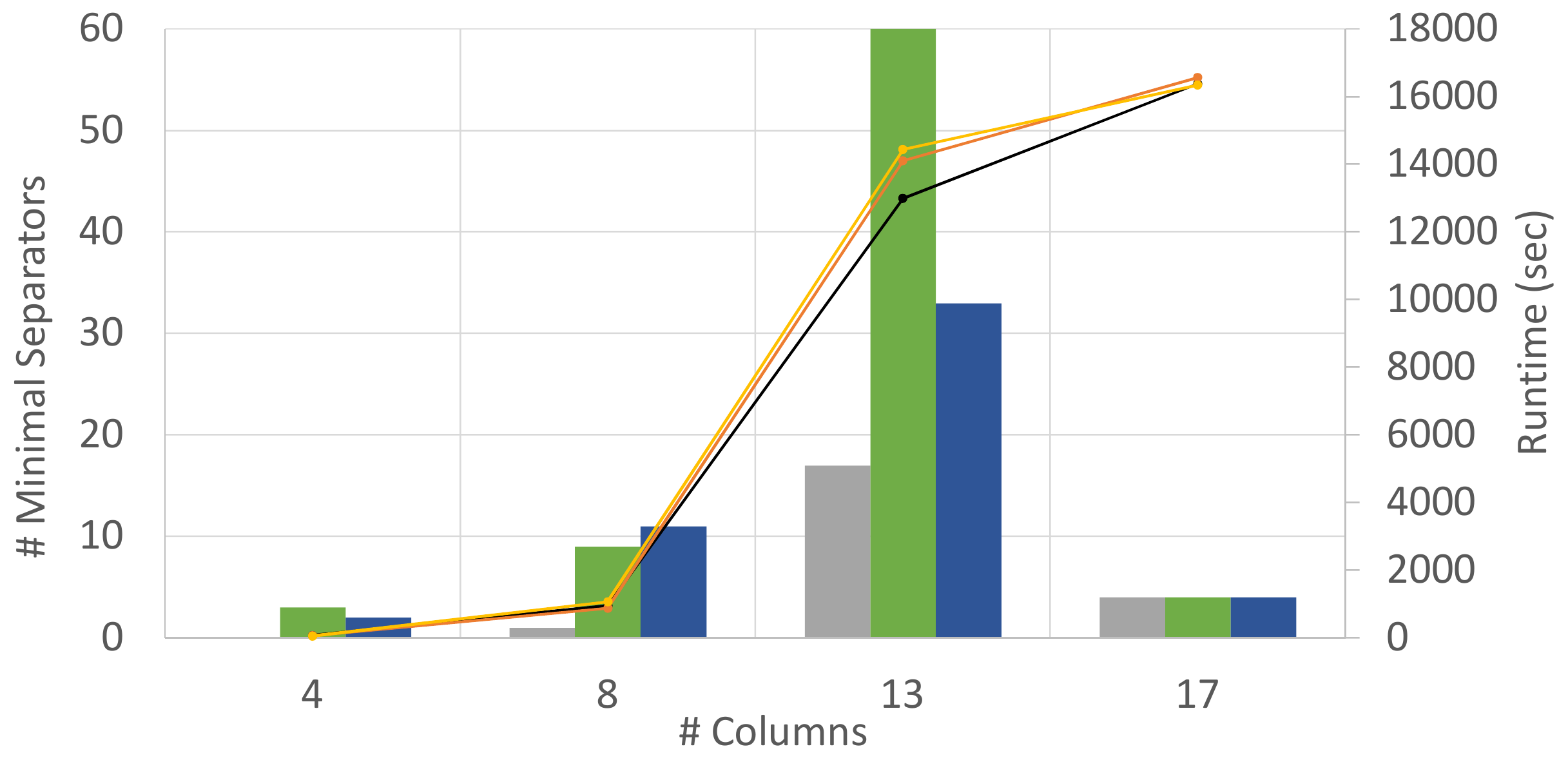}\label{fig:census}}\\	
		 \includegraphics[width=0.12\textwidth]{MinSep00.pdf}\hfill\includegraphics[width=0.12\textwidth]{MinSep001.pdf}&\includegraphics[width=0.12\textwidth]{MinSep01.pdf}\hfill\includegraphics[width=0.1\textwidth]{ScaleTime00.pdf}& \includegraphics[width=0.1\textwidth]{ScaleTime001.pdf}\hfill\includegraphics[width=0.1\textwidth]{ScaleTime01.pdf}
		 \eat{
		\includegraphics[width=0.12\textwidth]{pics/MinSep00.pdf}&\includegraphics[width=0.12\textwidth]{pics/MinSep001.pdf}\hfill\includegraphics[width=0.1\textwidth]{pics/ScaleTime00.pdf}& \includegraphics[width=0.1\textwidth]{pics/ScaleTime001.pdf}
	}
		%&\includegraphics[width=0.32\textwidth]{ColumnScaleLegend}&
	\end{tabular}
	\caption{Column scalability experimentsfor $\varepsilon \in
          \set{0, 0.01,0.1}$ (Sec~\ref{subsubsec:row}). We timed out at five hours (red clock).}
	\label{tab:columnScalabilityExperiments}	
\end{figure*}

\eat{
\begin{figure*}[ht]
	\begin{tabular}{ccccc}
		\subfigure[\textsf{Classification}\label{chart:classificationFullMVDs}]{\includegraphics[width=0.25\textwidth]{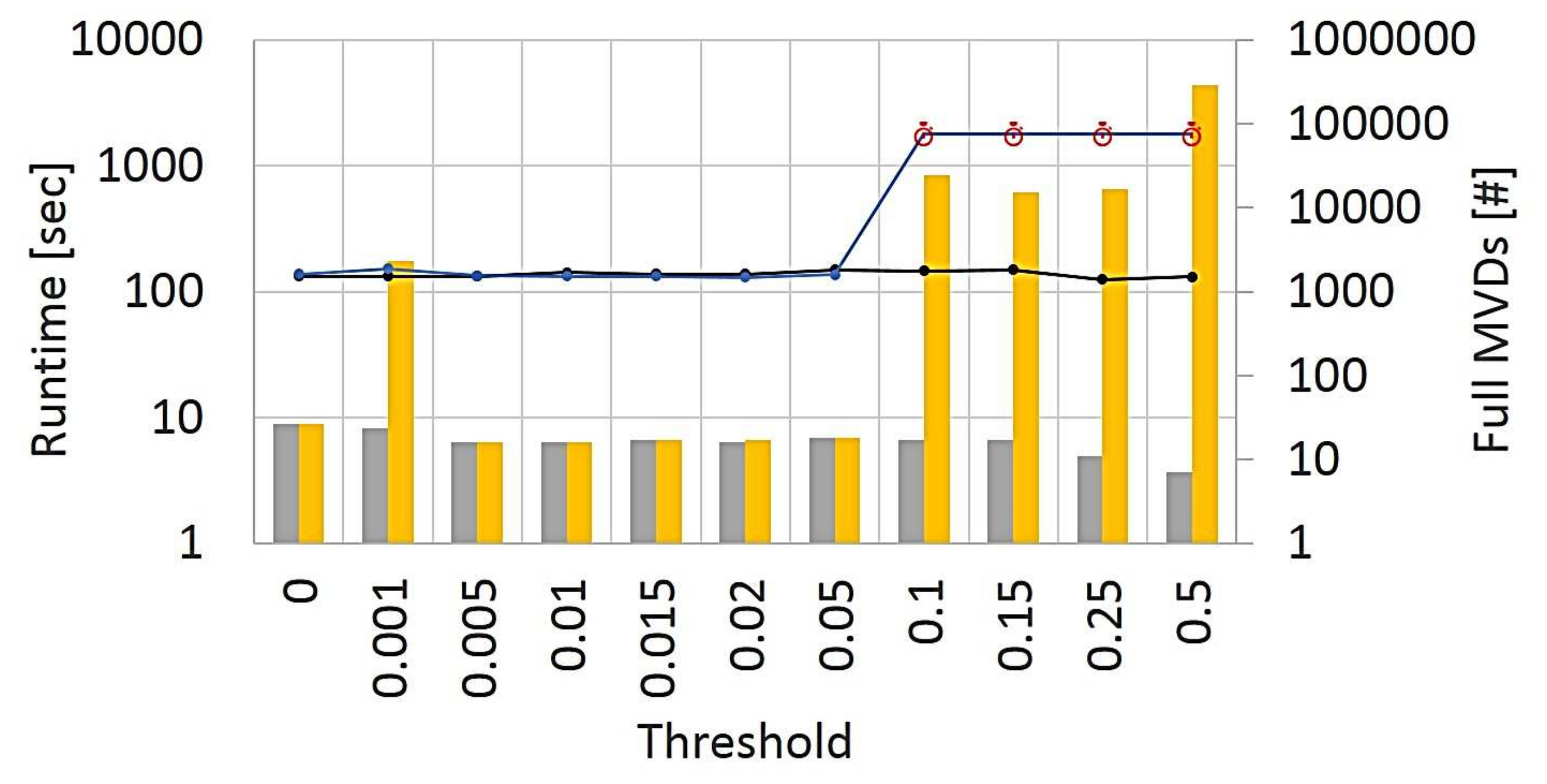}}& 
		\subfigure[\textsf{BreastCancer}]{\includegraphics[width=0.25\textwidth]{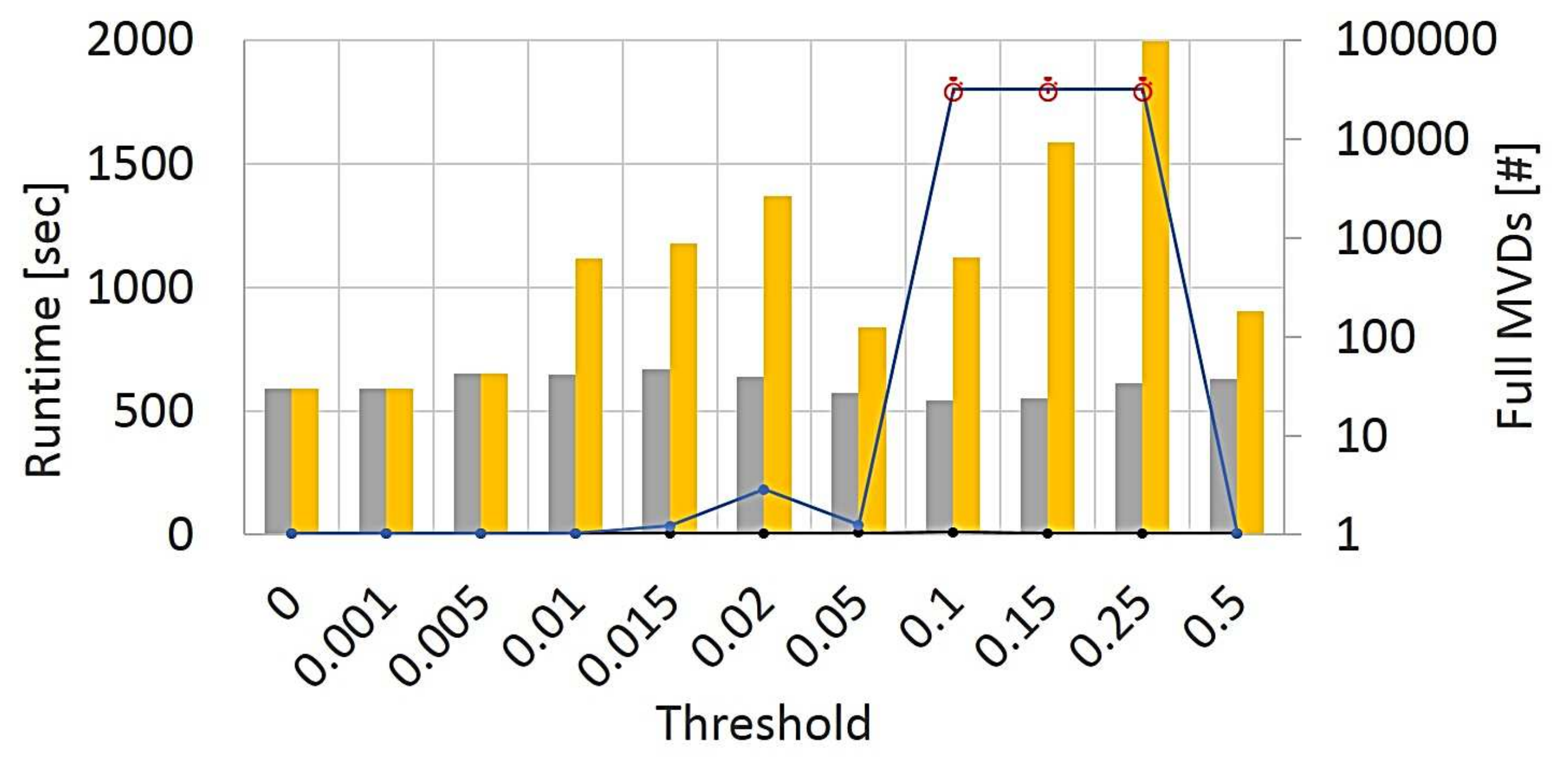}\label{chart:breastCancerFullMVDs}}&
		\subfigure[\textsf{Adult}]{\includegraphics[width=0.25\textwidth]{pics/adultFullMVDs}\label{chart:adultFullMVDs}}&
		\subfigure[\textsf{Bridges}]{\includegraphics[width=0.25\textwidth]{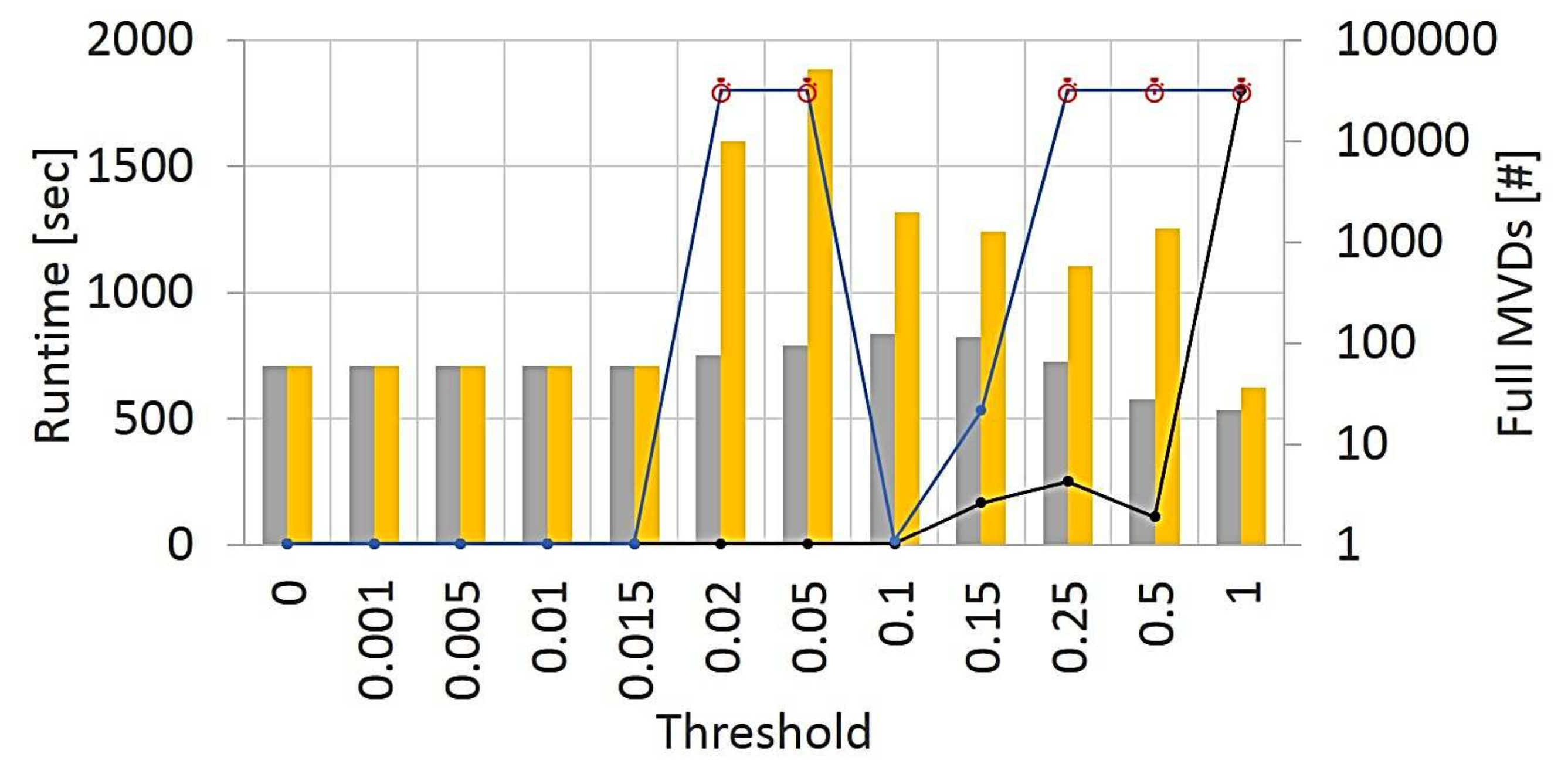}\label{chart:bridgesFullMVDs}}
		\\		
		%\subfloat[]{\includegraphics[width=0.2\textwidth]{pics/horse}}
		\includegraphics[width=0.13\textwidth]{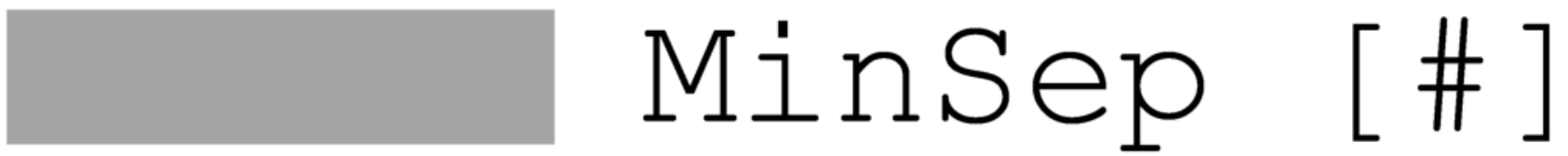}& 
		\includegraphics[width=0.13\textwidth]{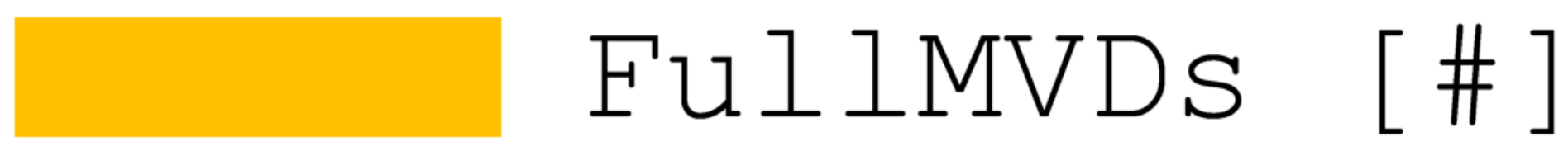} &
		\includegraphics[width=0.13\textwidth]{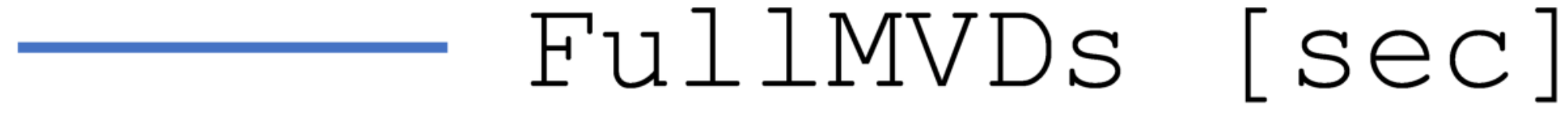}&
		\includegraphics[width=0.13\textwidth]{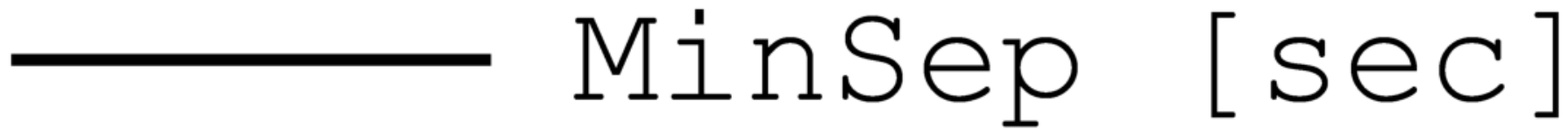}&
	\end{tabular}
	\caption{Full MVDs Experiments. Red stopwatch indicates that the algorithm stopped after 30 minutes.}
	\label{tab:FullMVDsExperiments}	
\end{figure*}
}

\begin{figure*}[ht]
	\begin{tabular}{cccc}
		\subfigure[\textsf{IMAGE}]{\includegraphics[width=0.25\textwidth]{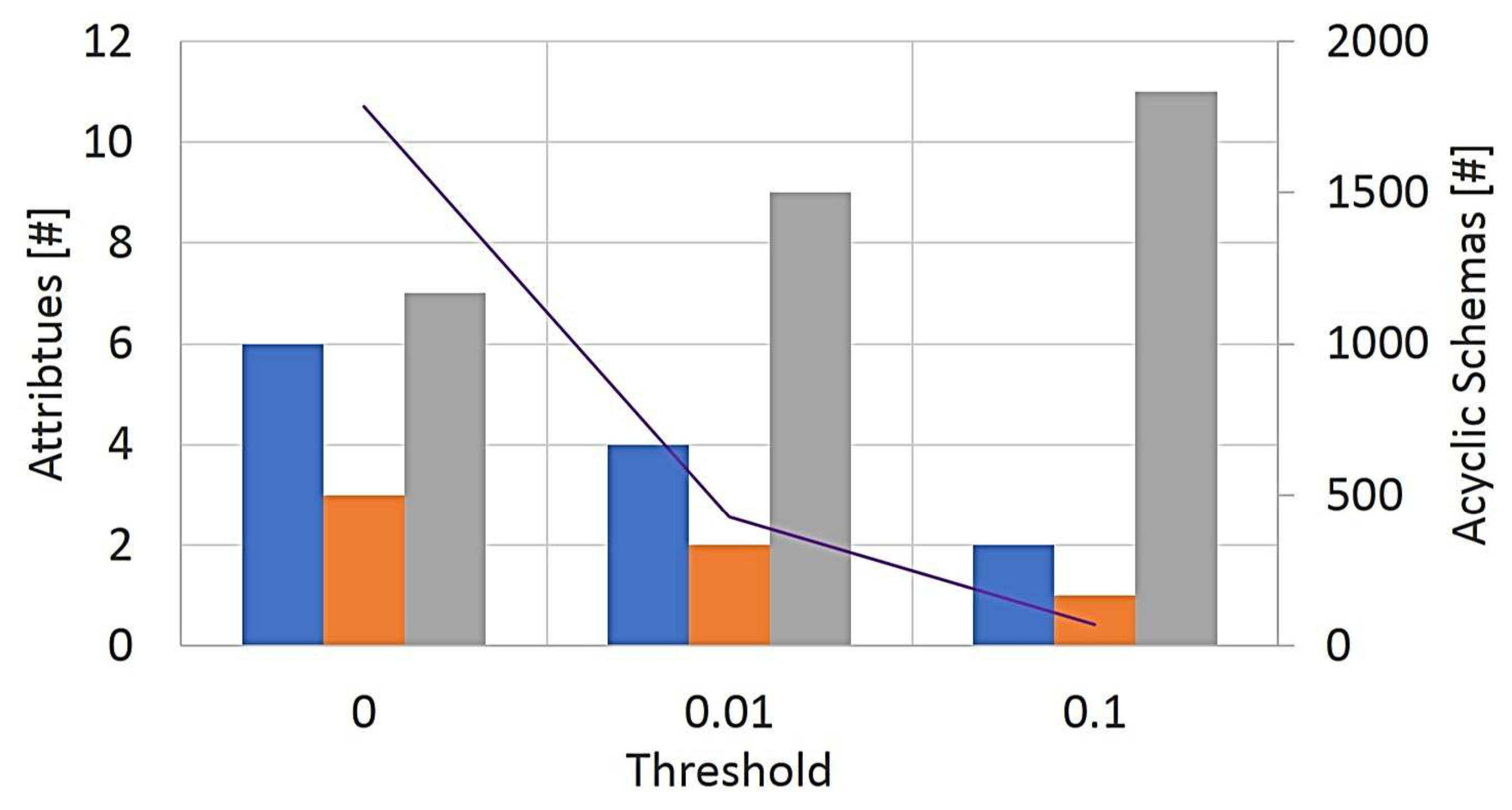}\label{chart:image}}& 
		\subfigure[\textsf{Abalone}]{\includegraphics[width=0.25\textwidth]{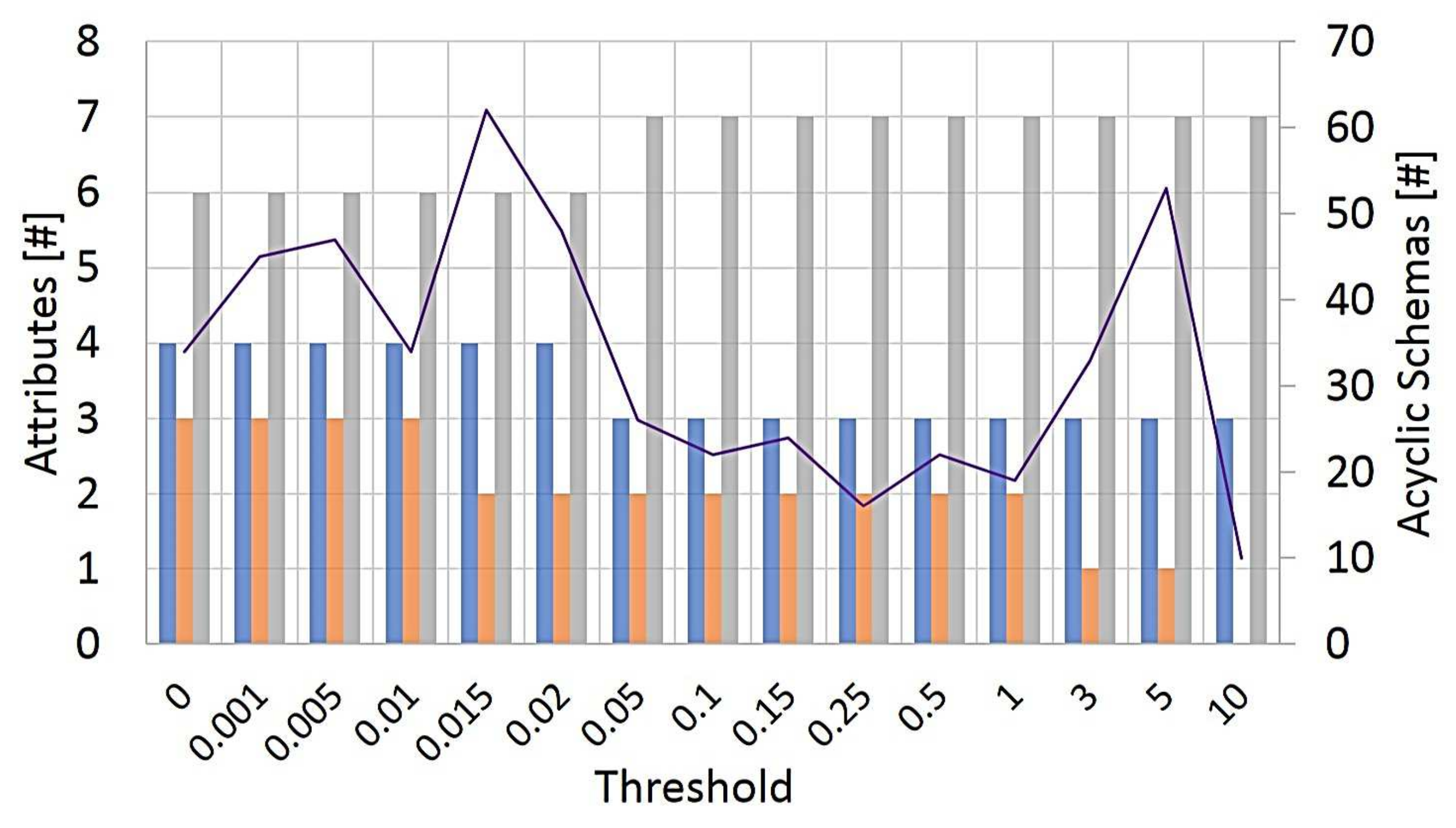}\label{chart:abalone}}& 
		\subfigure[\textsf{Adult}]{\includegraphics[width=0.25\textwidth]{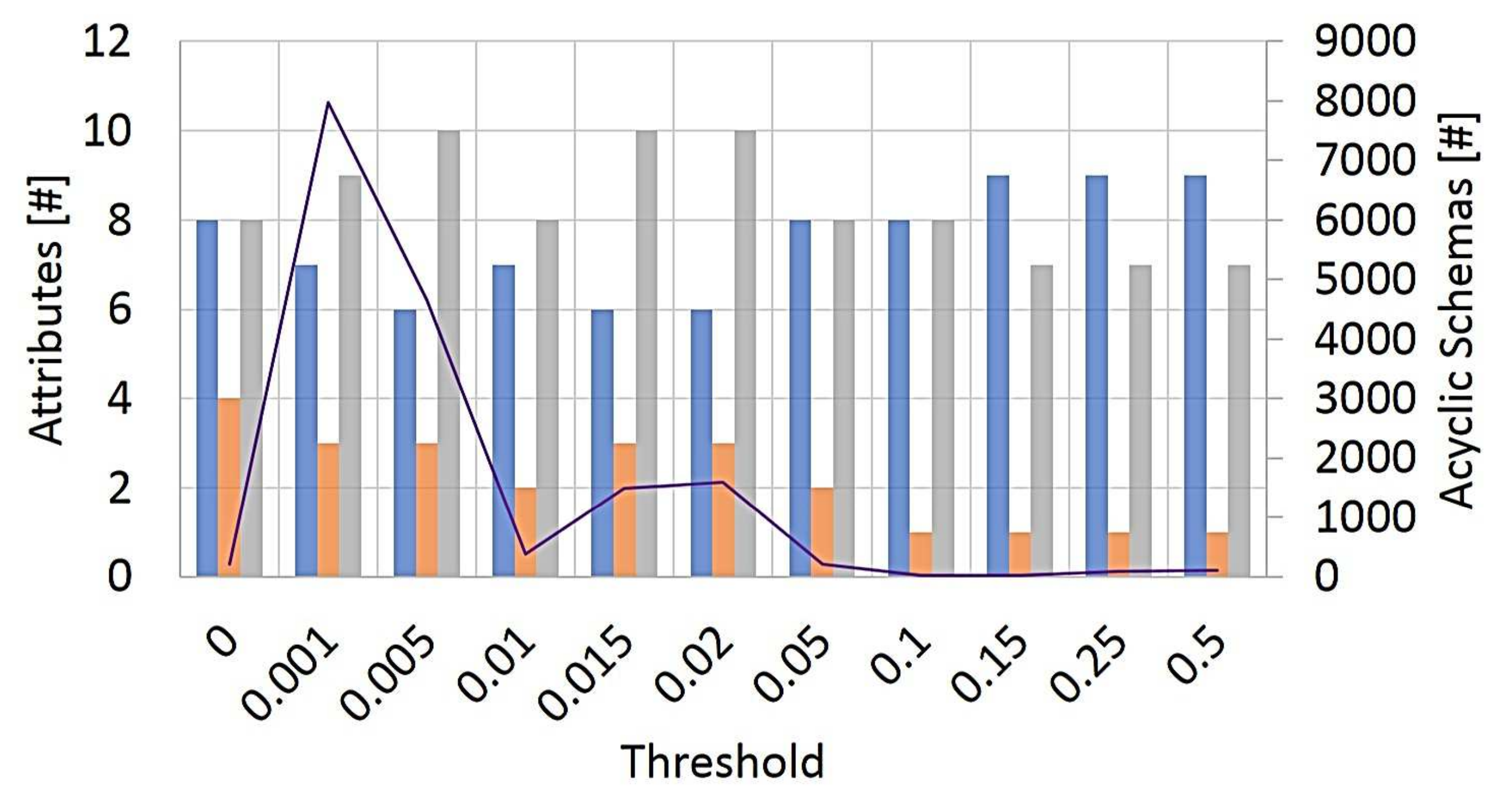}\label{chart:adult}}&
		\subfigure[\textsf{BreastCancer}]{\includegraphics[width=0.25\textwidth]{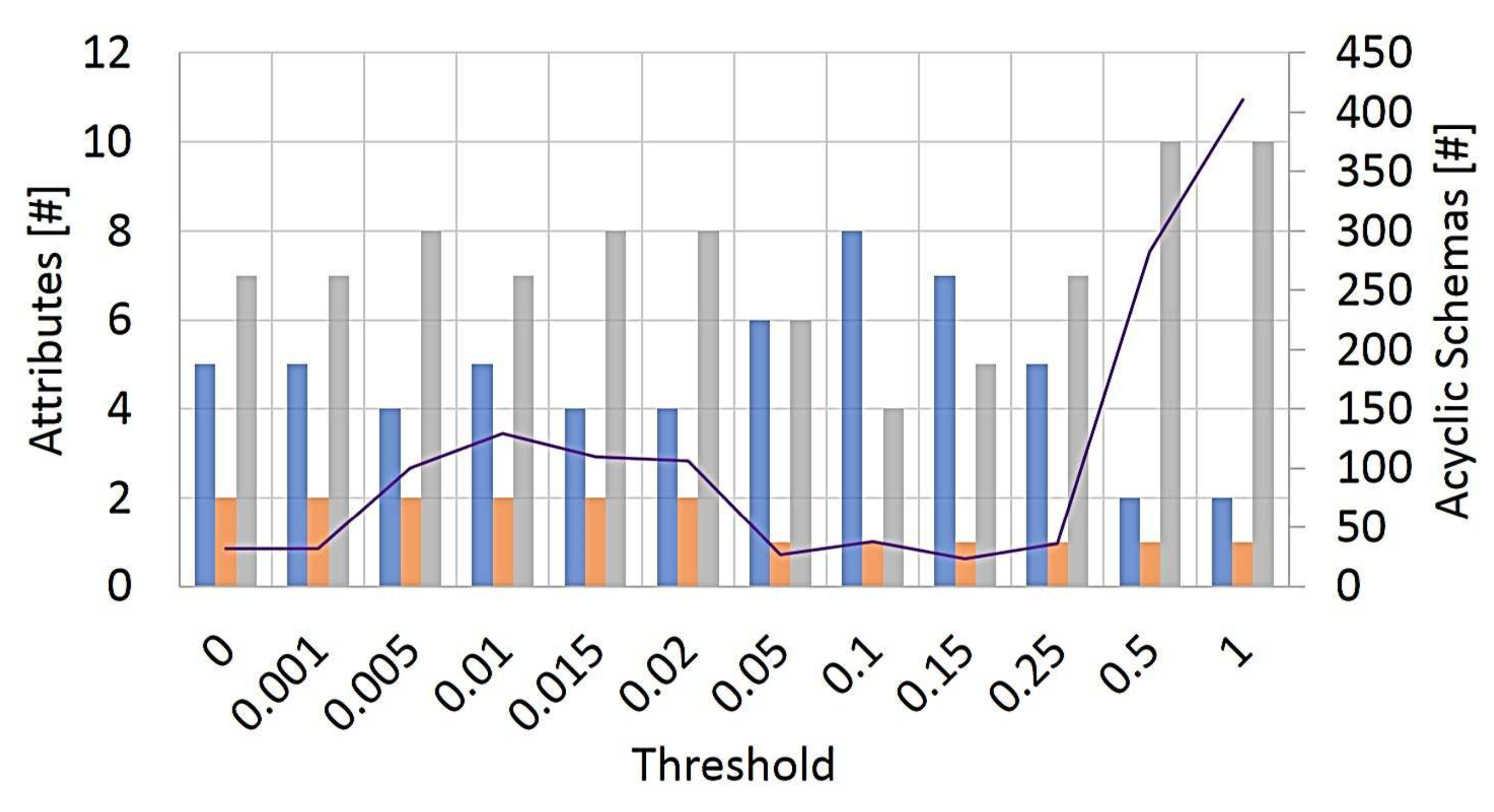}\label{chart:breast}}
		\\
		\subfigure[\textsf{Bridges}]{\includegraphics[width=0.25\textwidth]{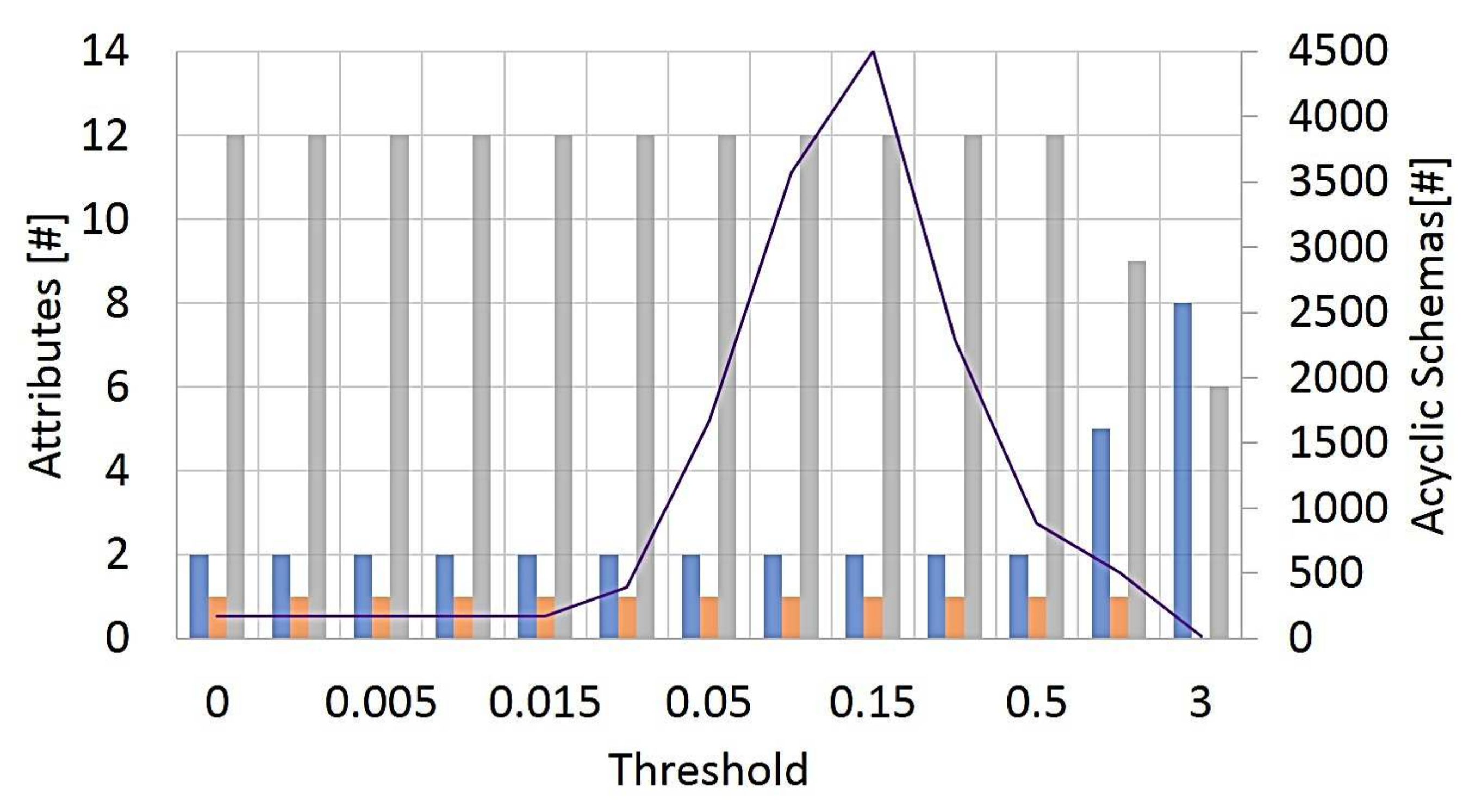}\label{chart:bridge}}&
		\subfigure[\textsf{Echocardiogram}]{\includegraphics[width=0.25\textwidth]{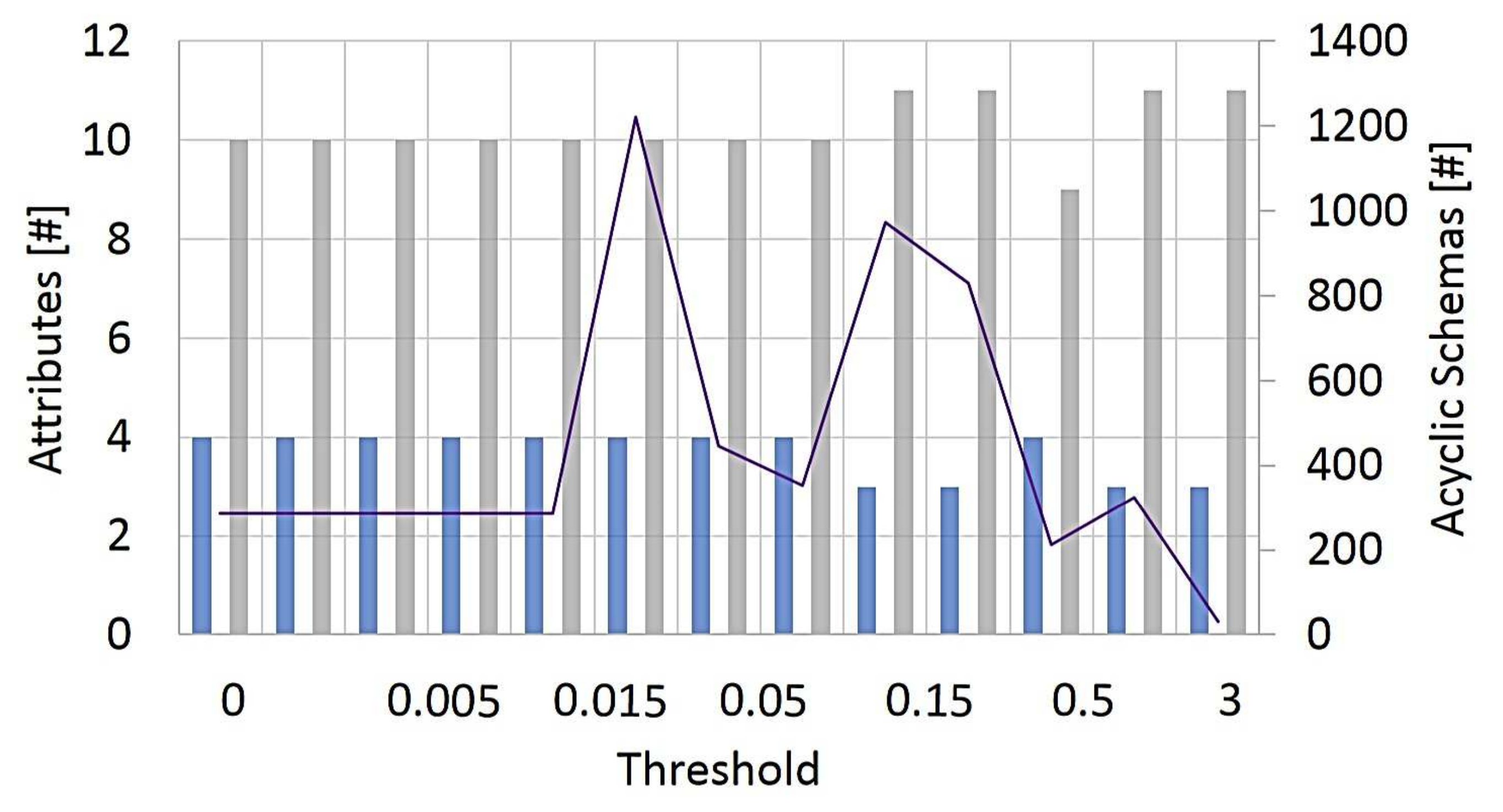}\label{chart:echo}}	&
		\subfigure[\textsf{FD\_Reduced\_15}]{\includegraphics[width=0.25\textwidth]{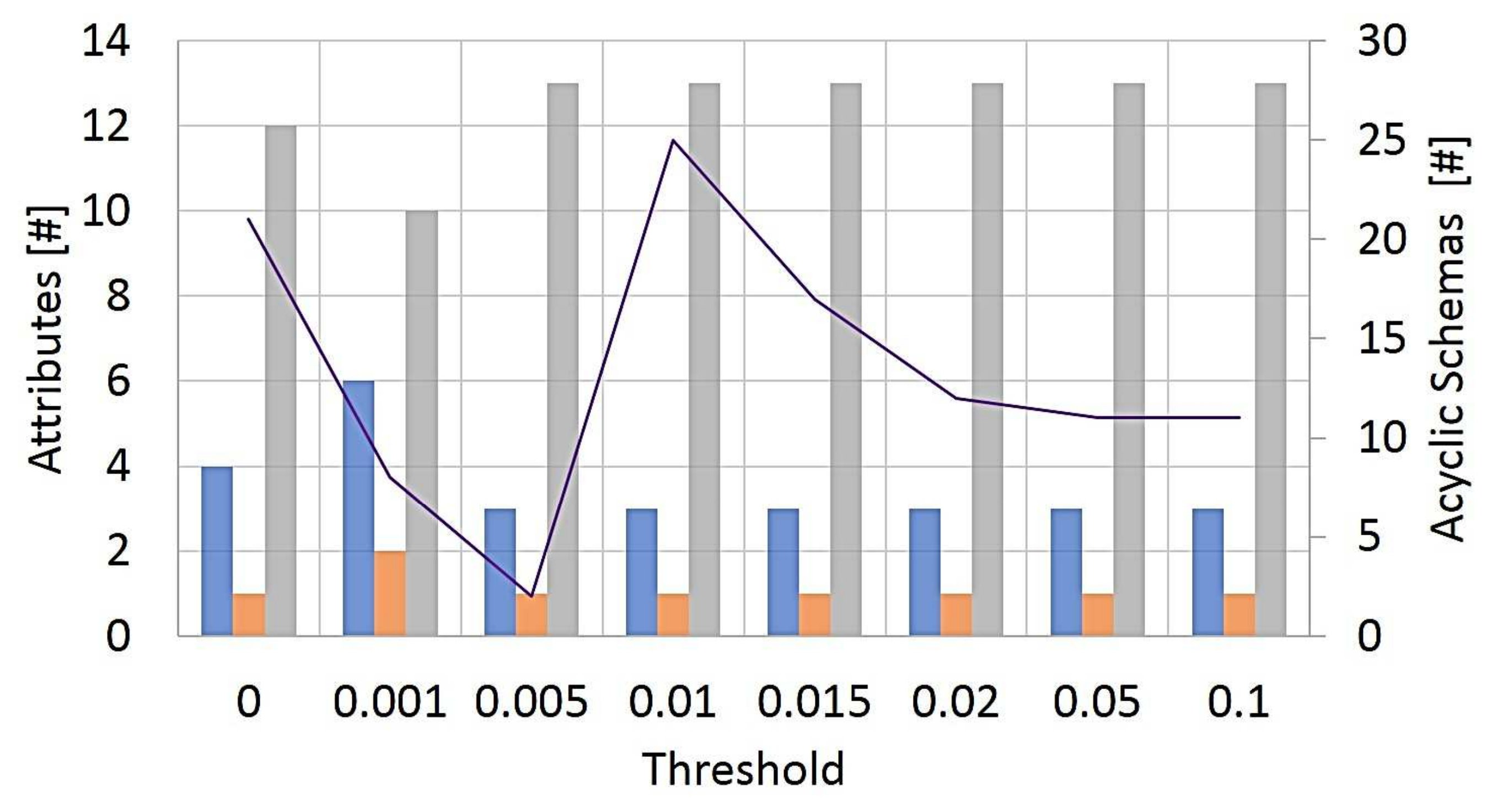}\label{chart:reduced15}}	&
		\subfigure[\textsf{Hepatitis}]{\includegraphics[width=0.25\textwidth]{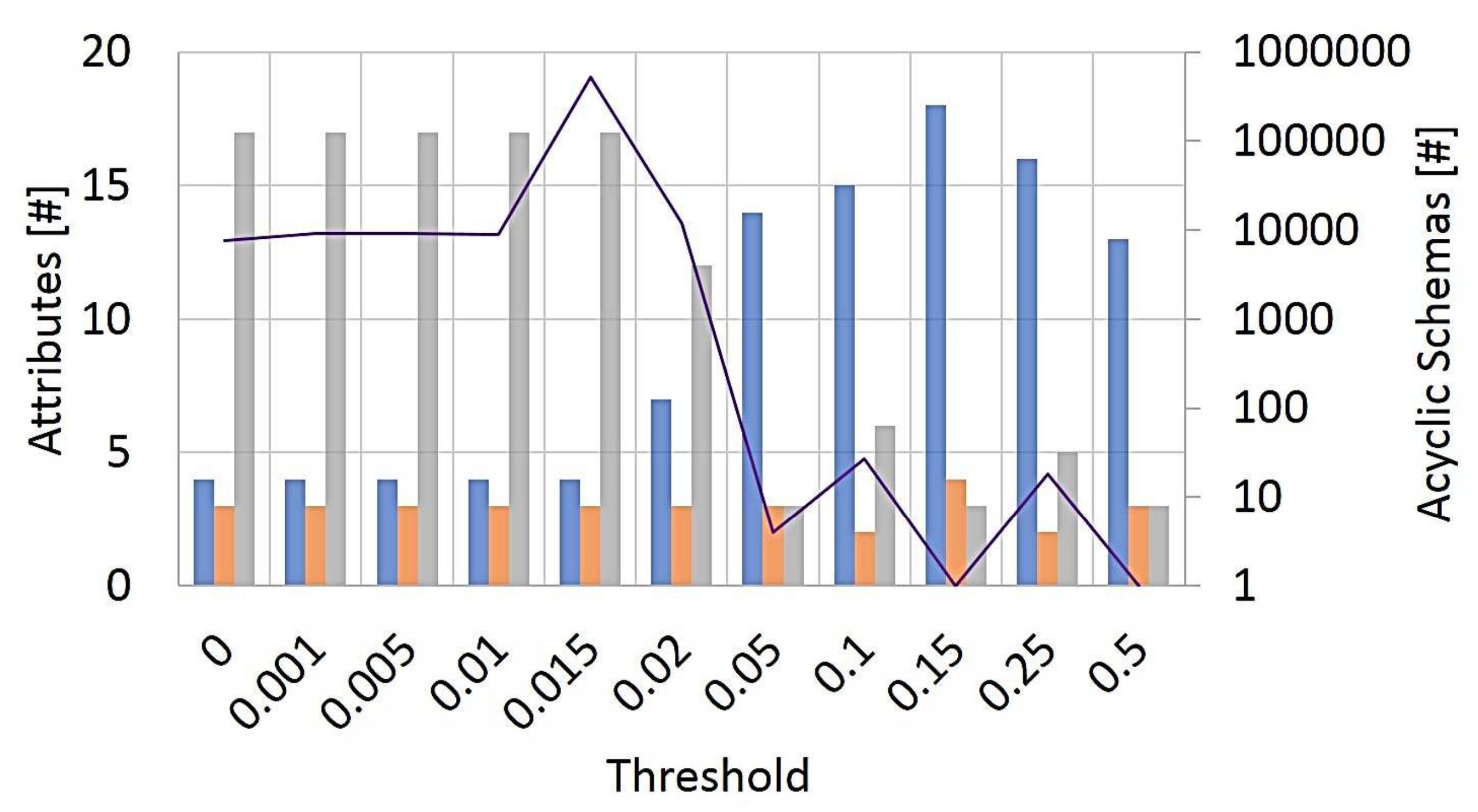}\label{chart:hepatits}}	\\
		%\subfloat[]{\includegraphics[width=0.2\textwidth]{pics/horse}}
		\includegraphics[width=0.2\textwidth]{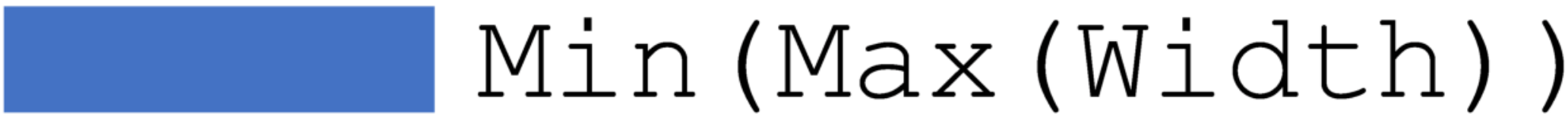}& 
		\includegraphics[width=0.2\textwidth]{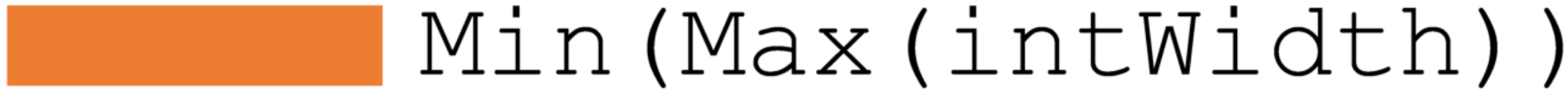}& 
		\includegraphics[width=0.2\textwidth]{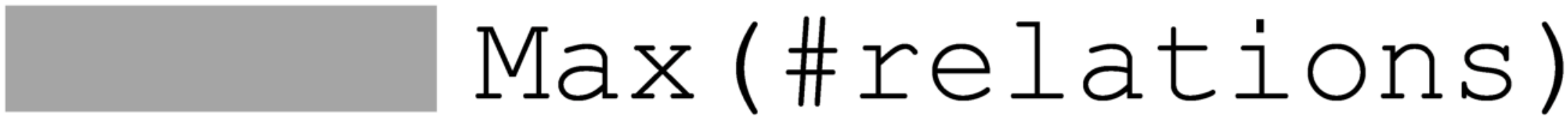}&
		\includegraphics[width=0.2\textwidth]{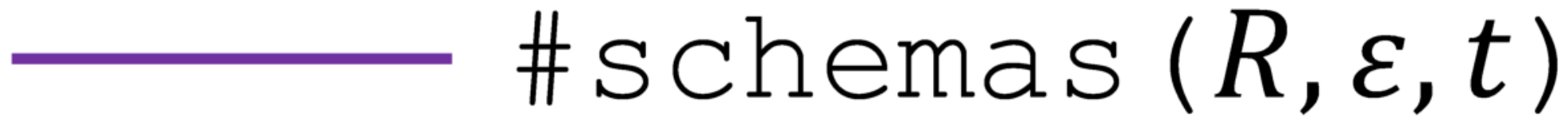}
	\end{tabular}
	\caption{Quality of approximate schemas (Sec.~\ref{sec:enumerationExperiments})}
	\label{tab:EnmerationExperiments}	
\end{figure*}

\subsection{Scalability}\label{sec:scalability}
Next, we evaluated the scalability of \system.  We
    started by computing all exact MVDs ($\varepsilon=0$) on all 20
    datasets and report the runtimes in
    Table~\ref{tab:evaluationDatasets}.  On five of the datasets, our
    system timed out after 5h: for \textsf{Atom Sites},
    \textsf{REFLNS}, and \textsf{Voter State}, it did report a large
    number of full MVDs, while for \textsf{DITAG Feature} and
    \textsf{Census} it did not find any within this limit, but it
    terminated on subsets, as we report below.

    The discovery of acyclic schemes has three parts:
    computing all minimal separators
    (Sec.~\ref{sec:minimal:separators}), discovering all full MVDs
    (Sec.~\ref{sec:getFullMVDs}), and enumerating the acyclic schemes
    (Sec.~\ref{sec:EnumerateAcyclicSchemas}).  We found that the first
    step by far dominates the total runtime, and we report it here; we
    report the other two runtimes in the technical report.  We report
    here the time to compute all minimal separators as a function of
    \#rows, and of \#columns.

\subsubsection{Row Scalability}
\label{subsubsec:row}
We evaluated the algorithm over three large datsets: \textsf{Image},
\textsf{foursquare}, and \textsf{Ditag Feature}. We included all
columns, and a subset of $10 \%$ to $100 \%$ of the tuples. The
results are in Figure~\ref{tab:rowScalabilityExperiments}. In general,
we found that the runtime increases mostly linearly with the size of
the data \eat{.  The runtime increases} even when the number of minimal
separators is mostly constant, e.g.  for \textsf{Image} and
\textsf{Ditag Feature}.

\eat{We note
that for all thresholds we see that the runtime scales linearly with
the dataset size. For both \textsf{Image} and \textsf{Ditag Feature}
there are many more minimal separators generated with a threshold of
$0.0$, but this seems to have little affect on the runtime.}

\subsubsection{Column Scalability}

Next, we varied the number of columns.  Here we kept all rows of the
datasets, and included between $10 \%$ to $100 \%$ of the columns. The
results are presented in
Figure~\ref{tab:columnScalabilityExperiments}.  
We let the algorithm run for 5 hours and measured the resulting number of minimal separators. For example, in the \textsf{Voter State} 
dataset with 32 columns \system~discovered 682, 306 and 242 minimal separators for thresholds 0,0.01, and 0.1 respectively, within the 5h time limit.
\eat{
We note that even in situations where the algorithm 
In all datasets, even though the algorithm timed out after 5h, 
The algorithm timed out after 5h on all 36 columns of \textsf{Voter
  State}, but reported 262 minimal separators.}We found that the
runtime is affected both by the number of attributes, and, quite
significantly, by the number of minimal separators.
This is explained by considering Corollary~\ref{corr:Delay} that analyzes the \e{delay} between the output of minimal separators. First, we note that the delay depends exponentially on the number of attributes (via $\algname{getFullMVDs}$, see Sec.~\ref{subsubsec:an:optimization}) which explains why the delay significantly increases with the number of attributes, leading to an overall reduction in the number of minimal separators returned. Second, the delay also depends on  the number of minimal separators generated up to that point, which explains the high runtime in cases where the data contains a large number of minimal separators.

\vspace{-0.2cm}
\subsection{Quality}\label{sec:enumerationExperiments}
We conducted an empirical evaluation of the
  quality of the schemes generated by \system, and report the results
  in Figure~\ref{tab:EnmerationExperiments}.  
  Per threshold, we ran the enumeration algorithm for half an hour and measured the number of schemes generated (i.e., $\#\texttt{schemes}$), and the following quality measures, for which we report on their aggregate values.

\begin{enumerate}
\item The number of relations in any scheme $\schema$ generated,\eat{relations in $\schema$,} denoted
  $\numRelations(\schema)$.
\item The \e{width} attained by any generated scheme, where width refers to 
the largest number of attributes in any relation of  $\schema$. Formally\footnote{$\width(\schema)$ is precisely the
    treewidth plus one.},
  $\width(\schema) \eqdef \max_{i\in [1,m]}|\Omega_i|$.
\item The \e{intersection width}  attained by any scheme generated, where intersection width refers to the the largest size of any separator of $\schema$.  Formally,
  $\intWidth(\schema)\eqdef\max_{i,j\in
    [1,m]}|\Omega_i{\cap}\Omega_j|$.
\end{enumerate}
\vspace{-0.01cm}
In Figure~\ref{tab:EnmerationExperiments} we increased the threshold
$\varepsilon$, and report for each threshold the maximum
$\numRelations(\schema)$, and the minimum
$\width(\schema), \intWidth(\schema)$ for all schemas at that
threshold.  In general, we observed that, as we increase the
threshold, the system can find more interesting schemes.  For example,
for \textsf{Image} and \textsf{Abalone}, $\width$ (blue bar)
decreases, which means that the number of attributes in the widest
relation decreases. For \textsf{Adult} and \textsf{BreastCancer} the
number of relations ($\numRelations$ -- gray bar) increases, another
indicator of the quality of the schema. 

\section{Conclusions}

\label{sec:conclusions}

We present \system, the first system for the discovery of
approximate acyclic schemes and approximate MVDs from data.  To define
``approximate'', we used concepts from information theory, where each
MVD or acyclic schema is defined by an expression over entropic terms;
when the expression is 0, then the MVD or acyclic schema holds
exactly.  We then presented the two main algorithms in \system, mining
all full $\varepsilon$-MVDs with minimal separators, and discovering
acyclic schemes from a set of $\varepsilon$-MVDs.  Both algorithms
improve over prior work in the literature.  We conducted an
experimental evaluation of \system\ on over 20 real-world data sets.

Our approach of using information theory to define approximate data
dependencies differs from the previous definitions that rely mostly on
counting the number of offending tuples.  On one hand, our definitions
provide us with more powerful mathematical tools, on the other hand
the connection to the actual data quality is less intuitive.  We leave
it up to future work to explore the connection between information
theory and data quality.

Depending on the dataset, \system~ generates hundreds and even thousands of acyclic $\varepsilon$-schemas in as little as 30 minutes. As part of future work we intend to investigate acyclic schema generation in \e{ranked order}. The categories to rank on may be the extent of decomposition (e.g., $\width$ of the schema), or other measures indicative of how well the schema meets the requirements of the application.
\eat{
\batya{This wealth of schemes can be used to develop methods that will learn to identify the best acyclic schemes for a particular workload of queries. 
	This is especially relevant for query processing algorithms that exhibit a high variance in runtime when using different relational schemes. For example, in the
	work of Kalinsky et al.~\cite{DBLP:conf/edbt/KalinskyEK17} on database join optimization,
	the authors present real-life scenarios
	where different schemes (of the same \e{width})
	feature orders-of-magnitude difference in performance.
	Using the acyclic schemes produced by our system, we can build a training set in which every acyclic scheme is associated with a score representing the performance of the query workload over the scheme. This training set will allow us to develop learning algorithms that identify the optimal acyclic scheme for a given query workload.}
}
\clearpage
\newpage
\bibliographystyle{plain}
\bibliography{main}

\begin{thebibliography}{10}

\bibitem{doi:10.2200/S00878ED1V01Y201810DTM052}
Ziawasch Abedjan, Lukasz Golab, Felix Naumann, and Thorsten Papenbrock.
\newblock Data profiling.
\newblock {\em Synthesis Lectures on Data Management}, 10(4):1--154, 2018.

\bibitem{Beeri:1980:MPF:320613.320614}
Catriel Beeri.
\newblock On the menbership problem for functional and multivalued dependencies
  in relational databases.
\newblock {\em ACM Trans. Database Syst.}, 5(3):241--259, September 1980.

\bibitem{Beeri:1979:CPR:320064.320066}
Catriel Beeri and Philip~A. Bernstein.
\newblock Computational problems related to the design of normal form
  relational schemas.
\newblock {\em ACM Trans. Database Syst.}, 4(1):30--59, March 1979.

\bibitem{DBLP:conf/sigmod/BeeriFH77}
Catriel Beeri, Ronald Fagin, and John~H. Howard.
\newblock A complete axiomatization for functional and multivalued dependencies
  in database relations.
\newblock In {\em Proceedings of the 1977 {ACM} {SIGMOD} International
  Conference on Management of Data, Toronto, Canada, August 3-5, 1977.}, pages
  47--61, 1977.

\bibitem{DBLP:conf/stoc/BeeriFMMUY81}
Catriel Beeri, Ronald Fagin, David Maier, Alberto~O. Mendelzon, Jeffrey~D.
  Ullman, and Mihalis Yannakakis.
\newblock Properties of acyclic database schemes.
\newblock In {\em Proceedings of the 13th Annual {ACM} Symposium on Theory of
  Computing, May 11-13, 1981, Milwaukee, Wisconsin, {USA}}, pages 355--362,
  1981.

\bibitem{Beeri:1983:DAD:2402.322389}
Catriel Beeri, Ronald Fagin, David Maier, and Mihalis Yannakakis.
\newblock On the desirability of acyclic database schemes.
\newblock {\em J. ACM}, 30(3):479--513, July 1983.

\bibitem{DBLP:journals/tods/Bernstein76}
Philip~A. Bernstein.
\newblock Synthesizing third normal form relations from functional
  dependencies.
\newblock {\em {ACM} Trans. Database Syst.}, 1(4):277--298, 1976.

\bibitem{DBLP:conf/cikm/BleifussBFRW0PN16}
Tobias Bleifu{\ss}, Susanne B{\"{u}}low, Johannes Frohnhofen, Julian Risch,
  Georg Wiese, Sebastian Kruse, Thorsten Papenbrock, and Felix Naumann.
\newblock Approximate discovery of functional dependencies for large datasets.
\newblock In {\em Proceedings of the 25th {ACM} International Conference on
  Information and Knowledge Management, {CIKM} 2016, Indianapolis, IN, USA,
  October 24-28, 2016}, pages 1803--1812, 2016.

\bibitem{DBLP:conf/pods/CarmeliKK17}
Nofar Carmeli, Batya Kenig, and Benny Kimelfeld.
\newblock Efficiently enumerating minimal triangulations.
\newblock In {\em Proceedings of the 36th {ACM} {SIGMOD-SIGACT-SIGAI} Symposium
  on Principles of Database Systems, {PODS} 2017, Chicago, IL, USA, May 14-19,
  2017}, pages 273--287, 2017.

\bibitem{Codd1971FurtherNO}
E.~F. Codd.
\newblock Further normalization of the data base relational model.
\newblock {\em IBM Research Report, San Jose, California}, RJ909, 1971.

\bibitem{DBLP:journals/jcss/CohenKS08}
Sara Cohen, Benny Kimelfeld, and Yehoshua Sagiv.
\newblock Generating all maximal induced subgraphs for hereditary and
  connected-hereditary graph properties.
\newblock {\em J. Comput. Syst. Sci.}, 74(7):1147--1159, 2008.

\bibitem{draeger2016}
Tim Draeger.
\newblock Multivalued dependency discovery, 2016.
\newblock Master's Thesis, Hasso-Plattner-Institute, Potsdam.

\bibitem{Fagin:1977:MDN:320557.320571}
Ronald Fagin.
\newblock Multivalued dependencies and a new normal form for relational
  databases.
\newblock {\em ACM Trans. Database Syst.}, 2(3):262--278, September 1977.

\bibitem{DBLP:journals/jacm/Fagin82}
Ronald Fagin.
\newblock Horn clauses and database dependencies.
\newblock {\em J. {ACM}}, 29(4):952--985, 1982.

\bibitem{DBLP:journals/tods/FaginMU82}
Ronald Fagin, Alberto~O. Mendelzon, and Jeffrey~D. Ullman.
\newblock A simplified universal relation assumption and its properties.
\newblock {\em {ACM} Trans. Database Syst.}, 7(3):343--360, 1982.

\bibitem{FREDMAN1996618}
Michael~L. Fredman and Leonid Khachiyan.
\newblock On the complexity of dualization of monotone disjunctive normal
  forms.
\newblock {\em Journal of Algorithms}, 21(3):618 -- 628, 1996.

\bibitem{GeigerPearl1993}
Dan Geiger and Judea Pearl.
\newblock Logical and algorithmic properties of conditional independence and
  graphical models.
\newblock {\em The Annals of Statistics}, 21(4):2001--2021, 1993.

\bibitem{DBLP:journals/ipl/GoodmanT84}
Nathan Goodman and Y.~C. Tay.
\newblock A characterization of multivalued dependencies equivalent to a join
  dependency.
\newblock {\em Inf. Process. Lett.}, 18(5):261--266, 1984.

\bibitem{DBLP:journals/tcs/Gucht88}
Dirk~Van Gucht.
\newblock Interaction-free multivalued dependency sets.
\newblock {\em Theor. Comput. Sci.}, 62(1-2):221--233, 1988.

\bibitem{DBLP:journals/tods/GunopulosKMSTS03}
Dimitrios Gunopulos, Roni Khardon, Heikki Mannila, Sanjeev Saluja, Hannu
  Toivonen, and Ram~Sewak Sharm.
\newblock Discovering all most specific sentences.
\newblock {\em {ACM} Trans. Database Syst.}, 28(2):140--174, 2003.

\bibitem{DBLP:journals/cj/HuhtalaKPT99}
Yk{\"{a}} Huhtala, Juha K{\"{a}}rkk{\"{a}}inen, Pasi Porkka, and Hannu
  Toivonen.
\newblock {TANE:} an efficient algorithm for discovering functional and
  approximate dependencies.
\newblock {\em Comput. J.}, 42(2):100--111, 1999.

\bibitem{DBLP:journals/ipl/JohnsonP88}
David~S. Johnson, Christos~H. Papadimitriou, and Mihalis Yannakakis.
\newblock On generating all maximal independent sets.
\newblock {\em Inf. Process. Lett.}, 27(3):119--123, 1988.

\bibitem{KHACHIYAN20062350}
Leonid Khachiyan, Endre Boros, Khaled Elbassioni, and Vladimir Gurvich.
\newblock An efficient implementation of a quasi-polynomial algorithm for
  generating hypergraph transversals and its application in joint generation.
\newblock {\em Discrete Applied Mathematics}, 154(16):2350 -- 2372, 2006.
\newblock Discrete Algorithms and Optimization, in Honor of Professor Toshihide
  Ibaraki at His Retirement from Kyoto University.

\bibitem{DBLP:conf/sigmod/Khamis0NOS18}
Mahmoud~Abo Khamis, Hung~Q. Ngo, XuanLong Nguyen, Dan Olteanu, and Maximilian
  Schleich.
\newblock {AC/DC:} in-database learning thunderstruck.
\newblock In {\em Proceedings of the Second Workshop on Data Management for
  End-To-End Machine Learning, DEEM@SIGMOD 2018, Houston, TX, USA, June 15,
  2018}, pages 8:1--8:10, 2018.

\bibitem{DBLP:conf/pods/KhamisNR16}
Mahmoud~Abo Khamis, Hung~Q. Ngo, and Atri Rudra.
\newblock {FAQ:} questions asked frequently.
\newblock In {\em Proceedings of the 35th {ACM} {SIGMOD-SIGACT-SIGAI} Symposium
  on Principles of Database Systems, {PODS} 2016, San Francisco, CA, USA, June
  26 - July 01, 2016}, pages 13--28, 2016.

\bibitem{DBLP:journals/tcs/KivinenM95}
Jyrki Kivinen and Heikki Mannila.
\newblock Approximate inference of functional dependencies from relations.
\newblock {\em Theor. Comput. Sci.}, 149(1):129--149, 1995.

\bibitem{DBLP:journals/pvldb/0001N18}
Sebastian Kruse and Felix Naumann.
\newblock Efficient discovery of approximate dependencies.
\newblock {\em {PVLDB}}, 11(7):759--772, 2018.

\bibitem{DBLP:journals/tcs/Lakshmanan88}
V.~S. Lakshmanan.
\newblock Split-freedom and mvd-intersection: {A} new characterization of
  multivalued dependencies having conflict-free covers.
\newblock {\em Theor. Comput. Sci.}, 62(1-2):105--122, 1988.

\bibitem{DBLP:journals/tse/Lee87}
Tony~T. Lee.
\newblock An information-theoretic analysis of relational databases - part {I:}
  data dependencies and information metric.
\newblock {\em {IEEE} Trans. Software Eng.}, 13(10):1049--1061, 1987.

\bibitem{DBLP:journals/tse/Lee87a}
Tony~T. Lee.
\newblock An information-theoretic analysis of relational databases - part
  {II:} information structures of database schemas.
\newblock {\em {IEEE} Trans. Software Eng.}, 13(10):1061--1072, 1987.

\bibitem{DBLP:journals/is/LeveneL03}
Mark Levene and George Loizou.
\newblock Why is the snowflake schema a good data warehouse design?
\newblock {\em Inf. Syst.}, 28(3):225--240, 2003.

\bibitem{Lien:1981:HSR:319540.319546}
Y.~Edmund Lien.
\newblock Hierarchical schemata for relational databases.
\newblock {\em ACM Trans. Database Syst.}, 6(1):48--69, March 1981.

\bibitem{DBLP:journals/tkde/LiuLLC12}
Jixue Liu, Jiuyong Li, Chengfei Liu, and Yongfeng Chen.
\newblock Discover dependencies from data - {A} review.
\newblock {\em {IEEE} Trans. Knowl. Data Eng.}, 24(2):251--264, 2012.

\bibitem{Papenbrock:2015:DPM:2824032.2824086}
Thorsten Papenbrock, Tanja Bergmann, Moritz Finke, Jakob Zwiener, and Felix
  Naumann.
\newblock Data profiling with metanome.
\newblock {\em Proc. VLDB Endow.}, 8(12):1860--1863, August 2015.

\bibitem{DBLP:conf/sigmod/PapenbrockN16}
Thorsten Papenbrock and Felix Naumann.
\newblock A hybrid approach to functional dependency discovery.
\newblock In {\em Proceedings of the 2016 International Conference on
  Management of Data, {SIGMOD} Conference 2016, San Francisco, CA, USA, June 26
  - July 01, 2016}, pages 821--833, 2016.

\bibitem{DBLP:journals/corr/abs-1908-07924}
Babak Salimi, Bill Howe, and Dan Suciu.
\newblock Data management for causal algorithmic fairness.
\newblock {\em IEEE Data Engineering Bulletin}, vol. 42, no. 3, 2019.

\bibitem{DBLP:conf/sigmod/SalimiRHS19}
Babak Salimi, Luke Rodriguez, Bill Howe, and Dan Suciu.
\newblock Interventional fairness: Causal database repair for algorithmic
  fairness.
\newblock In {\em Proceedings of the 2019 International Conference on
  Management of Data, {SIGMOD} Conference 2019, Amsterdam, The Netherlands,
  June 30 - July 5, 2019}, pages 793--810, 2019.

\bibitem{DBLP:journals/ida/SavnikF00}
Iztok Savnik and Peter~A. Flach.
\newblock Discovery of multivalued dependencies from relations.
\newblock {\em Intell. Data Anal.}, 4(3-4):195--211, 2000.

\bibitem{DBLP:conf/sigmod/SchleichOC16}
Maximilian Schleich, Dan Olteanu, and Radu Ciucanu.
\newblock Learning linear regression models over factorized joins.
\newblock In {\em Proceedings of the 2016 International Conference on
  Management of Data, {SIGMOD} Conference 2016, San Francisco, CA, USA, June 26
  - July 01, 2016}, pages 3--18, 2016.

\bibitem{DBLP:conf/sigmod/SchleichOK0N19}
Maximilian Schleich, Dan Olteanu, Mahmoud~Abo Khamis, Hung~Q. Ngo, and XuanLong
  Nguyen.
\newblock A layered aggregate engine for analytics workloads.
\newblock In {\em Proceedings of the 2019 International Conference on
  Management of Data, {SIGMOD} Conference 2019, Amsterdam, The Netherlands,
  June 30 - July 5, 2019.}, pages 1642--1659, 2019.

\bibitem{DBLP:conf/dawak/WyssGR01}
Catharine~M. Wyss, Chris Giannella, and Edward~L. Robertson.
\newblock Fastfds: {A} heuristic-driven, depth-first algorithm for mining
  functional dependencies from relation instances - extended abstract.
\newblock In {\em Data Warehousing and Knowledge Discovery, Third International
  Conference, DaWaK 2001, Munich, Germany, September 5-7, 2001, Proceedings},
  pages 101--110, 2001.

\bibitem{Yannakakis:1981:AAD:1286831.1286840}
Mihalis Yannakakis.
\newblock Algorithms for acyclic database schemes.
\newblock In {\em Proceedings of the Seventh International Conference on Very
  Large Data Bases - Volume 7}, VLDB '81, pages 82--94. VLDB Endowment, 1981.

\bibitem{naumann-datasets}
Datasets of the metanome data profiling project.
\newblock
  \url{https://hpi.de/naumann/projects/repeatability/data-profiling/fds.html\#c168191}.

\bibitem{h2-database}
$h2$ main memory database.
\newblock \url{https://www.h2database.com/html/main.html}.

\end{thebibliography}

\clearpage
\newpage
\pagebreak 
\onecolumn
\section{Appendix}
\section{Proofs from Section~\ref{sec:techniques}}
Given two MVDs
$\phi = S \mvd X_1 | \dots |X_m$ and $\psi = S \mvd Y_1|\dots|Y_k$,
define their \e{join} as
$\phi \vee \psi = S \mvd Z_{11} | Z_{12} | \cdots | Z_{mk}$, where
$Z_{ij} = X_i \cap Y_j$.  Clearly, $\phi \vee \psi$ refines both
$\phi$ and $\psi$, i.e.
$\J(\phi \vee \psi) \geq \max(\J(\phi),\J(\psi))$.  We prove a weak
form of converse:
\begin{replemma}{\ref{lem:singleMaxMVDLemma}}
\singleMaxMVDLemma
\end{replemma}

\begin{proof} We prove the first inequality (the second is similar),
	and for that we need to show:
	$\left(\sum_{i=1}^m H(SX_i) - (m-1)H(S) - H(\Omega)\right) +
	m\left(\sum_{j=1}^k H(SY_j) - (k-1)H(S) - H(\Omega)\right) \geq
	\sum_{ij} H(SZ_{ij}) - (mk-1) H(S) - H(\Omega)$, or, equivalently:
	\begin{align}
	\sum_{i=1}^m  H(SX_i) + m\sum_{j=1}^k H(SY_j) \geq \sum_{ij} H(SZ_{ij})+ m H(\Omega)
	\label{eq:singleMax1}
	\end{align}
	For that we prove by induction on $\ell$:
	\begin{align}
	H(SX_i) + \sum_{j=1}^\ell H(SY_j) \geq
	\sum_{j=1}^\ell H(SZ_{ij}) + H(SX_iY_1\ldots Y_\ell) \label{eq:singleMax2}
	\end{align}
	Indeed, assuming the statement for $\ell-1$ holds, then the statement
	for $\ell$ follows from:
	\begin{align*}
	H(SY_\ell) + H(SX_iY_1\ldots Y_{\ell-1}) \geq H(SZ_{i\ell}) + H(SX_iY_1\ldots Y_\ell)
	\end{align*}
	which is the submodularity inequality, since
	$SY_\ell \cap (SX_iY_1\ldots Y_{\ell-1}) = SY_\ell \cap SX_i =
	SZ_{i\ell}$.  Setting $\ell=k$ in \eqref{eq:singleMax2} and summing
	over $i=1,m$ we obtain
	$\sum_{i=1}^m H(SX_i) + m\sum_{j=1}^k H(SY_j) \geq \sum_{ij}
	H(SZ_{ij}) + \sum_{i=1}^m H(SX_iY_1\cdots Y_k) = \sum_{ij} H(SZ_{ij})
	+ m H(\Omega)$, proving \eqref{eq:singleMax1}.
\end{proof}

\section{Proofs and details from Section~\ref{sec:MiningMVDs}}
\subsection{Correctness of Algorithm $\algname{MineAllMinSeps}$}\label{sec:ComplexityAnalysisOfMineAllMinSeps}

\begin{reptheorem}{\ref{thm:AcyclicCharacterizationLee}}
	\completenessTheorem
\end{reptheorem}
\begin{proof}
	We first note that every set of attributes $S$ that is added to $\bS$ in lines~\ref{algline:addToS1} and~$\ref{algline:addToS2}$ is a minimal $AB$ separator. Therefore, we proceed by showing that \e{all} minimal $AB$ separators are mined by $\MVDAlg$.
	
	Let $\Omega=X_1\dots X_n$, and let $p=X_1,\dots, X_n$
	be some predefined order over the attributes that is used in algorithm $\algname{ReduceMinSep}$ (Figure~\ref{alg:ReduceToMinAPSep}). We view every minimal $AB$-separator as a subsequence of $p$, whose letters (i.e., attributes) are ordered according to $p$. That is, the permutation $p$ induces a lexicographic ordering over the subsets of $\Omega$. For example, $X_3X_4X_9X_{15} \succ X_3X_4X_7X_{100}$.
	We prove the claim by backwards induction on the lexicographic ordering of the subsets of $\Omega$. That is, for every subsequence $p_S$ of $p$, over attributes $S$, we show that if $S$ is a minimal $AB$ separator, then $S$ is discovered by the algorithm. The induction follows reverse lexicographic order of the sequences (e.g., $X_3X_4X_9X_{15}$ before $X_3X_4X_7X_{100}$ ). 
	\eat{
		We prove by backwards induction on $n$ that all minimal $AB$-separators $X=X_1,\dots,X_m$ are mined.
	}
	
	\paragraph{Base case}: $p_S$ is the lexicographically largest subsequence: $p_S=X_n$, or $S=X_n$.
	By Theorem~\ref{thm:transversals}, if $S$ is a minimal $AB$ separator that is not in $\bS$, then there exists a minimal transversal $D$ of $\bS$ such that $S \subseteq \comp{D}$. By Proposition~\ref{prop:downwardClosureGeneral} if $X_n$ separates $A$ and $B$, then so does each one of its supersets.
	Therefore, algorithm $\algname{ReduceMinSep}$ (Figure~\ref{alg:ReduceToMinAPSep}) that uses the attribute sequence $p$, will return the minimal $AB$ separator $S=X_n$ when provided with input $\comp{D} \supseteq \set{X_n} =S$.
	
	\paragraph{Step}: Let $p_S$ denote the subsequence corresponding to the set $S\subset \Omega$. By the induction hypothesis, we assume that all minimal $AB$ separators that are lexicographically larger than $S$ have been mined and are in $\bS$.
	By Theorem~\ref{thm:transversals}, there exists a  minimal transversal $D$ of $\bS$ such that $S \subseteq \comp{D}$.
	Now, let $p_S=X_{i_1},\dots,X_{i_m}$ denote the subsequence associated with $S$ (i.e., $S=\set{X_{i_1},\dots, X_{i_m}}$).
	Now, consider how algorithm $\algname{ReduceMinSep}$ handles the input $\comp{D}$ (line~\ref{algline:reduceMinSep}). Clearly, it will remove all attributes $X_j \in \comp{D}$ such that $X_{i_1} \succ X_j$ because the resulting set contains the minimal separator $S$ (line~\ref{algline:removeAtt} in $\algname{ReduceMinSep}$).
	\eat{	Let us assume that $X_1=\Omega_k$, then all attributes $\Omega_i$ where $i<k$ will be removed in line 6 of the algorithm $\algname{ReduceMinSep}$.}
	Now, suppose, by contradiction, that $X_{i_k}\in S$ is removed in line 6 of $\algname{ReduceMinSep}$. This means that $\comp{D}$ contains a minimal $AB$ separator $C$ that is lexicographically larger than $S$. But by the induction hypothesis, such a minimal separator $C$ is already in $\bS$. Since $C \subseteq \comp{D}$, it means that $C \cap D=\emptyset$, contradicting the fact that $D$ is a minimal transversal of $\bS$.
\end{proof}

\subsection{Runtime Analysis of \algname{MineAllMinSeps}}
\begin{definition}
	Let $\bS$ be a (not necessarily complete) set of minimal $AB$ separators. We define the \e{negative border of} $\bS$ to be:
	\begin{equation}
	\BD^{-}(\bS)=\set{U \subset \Omega | U \notin \bS, \text{ there exists a } X_i \in \Omega \text{ s.t. } U\cup \set{X_i}\in \bS}
	\end{equation}
\end{definition}
Since every minimal separator in $\bS$ contains at most $n$ attributes then $|BD^{-}(\bS)|\leq |\bS|\cdot n$.
\begin{theorem}\label{thm:numOfIterations}
	The number of minimal transversals processed in lines \ref{algline:startTransversalLoop}-\ref{algline:endWhile} of algorithm \algname{MineAllMinseps} is at most $|\BD^{-}(\bS)|$.
\end{theorem}
\begin{proof}
	Let $D$ be a minimal transversal of $\bS$ processed by in lines \ref{algline:startTransversalLoop}-\ref{algline:endWhile}. It  cannot be the case that $\comp{D}\supseteq C$ for any $C \in \bS$, and in particular $\comp{D}\notin \bS$. \eat{,  because that would mean that $C \cap D=\emptyset$, contradicting the fact that $D$ is a transversal of $\bS$.}
	%	Clearly, the loop iterates over all minimal transversals $X$ such that $\comp{X}$ is not a minimal $AB$ separator.
	Since $D$ is a minimal transversal, then for every attribute $Y \in D$ it holds that $D\sm \set{Y}$ is no longer a transversal for $\bS$. That is, there is an $AB$ minimal separator $C \in \bS$ such that $C \cap (D\sm \set{Y})=\emptyset$, or that $C \subseteq \comp{(D\sm \set{Y})}$. Noting that $\comp{(D\sm \set{Y})}=\comp{D}\cup\set{Y}$, we get that $C \subseteq \comp{D}\cup\set{Y}$, or that $C\sm \set{Y} \subseteq \comp{D}$. So we get that $C\sm \set{Y} \subseteq \comp{D}$, and that $C \not\subseteq \comp{D}$. In other words, every minimal transversal $D$ processed corresponds to a set in $\BD^{-}(\bS)$.
\end{proof}

\subsection{An Optimization to $\algname{getFullMVDs}$}
In the worst case, if $S$ is \e{not} an $AB$ separator then Algorithm~\algname{getFullMVDs} will traverse the complete search space of size\eat{$\sum_{k=2}^{n-2}{n-2 \choose k}=O(2^n)$} $O(2^n)$. While, in general, this is unavoidable, we implemented an optimization, described in the complete version of this paper, that leads to a significant reduction in the search space.

By~\eqref{eq:downward:1} in Proposition~\ref{prop:downwardClosureGeneral} it holds that if $I(A;B|S)>\varepsilon$
for a pair of attributes $A,B\in\Omega$, then for any MVD $\phi=S \mvd C_1|\dots|C_m$ in which $A$ and $B$ are in distinct components it holds that $\J_H(\phi) > \varepsilon$.

We say that an MVD $\phi = S \mvd
C_1|\dots|C_m$ is \e{pairwise consistent} if $I(C_i;C_j|S)\leq \varepsilon$ for every pair of distinct components $C_i,C_j \in \components(\phi)$.
Since $I(C_i;C_j|S) \leq \J(S \mvd C_1|\dots|C_m)$, then we can prune  
an MVD $S \mvd C_1|\dots|C_m$ if it is not pairwise consistent, and avoid traversing its neighbors and descendants. 
In Figure~\ref{alg:getPairwiseConsistentMVD} we present the algorithm $\algname{getPairwiseConsistentMVD}$ that receives an MVD $\phi=S \mvd C_1|\dots|C_m$ where $A$ and $B$ are in distinct components, and returns a \e{pairwise consistent} MVD where $A$ and $B$ are in distinct components, if one exists. %For an MVD $\phi=S \mvd C_1|\dots|C_m$ we denote the set $\set{C_1,\dots,C_m} $ by $\components(\phi)$.
In Figure~\ref{alg:OptimizedMineAllJDs} we present the optimized $\algname{getFullMVDs}$ that prunes MVDs that cannot lead (via merges to components) to an MVD in which $A$ and $B$ are in distinct components.

\begin{algserieswide}{H}{Return an MVD $S \mvd C_1|\dots|C_m$ s.t. $I(C_i;C_j|S)\leq \varepsilon$ for every pair $C_i,C_j$, and $A$ and $B$ are in distinct components or $nil$ if no such MVD exists. \label{alg:getPairwiseConsistentMVD}}
	\begin{insidealgwide}{getPairwiseConsistentMVD}{$\varepsilon$, $\phi$, $(A,B)$}
		\WHILE{$A$ and $B$ are in distinct components of $\phi$ AND $\phi$ is not pairwise consistent}
		\STATE Let $C_i,C_j\in \components(\phi)$ s.t. $I(C_i;C_j|S) > \varepsilon$
		\STATE $\phi \gets \mergeFunc_{ij}(\phi)$	
		\ENDWHILE
		\IF{ $A$ and $B$ are in distinct components of $\phi$}
		\STATE return $\phi$
		\ENDIF%\textsl{}
		\STATE return $nil$
	\end{insidealgwide}
	
\end{algserieswide}

\begin{algserieswide}{H}{Returns a set of at most $K$ full MVDs with key $S$ that approximately
		hold in $R$ in which $A$ and $B$ are in distinct components. \label{alg:OptimizedMineAllJDs}}
	\begin{insidealgwide}{getFullMVDsOpt}{$S$, 	$\varepsilon$, $(A,B)$, $K$}
		\STATE $\Pm \gets \emptyset$ \COMMENT{Output set}
		\STATE $\Q \gets \emptyset$ \COMMENT{$\Q$ is a stack}
		\STATE $\phi_0=S\mvd X_1|\dots|X_n$
		where $X_i$ are singletons.
		{ \color{blue}
			\STATE $\phi'_0 \gets \algname{getPairwiseConsistentMVD}(\varepsilon, \phi,(A,B))$
			\IF{$\phi'_0 = nil$}
			\STATE return $\emptyset$
			\ENDIF
		}
		\STATE $\Q.\push(\phi'_0)$
		\WHILE {$\Q \neq \emptyset$ \AND $|\Pm| < K$}
		\STATE $\varphi \gets \Q.\pop()$
		\STATE Computed $\J_H(\varphi)$  \COMMENT{using $\entropyAlg$}
		\IF{$\J_H(\varphi) \leq \varepsilon $}
		\STATE $\Pm \gets \Pm \cup \set{\varphi}$
		%	\IF{$|\Pm|=K$}
		%	\STATE break
		%	\ENDIF
		\ELSE				
		\FORALL{$\phi \in \Nbr(\varphi)$}	
		{ \color{blue}
			\STATE $\phi' \gets \algname{getPairwiseConsistentMVD}(\varepsilon, \phi,(A,B))$}					
		\IF{$\phi' \neq nil$}			
		\STATE $\Q.\push(\phi')$	\COMMENT{See~\eqref{eq:tracersalAlgNbrs}}
		\ENDIF
		\ENDFOR	
		\ENDIF
		\ENDWHILE		
		\STATE \textbf{return} $\Pm$
	\end{insidealgwide}
	
\end{algserieswide}

\section{Proofs from Section~\ref{sec:EnumerateAcyclicSchemas}}
\begin{reptheorem}{\ref{thm:equivalenceThm}}
	\CompatibleCompletenessThm
\end{reptheorem}
\begin{proof}
	Every key of $\MVD(\T)$ is the label on an edge of $\T$, and thus contained in a bag of $\T$. Hence, the set $\MVD(\T)$ is split-free and satisfies the first condition of definition~\ref{def:CompatibleSJDs}.
	
	Let $\phi_1,\phi_2 \in \MVD(\T)$ corresponding to edges $e_1,e_2\in \edges(\T)$. Let $\T_1$, $\T_2$, and $\T_3$ be the three connected subtrees resulting from removing $e_1,e_2$ from $\T$. W.l.o.g, any path from a node in $\nodes(\T_1)$ to a node in $\nodes(\T_3)$ must pass through a node in $\nodes(\T_2)$. Therefore, $\components(\phi_1)=\set{\jointreeMapFunction(\T_1)\sm \key(\phi_1), \jointreeMapFunction(\T_2){\cup}\jointreeMapFunction(\T_3)\sm \key(\phi_1)}$, and $\components(\phi_2)=\set{\jointreeMapFunction(\T_3)\sm \key(\phi_2), \jointreeMapFunction(\T_1){\cup}\jointreeMapFunction(\T_2)\sm \key(\phi_2)}$. In particular, $\phi_2$ splits the set $\jointreeMapFunction(\T_2)\cup\jointreeMapFunction(\T_3)$, and $\phi_1$ splits the set $\jointreeMapFunction(\T_2)\cup\jointreeMapFunction(\T_1)$. Hence, $\MVD(\T)$ satisfies the second condition of definition~\ref{def:CompatibleSJDs}.
\end{proof}

\section{Further Experiments}
\begin{figure*}[ht]
	\begin{tabular}{ccccc}
		\subfigure[\textsf{Classification}\label{chart:classificationFullMVDs}]{\includegraphics[width=0.25\textwidth]{classificationFullMVDs}}& 
		\subfigure[\textsf{BreastCancer}]{\includegraphics[width=0.25\textwidth]{breastCancerFullMVDs}\label{chart:breastCancerFullMVDs}}&
		\subfigure[\textsf{Adult}]{\includegraphics[width=0.25\textwidth]{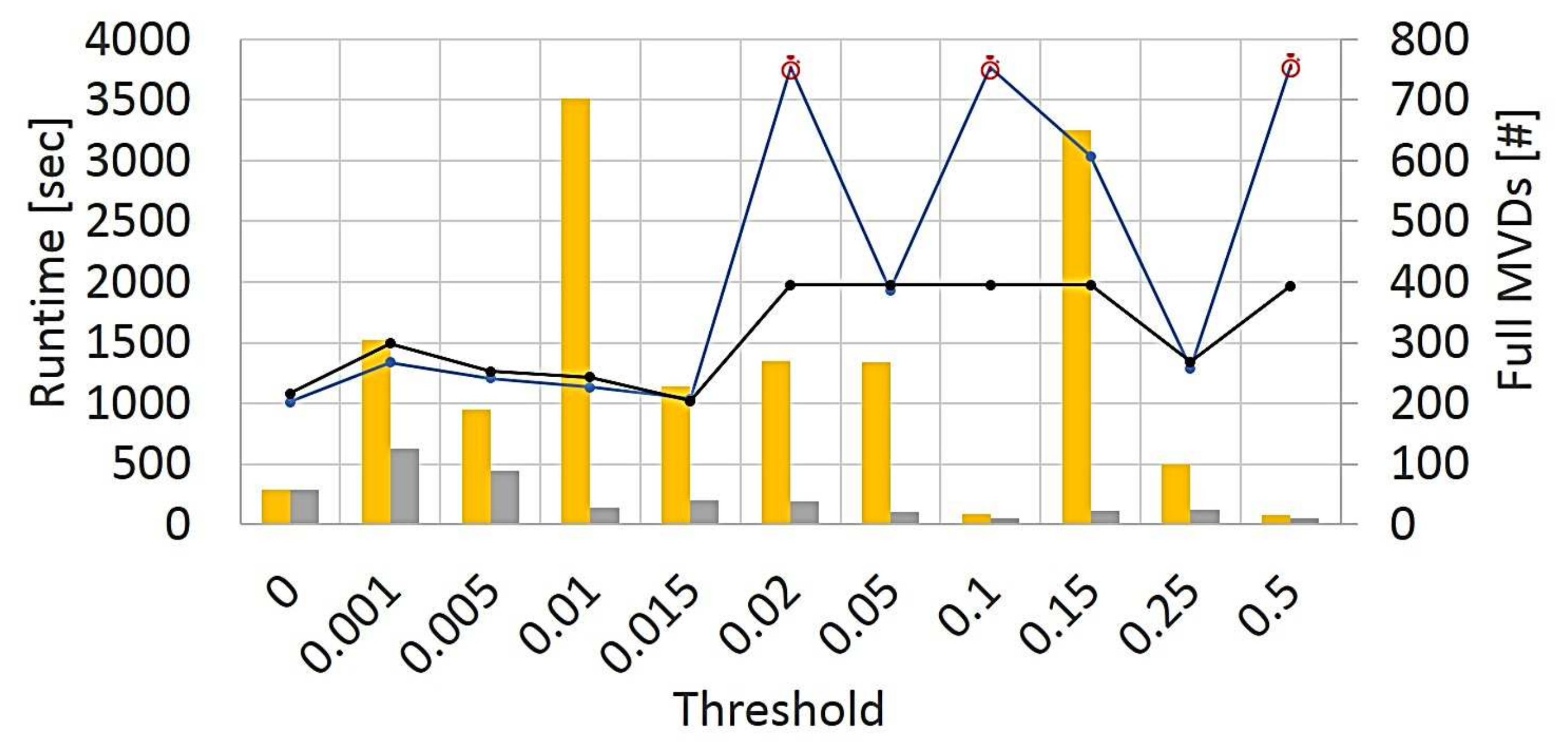}\label{chart:adultFullMVDs}}&
		\subfigure[\textsf{Bridges}]{\includegraphics[width=0.25\textwidth]{bridgesFullMVDs}\label{chart:bridgesFullMVDs}}
		\\		
		%\subfloat[]{\includegraphics[width=0.2\textwidth]{pics/horse}}
		\includegraphics[width=0.13\textwidth]{FullMVDsLegendMinSepNum}& 
		\includegraphics[width=0.13\textwidth]{FullMVDsLegendFullMVDsNum} &
		\includegraphics[width=0.13\textwidth]{FullMVDsLegendFullMVDsRT}&
		\includegraphics[width=0.13\textwidth]{FullMVDsLegendMinSepRT}&
	\end{tabular}
	\caption{Full MVDs Experiments. Red stopwatch indicates that the algorithm stopped after 30 minutes.}
	\label{tab:FullMVDsExperiments}	
\end{figure*}

\subsection{From minimal separators to full MVDs}\label{sec:FullMVDs}
We now experiment with the transition from minimal separators to full MVDs. 
We recall that an MVD $\phi$ is full with regard to $\varepsilon$ if
$\relation \models_\varepsilon \phi$ and for all MVDs $\psi \succ
\phi$ that strictly refine $\phi$ then $\relation
\not\models_\varepsilon \psi$. 

In this set of experiments we have, for every pair of attributes $A,B \in \Omega$, the set $\minsep_{\varepsilon,A,B}(\relation)$ of minimal $AB$-separators that hold in $\relation$ w.r.t. $\varepsilon$, and we apply the algorithm for generating the set $\fullMVDs_{\varepsilon,A,B}$ by calling $\algname{getFullMVDs}$ (Fig.~\ref{alg:NaiveMineAllJDs}) with the pair $(A,B)$, and an unlimited number of MVDs to return (i.e., $K=\infty$). 
\footnote{We actually call the optimized version of this algorithm, $\algname{getFullMVDsOpt}$ described in the full version of this paper.} In particular, the runtimes presented here do not include the time taken to mine the minimal separators. The performance of this phase is analyzed in Section~\ref{sec:scalability} and Table~\ref{tab:evaluationDatasets}.
% $\fullMVDs(\relation)=\cup_{S\in \minsep(\relation)}\fullMVDs(S)$.

\eat{Let $A,B \in \Omega$ be a pair of attributes.
	For every minimal $AB$ separator $S \in \minsep(\relation)$, we compute the set of $\fullMVDs(S)$ by calling the method $\algname{getFullMVDs}$ (Figure~\ref{alg:NaiveMineAllJDs}) with parameters $S$, $\varepsilon$, the pair $(A,B)$, and an unbounded $K$.}
We conduct the experiment as follows. For every dataset we vary the threshold in the range $[0,0.5]$, and for every threshold execute the procedure $\algname{getFullMVDsOpt}$ for a total of $30$ minutes. The results are presented in Figure~\ref{tab:FullMVDsExperiments}. When the threshold is $\varepsilon=0$ then the number of full MVDs is identical to the number of minimal separators as expected by Lemma~\ref{lem:singleMaxMVDLemma}.
In practice, when the threshold is
$0$, our algorithm for mining all minimal separators also
discovers all full MVDs. 
As the threshold increases so does the difference between the number of minimal separators and the number of full MVDs. Overall, Algorithm $\algname{getFullMVDsOpt}$ for generating full MVDs is capable of reaching a rate of about 55 full MVDs per second for thresholds larger than $0.1$ (see Figures~\ref{tab:FullMVDsExperiments}(a),~\ref{tab:FullMVDsExperiments}(b), and~\ref{tab:FullMVDsExperiments}(d)).

\end{document}